\documentclass[11pt,a4paper]{article}

\usepackage[margin=1in]{geometry}
\usepackage[T1]{fontenc}
\usepackage[utf8]{inputenc}
\usepackage{lmodern}
\usepackage{microtype}

\usepackage{amsmath,amssymb,amsthm,mathtools}
\usepackage{booktabs,longtable,array,multirow}
\usepackage{graphicx}
\usepackage{xcolor}
\usepackage[hyphens]{url}
\usepackage[round,authoryear]{natbib}
\usepackage[unicode,hidelinks]{hyperref}
\usepackage[capitalise,noabbrev]{cleveref}
\usepackage{caption}
\usepackage{enumitem}
\usepackage{amssymb}   

\usepackage{tikz}
\usetikzlibrary{arrows.meta,positioning,fit,backgrounds,calc,shapes.geometric,decorations.pathmorphing}
\usepackage{pgfplots}
\pgfplotsset{compat=1.18}

\theoremstyle{definition}
\newtheorem{definition}{Definition}[section]
\newtheorem{assumption}{Assumption}[section]
\newtheorem{example}{Example}[section]
\newtheorem{remark}{Remark}[section]

\theoremstyle{plain}
\newtheorem{theorem}{Theorem}[section]
\newtheorem{lemma}{Lemma}[section]
\newtheorem{proposition}{Proposition}[section]
\newtheorem{corollary}{Corollary}[section]
\crefname{definition}{Definition}{Definitions}
\Crefname{definition}{Definition}{Definitions}
\crefname{theorem}{Theorem}{Theorems}
\Crefname{theorem}{Theorem}{Theorems}
\crefname{lemma}{Lemma}{Lemmas}
\Crefname{lemma}{Lemma}{Lemmas}
\crefname{proposition}{Proposition}{Propositions}
\Crefname{proposition}{Proposition}{Propositions}
\crefname{corollary}{Corollary}{Corollaries}
\Crefname{corollary}{Corollary}{Corollaries}
\crefname{assumption}{Assumption}{Assumptions}
\Crefname{assumption}{Assumption}{Assumptions}
\crefname{example}{Example}{Examples}
\Crefname{example}{Example}{Examples}
\crefname{remark}{Remark}{Remarks}
\Crefname{remark}{Remark}{Remarks}

\newcommand{\N}{\mathbb{N}}

\newcommand{\Prob}{\mathbb{P}}
\newcommand{\Ind}[1]{\mathbf{1}\!\left[#1\right]}

\newcommand{\Loop}{\ell}                      
\newcommand{\Worker}{W}                       
\newcommand{\Gate}{G}                         
\newcommand{\Budget}{B}                       
\newcommand{\Graph}{\mathcal{G}}              
\newcommand{\Nodes}{V}
\newcommand{\Edges}{E}
\newcommand{\Art}{a}                          
\newcommand{\Verdict}{\nu}                    
\newcommand{\Pass}{\textsc{pass}}
\newcommand{\Reject}{\textsc{reject}}
\newcommand{\Incap}{\textsc{incapacity}}
\newcommand{\Done}{\textsc{done}}
\newcommand{\Halt}{\textsc{halt}}
\newcommand{\maxatt}{\mathrm{max\_attempts}}
\newcommand{\rounds}{R}                       
\newcommand{\alphafa}{\alpha}                 
\newcommand{\betafr}{\beta}                   
\newcommand{\logmean}{\bar{\mu}_n}            

\title{\bfseries Bounded Loops:\\[3pt]
  {\normalfont\large Pre-Run Spend Bounds, Proved Termination,\\
   and Verified Completion for Agent Harnesses}}
\author{%
  Varun Pratap Bhardwaj\\
  {\small Qualixar / Independent Researcher, India}\\
  {\small\texttt{varun.pratap.bhardwaj@gmail.com}}\\[1pt]
  {\small ORCID: 0009-0002-8726-4289}
  \and
  Garima Singh\\
  {\small Independent Researcher, India}\\
  {\small\texttt{garima1213@gmail.com}}
  \and
  Arun Pratap Bhardwaj\\
  {\small Independent Researcher, India}\\
  {\small\texttt{arun.pratap.bhardwaj@gmail.com}}
}
\date{}

\begin{document}
\maketitle

\begin{abstract}
\noindent
In the agent frameworks that dominate current deployments, a step ends when the agent's own
output says it has finished. Durable-execution platforms do better --- they bound retries and
elapsed time, and their completion conditions live in workflow code --- but the checker is
conventionally a function in the same codebase as the work, which is a discipline the deployment
is trusted to keep rather than a property the harness enforces.

We state what an \emph{agent harness}\footnote{We use the term as a category label for this class
of system. Microsoft's Agent Framework ships a component of the same name
(\cref{sec:ms-harness}); no relation is implied.} must guarantee, prove it, and build the
instrument that measures whether a given harness delivers it. A \emph{bounded loop} is a worker,
an \emph{independent gate} the worker cannot write to, and a declared budget; a
\emph{bounded-loop graph} composes them with a repair relation that lets a downstream failure
re-run an already-finished upstream node. Three guarantees follow, each readable before a run is
authorised. \textbf{It finishes}: termination holds under repair, with the worst-case attempt
total in closed form from the manifest, on the condition that the repair budget is global rather
than per node. \textbf{It does not drift}: no node reaches \textsc{done} without a gate verdict
in an append-only hash-chained ledger, and the worker's own claim appears in no guard of the
procedure --- proved from control flow, since repair leaves no topological order to induct along.
A third party holding only the run directory can check this, and the released verifier does.
\textbf{It does not overspend}: the declared ceiling is enforced inside an attempt, not merely
between attempts, with a reserve partitioned out of it so a hard stop leaves a written handoff.

A proof about gates is worth what the gates are worth, so we measure them with a two-tier
held-out mutant corpus whose blindness is enforced as a static property of the generator's syntax
tree, and we characterise two defect classes that let a sound-looking check pass anything:
\emph{vacuity} --- a check satisfied by the absence of the thing it checks, which admits a
worker that deletes the artifact and completes in one attempt --- and
\emph{self-attestation}, a check whose criterion is exacting and whose subject supplied the
value it is applied to. Applied to
a 69-loop catalogue, the instrument found 47 such gates in code that had shipped and passed
review. Against the repaired gates it reports no false accepts over 209 destroying mutants
($\alpha \le 1.8\%$, Wilson 95\%) --- and we show that figure is saturation, not quality, by
freezing the gates and applying a fresh operator family, which recovers a false-accept rate of
23.3\% where the exhausted corpus reported none. A rate belongs to a specific gate; the
apparatus, not our number, is the contribution.

The same measurement turned on this paper found a bound declared in a readable place and enforced
nowhere, six times, once in the hypothesis of one of our own propositions. A later release audit,
by agent auditors who had not written the code, found three self-attested checks in subsystems
sharing none: a rule about a verdict's reason that asked the verdict, a test about a platform's
install policy that asked a stub, and a record of which gate ran that read the answer off the gate. Engine, catalogue,
corpus and analysis scripts are released under Apache-2.0.
\end{abstract}

\section{Introduction}
\label{sec:intro}

An agent is asked to do something. It does something. It reports that it is
finished. In the agent frameworks that dominate current deployments, that report
\emph{is} the outcome: the orchestrator marks the step complete, downstream steps begin,
and every guarantee the system offers rests on a self-assessment by the party least able
to be objective about it. This is not a straw man: a harness shipped in August 2026
documents exactly that about its own loop, and defers the missing component by name
(\cref{sec:dsh}).

The gap is narrower than it is usually stated. Durable-execution platforms and container
schedulers do bound retries and elapsed time, and their completion conditions live in
workflow code rather than in a model's output. What they do not provide is a checker whose
independence from the work is a property the harness \emph{enforces}. The common production
pattern --- a validation step in the same repository as the work --- is a convention
(\cref{rem:separate-function}).

Substantial machinery has been built around that moment without changing it.
Graph-structured orchestration made control flow explicit and inspectable
\citep{langgraph,gao2024agentscope}; durable-execution backends made it survive
process death \citep{temporal}; and a recent scheduler-theoretic treatment
\citep{sgh2026} reports that 42 of 70 surveyed open-source agent projects still
use a variant of the unstructured agent loop, which is the measure of how
recently this became a settled question. This paper takes all of it as given and
builds on top. What none of it settles is the question in the first paragraph:
\emph{what licenses the claim that a step is done?}

Answering it takes two things that do not currently exist. A \textbf{contract}
--- a statement of what a harness guarantees, in a form that can be proved rather
than intended. And an \textbf{instrument} --- a way to measure whether a
particular harness delivers it, since a verification layer is precisely the
component whose failures are invisible. When a verifier wrongly accepts, the
system reports success, the receipt looks correct, and nothing complains. Every
incentive points at not measuring it.

The nearest work to ours has already begun on the first of these. The Graph
Harness proposal \citep{sgh2026} states termination and conditional soundness
for a scheduler over a static DAG as theorems, which is the right form for the
question. It draws two boundaries, explicitly and for stated reasons: it
excludes the construct that lets a downstream failure re-run an already-finished
upstream node, because that construct breaks the premise its termination proof
rests on; and its soundness bound is expressed in a per-node validation accuracy
that it takes as given, since no measured value of that quantity exists for a
hand-written checker. This paper crosses the first boundary and pays a stated
price for the crossing, and builds an instrument that produces the second
quantity rather than assuming it. We are extending a line of work, not
correcting it (\cref{sec:related-sgh}).

\subsection{Contributions}

The three guarantees above are what an operator gets. What follows is how they are
established and measured: two of the eight items define the objects, three prove
the guarantees, and three build and validate the instrument that says what a
proof about gates is worth.

\paragraph{1. The bounded loop, and a graph of them.}
A \emph{bounded loop} is a triple $(\Worker, \Gate, \Budget)$: a worker that
attempts, an \emph{independent} gate that decides, and a declared budget that
bounds the attempts. A \emph{bounded-loop graph} composes them into a DAG with
\emph{repair edges} --- the one construct that lets a failed node send work back
upstream. This is the unit the rest of the paper reasons about, and it is small
enough that every property below is a statement about an object a reader can
inspect (\cref{sec:prelim}).

\paragraph{2. A verdict boundary that is derived rather than guessed.}
Every harness that checks its own work faces one question at the worst possible
moment: when a check fails, is the \emph{artifact} wrong or is the \emph{checker}
broken? The two cases can be byte-for-byte identical --- a file that will not
parse --- and demand opposite responses. We show the distinguishing fact is
\emph{ownership}; that a loop's anti-tamper declaration, written for unrelated
reasons before the run, already encodes it; and that the resulting three-valued
classification is total, disjoint, and decidable without consulting the worker,
its transcript, or the run's history (\cref{thm:trichotomy}). The corollary is
the design rule: \emph{``there was nothing to check'' is an answer about the
work, never a confession by the harness} (\cref{cor:vacuity-rejects}).

\paragraph{3. Termination under repair, and the three conditions that buy it.}
A single loop terminates by induction on its attempt counter. A graph does not
follow: a repair edge re-admits a completed sub-DAG and \emph{restores} the
attempt budgets, resetting exactly the quantity the inner induction relies on
decreasing. Restoring spent budget is a reset arc, and for nets with reset arcs
the soundness question is undecidable, so no general procedure is on offer and
what remains is a sufficient condition. We give three
(\cref{ass:bounded-repair}), of which the load-bearing one is that the repair
budget is \emph{global} rather than per node --- bound it per node, the
arrangement every other budget in the system uses, and two nodes can repair each
other for ever with both counters showing credit. Under those conditions the
measure that works is lexicographic on
$(\rounds - \rho,\ \text{remaining budget})$: the second component may rise,
provided the first strictly falls (\cref{thm:graph-termination}). The worst-case
attempt count is then a closed expression in quantities declared before the run,
so an operator reads it off the manifest rather than off a bill
(\cref{cor:manifest-readable}), and the cost of enabling repair is visible as a
single factor of $(\rounds+1)$ against the per-node total that bounds a graph
without it.

\paragraph{4. Soundness of \textsc{done}.}
Directly from the procedure's control flow: a graph reporting \textsc{done} has a
gate verdict in its append-only ledger for every node (\cref{thm:soundness}). The
proof's substance is what it never touches, so there is no execution path on which
an agent's claim about itself influences the outcome. The result is checkable by a
third party holding only the run directory (\cref{cor:third-party}), which is what
distinguishes a guarantee from an assurance.

\paragraph{5. An instrument for measuring a verifier.}
Testing a gate against defects written by the gate's author measures the author's
internal consistency and nothing else. We give a two-tier held-out corpus.
\textbf{Tier 1} applies mechanical operators whose blindness is a \emph{static
property of the code}: an assertion over the syntax tree of every generator
module fails the build if any of them could consult a gate (\cref{thm:blind}).
\textbf{Tier 2} covers what mechanical operators structurally cannot express ---
requirements that are a relation \emph{between} artifacts --- using semantic
mutants authored by five independent agent models shown only a loop's stated
purpose, each pinned to a digest of exactly what its author saw. The two tiers
are complementary capabilities, not a strong method and a weak one.

\paragraph{6. Ground truth by construction.}
Every label is fixed by the mutation \emph{operation}, before the edit, and never
by observing a gate. The denominator is therefore not a function of gate
behaviour (\cref{thm:no-equivalent}), so the equivalent-mutant problem does not
arise --- and we show that the tempting automation of the standard remedy,
treating ``the system under test did not notice'' as equivalence, provably drives
the measured rate to zero for any loop whose gate is a test suite
(\cref{cor:exclusion-fatal}).

\paragraph{7. How to admit a model judge without weakening the guarantee.}
Every gate in our catalogue is a deterministic predicate, and
\cref{tab:decidable} shows that requirements which read as judgement calls
usually reduce to one. Where no reduction exists a model judge is the remaining
option, and the composition question has a one-line answer that holds with
\emph{no} assumption about the judge: under conjunction with a deterministic
gate, false accepts cannot rise and only false rejects can
(\cref{thm:judge-asymmetry}). Three things follow --- a judge cannot certify
broken work, the cost it does impose lands in the direction we already measure
and publish, and prompt injection against it gains nothing, since persuading the
judge does not move the deterministic gate. The engine implements conjunction
and refuses every other composition, so the theorem's hypothesis is guaranteed
by construction rather than left to configuration.

\paragraph{8. Vacuity, and what the instrument found.}
A check satisfied by the \emph{absence} of the thing it checks certifies a worker
that deletes the artifact (\cref{thm:vacuity}). We take the name from model
checking, where the phenomenon has been studied since \citet{beer2001vacuity},
and supply what that literature cannot: a rate in deployed agent code. The
sharpest instance is a gate written to catch invented legal citations that
passed an invented citation, because it built its search pattern from the same
reference set it validates against --- so a fabricated case in an unknown
reporter was never examined at all (\cref{ex:citations}). Applied to a 57-loop
corpus drawn from the shipped catalogue, the instrument located \textbf{47} such
gates in code that had shipped and passed review (\cref{sec:eval}). The
\emph{instrument} is the contribution: that number describes our gates, and the
apparatus is released so a reader can obtain their own.

\paragraph{9. A second defect class, with a different remedy.}
Vacuity is a check satisfied by the \emph{absence} of its subject. A check can also
be satisfied because the subject \emph{supplied the answer}: the artifact is
present, the criterion is exacting, and the value it is applied to was chosen by the
party under judgement. We define this as self-attestation (\cref{def:self-attested})
and give the consequence (\cref{prop:self-attestation}). It is worth separating from
vacuity because the remedies do not overlap --- an existence obligation fixes nothing
when nothing is missing, and the fix is instead to consult a value the subject cannot
control. Three instances were found in subsystems sharing no code, and \textbf{no
operator in either mutant tier of \cref{sec:measuring} can express one}, because both
tiers mutate worker-owned artifacts while this class lives in the machinery that reads
them (\cref{sec:e8}). That is a stated limit on the instrument this paper builds, not
an addition to its yield.

\paragraph{Scope, stated once.} Two properties are proved --- termination and
soundness of \textsc{done}; the engine as a whole is not verified in any sense a
formal-methods reader would accept, and the rest is tested. Independence here is
\emph{structural} --- worker and gate are distinct objects with disjoint write
authority --- not statistical, and nothing here shows their errors are
uncorrelated, which matters increasingly as both are authored by models
(\cref{sec:future}). Soundness is relative to the gates: it says a verdict exists,
not that the verdict was deserved, and how often it is undeserved is $\alphafa$,
which \cref{sec:eval} reports with the caveat that the same corpus drove the
repairs it then measures. Finally, there is \textbf{no benchmark against another
agent framework}: none has been run under this measurement, and a table with one
populated column would imply a comparison that did not happen.

\subsection{Roadmap}

\Cref{sec:related} places the work. \Cref{sec:prelim} fixes the objects.
\Cref{sec:verdict,sec:termination,sec:soundness} are the formal core, and
\cref{sec:vacuity,sec:self-attestation} are the two defect classes that let a
sound-looking check pass anything. \Cref{sec:measuring} builds the instrument and
\cref{sec:eval} reports what it found, including one class it cannot reach
(\cref{sec:e8}). \Cref{sec:discussion,sec:novelty,sec:future,sec:conclusion}
discuss, position and close. Deferred proofs are in \cref{app:proofs}, the
catalogue in \cref{app:catalogue}, \cref{app:repro} gives the commands that
regenerate every number, and \cref{app:audit} records the audit protocol behind
\cref{sec:e8}.

\section{Related Work}
\label{sec:related}

This work takes the orchestration and durable-execution stack as given. Its
contribution is a property that stack does not specify: who is entitled to declare a
step done, and what the harness owes when that declaration is wrong.

\subsection{Orchestration: graphs and durable execution}

The dominant pattern for multi-agent systems is now an explicit graph. Nodes are
agent invocations, tool calls or human checkpoints; edges are control flow;
execution is a traversal. LangGraph \citep{langgraph} established this as the
mainstream formulation; Microsoft's Agent Framework consolidated two earlier
efforts, of which \textsc{AutoGen} \citep{wu2024autogen} is the more widely
described, into a single graph-oriented product; and AgentScope
\citep{gao2024agentscope} and comparable systems occupy adjacent positions with
different opinions about actor semantics and event flow. In parallel,
durable-execution backends (Temporal \citep{temporal} most prominently) supply the
property that a long-running graph survives process death, with retries,
persistence and human-in-the-loop suspension handled below the framework. The
research line on agents-as-graphs is broader still: \citet{zhuge2024gptswarm} treat
the agent graph itself as an optimisable object, and multi-agent development
frameworks \citep{qian2024chatdev,hong2024metagpt} organise roles into fixed graph
topologies; \citet{guo2024survey} survey the space.

The engine described here compiles to a DAG, and the properties proved in
\cref{sec:termination,sec:soundness} are ones you would want of any such traversal.
Graphs are not the wrong abstraction. A graph makes explicit \emph{where} a step
ends. It does not say who is entitled to declare that it ended well.

\subsection{The nearest formal treatment, and the boundary it draws}
\label{sec:related-sgh}

One recent framework states the guarantees we are after as theorems. The Graph
Harness proposal \citep{sgh2026} recasts the agent loop as a single-ready-unit
scheduler over an explicit static DAG and proves two properties of it: bounded
termination, and a conditional soundness bound on the correctness of a completed
run. It is, by its own scope statement, a theoretical framework and design proposal
--- it reports no implementation and no empirical results, and names experimental
validation as future work. It also supplies a useful measurement of the practice it
is reacting to: of 70 surveyed open-source agent projects, 42 use a variant of the
unstructured agent loop.

Its termination result assumes what it calls \emph{terminal stability}: once a node
reaches a terminal state, no subsequent event changes it, so the state machine has
no outgoing transitions from any terminal state. The proof is then a
transition-counting argument, and the bound is $\sum_{v} \tau_v (b_v + 1)$ over
per-node timeouts $\tau_v$ and retry budgets $b_v$.

That bound is correct on its own premise, and its per-round total is the same
quantity as ours: a retry budget $b_v$ and a total-attempt budget $m_v$ are related
by $m_v = b_v + 1$, so $\sum_v (b_v+1)$ and $\sum_v m_v$ are the same sum written
two ways. The framework is not mistaken about anything it claims.

What is true is narrower, and it is the gap this paper occupies. Terminal stability
is load-bearing in that proof, and a \emph{repair edge} (\cref{def:repair}) --- a
downstream failure re-executing an upstream node that has already terminated --- is
precisely the construct that removes it. The proposal excludes that construct
deliberately and says so three times: in its design commitments, in its rationale
for restricting the ready set, and in its limitations table, which records
parent-chain rollback as violating a plan invariant and therefore unable to express
deep hierarchical recovery. Its own recovery ladder escalates from bounded retry to
local patch to full replan, and replanning issues a new immutable plan version
rather than restoring a budget in place. So the excluded case is not an oversight;
it is a stated boundary, drawn for a reason.

This paper crosses that boundary and pays for the crossing.
\Cref{thm:graph-termination} covers the case where an already-terminal node is
re-executed, and the price is an explicit factor: the total is not the per-node sum
but $(\rounds+1)$ times it. \Cref{sec:termination} gives the three conditions under
which the crossing is safe, and notes why the general version of the question is
undecidable, which is what makes a restricted answer worth stating at all.

The same framework's soundness result bounds the probability that every output is
correct by $\prod_v p_v$, where $1 - p_v$ is the probability that node $v$'s
validation wrongly accepts, assuming those validation errors are independent across
nodes. Two features of that statement motivate what we build. First, the retry
budget granted by the termination theorem does not appear in it: under $b_v+1$
attempts against the same validator the per-node false-accept probability is not
$\alphafa_v$ but compounds across attempts, which is the quantity \cref{sec:eval}
measures rather than assumes. Second, $p_v$ enters as a known constant, and nothing
in the literature reports a measured value of it for a hand-written gate. Supplying
an instrument that produces one, rather than a better bound in terms of a number
nobody has, is the contribution of \cref{sec:measuring}.

\subsection{Harness engineering}

The scaffolding around a model --- context assembly, tool execution, the loop,
verification, failure handling --- has recently been named and studied as a unit, in
a literature that is still largely preprints and practitioner writing rather than
archival publication. The emerging decomposition includes a verification component,
and the prevailing advice is that verification should be \emph{deterministic}:
linters, type checkers, parsers, test suites, rather than a model asked whether the
work looks right. Contract-driven architectures push further, running implementer
and verifier as separate agents with no shared conversation history, so that the
verifier cannot inherit the implementer's assumptions.

Our contract in \cref{sec:verdict} is a refinement within this tradition rather than
an alternative to it, and we say so explicitly in \cref{rem:industry-convention}:
returning a tool's error to the model instead of raising past it is established
guidance; separating transient from permanent faults is standard in cloud SDKs; and
\texttt{terraform plan}'s three exit codes are a decade-old demonstration that a
two-valued channel cannot carry ``succeeded with changes'' and ``failed'' at once.
What we add is that the boundary between ``the worker's fault'' and ``the checker is
broken'' need not be guessed at failure time, because a loop's anti-tamper
declaration already encodes it (\cref{thm:trichotomy}).

The gap we address is that this literature states the independence requirement and
does not measure whether it is met. A verifier is a classifier, and no published
treatment we are aware of reports its false-accept rate against defects it did not
author.

\subsubsection{The nearest shipping neighbour, named precisely}
\label{sec:ms-harness}

Since July 2026 Microsoft's Agent Framework ships a component called the
\emph{Harness}, documented in exactly the terms we use: ``the runtime scaffolding that
turns a language model into an agent''.\footnote{Microsoft Learn, \emph{Agent Harness},
\url{https://learn.microsoft.com/en-us/agent-framework/concepts/harness}, page dated
2026-07-29 and revised 2026-08-12. We use ``agent harness'' throughout as a category
term for the class of systems, which is how that page also defines it; where a
distinction matters we write ``the Agent Framework Harness'' for the product.} Its
capability matrix includes ``optional bounded re-invocation driven by evaluators or
predicates'' and a per-request iteration limit on function invocation, so the nearest
shipping system to this paper is not a research prototype and it does bound a loop.

Three things separate it from what we prove, none of them a criticism of a well-built
product solving a different problem:

\begin{itemize}[leftmargin=1.4em]
  \item \textbf{An evaluator is not an independent gate in our sense.} The looping
        capability is driven by evaluators and predicates composed into the same runtime
        as the agent. \Cref{def:independence} requires the checker to be a distinct
        object with \emph{disjoint write authority} over the artifact, which is what makes
        \cref{thm:soundness} say anything; a predicate the agent's own runtime hosts may
        be perfectly good engineering and still not support that theorem.
  \item \textbf{An iteration limit is not a manifest-readable closed form.} A
        per-request cap bounds one request. \Cref{cor:manifest-readable} gives the
        worst-case total for a whole graph, \emph{including repair rounds}, in terms whose
        every symbol is a literal in one file readable before the run is authorised.
  \item \textbf{Neither an iteration cap nor an evaluator is a spend contract.} We could
        find no documented pre-emptive token or wall-clock ceiling in that stack, and none
        of the three capabilities is presented as one.
\end{itemize}

The documentation additionally marks looping as experimental at the time of writing,
which dates the comparison: this is a fast-moving surface and a reader checking it
later should expect it to have moved.

\subsubsection{A shipping harness that states this paper's premise about itself}
\label{sec:dsh}

The strongest available statement of the premise above is not ours. In August 2026
DeepSeek published an open-source agent harness on a plugin architecture
\citep{dsh2026}. Its loop-until-done primitive drives one immutable objective
through a sequence of fresh workers, bounded by a round cap, and its documentation
records that completion is ``worker self-declaration'', that no independent
evaluator or verifier decides whether the objective is actually complete, and that
such an evaluator is deferred. The limitation is recorded in four separate package
documents; the guidance shown to the model on every request states that completion
and blockers are worker reports rather than independent evaluation; and a test
asserts that wording. One of the four reserves the vacancy explicitly, deferring
evaluator-backed certification ``to a separate policy layer''. So the component
\cref{def:independence} specifies is named as absent by a system that lacks it, and
named in the same terms.

That system does enforce something about completion, and the distinction matters
more than the concession. Its goal tools refuse a completion mutation unless the
call carries either host-attested human input in the current root-agent turn or the
exact admitted round of the current goal revision. That is a check on the
\emph{provenance of the declaration} --- which party is entitled to say done, and
under what authority. \Cref{def:independence} asks the other question, whether the
artifact is done, and answers it with a checker holding disjoint write authority
over that artifact. Both are worth having and neither substitutes for the other: an
authenticated self-report is still a self-report, and this is the sharpest available
illustration of why \cref{thm:soundness} needs the second property rather than the
first.

Two further boundaries are documented there and are the ones this paper crosses. The
round cap is the only aggregate limit, with token, price and elapsed-time budgets
recorded as deferred, so there is no pre-run spend ceiling of the kind
\cref{sec:spend-contract} enforces. And its session log is append-only as an
interface contract rather than a hash chain, which supports replay but not the
third-party tamper-evidence of \cref{thm:chain}. As with the comparison above, this
one is dated: the project calls itself a developer preview and announces
compatibility-breaking changes, and a reader checking later should expect movement.
The concessions quoted here are recorded, tested, and surfaced to the model, which
is a better engineering posture than an unstated assumption and is why they can be
cited as evidence at all.

\subsection{The same defect, measured in another domain}
\label{sec:embodied-termination}

The closest independent statement of this paper's central premise comes from embodied
agents rather than software ones. \citet{chen2026done} separate \emph{world completion}
--- whether the environment actually reached the goal state --- from
\emph{self-termination}, the agent's own decision that it is finished, and report that
systems with comparable world completion differ by up to $19.7$ percentage points in
measured success. The gap is not capability. It is agents that achieved the goal and
failed to recognise or report it, and agents that reported completion without achieving
it.

This counts for more than a result in our own domain would, because their setting shares
no infrastructure with ours: different task class, different action space, different
evaluation harness. That the same separation is load-bearing in both is evidence the
defect belongs to the \emph{self-report}, not to any particular toolchain.

The contributions are complementary rather than competing. They \emph{measure} the
divergence and give an evaluation framework for it; we \emph{remove the dependency} by
routing the completion decision to a check the worker cannot write to, and prove what
that buys (\cref{thm:soundness}). Neither result implies the other: an architecture that
never consults a self-report still needs its checker's false-accept rate measured, which
is \cref{sec:measuring}, and a measurement of self-termination error says nothing about
what a harness should do instead.

\subsection{Self-correction, and why an external gate}

A parallel line of work has agents critique and revise their own output ---
\textsc{ReAct}-style interleaving of reasoning and action \citep{yao2023react},
\textsc{Reflexion}'s verbal reinforcement from self-generated feedback
\citep{shinn2023reflexion}, and \textsc{Self-Refine}'s iterative self-critique
\citep{madaan2023selfrefine}. These improve outputs on many tasks and are
complementary to what we describe: nothing prevents a bounded loop's worker from
self-refining internally.

The relevant caution is that intrinsic self-correction, without external feedback, has
been found not to reliably improve reasoning and can degrade it
\citep{huang2024selfcorrect}. That result is the empirical case for
\cref{def:independence}: the signal has to come from somewhere the worker cannot
reach. Our contribution is not the observation that external signal helps, but a
harness in which the external signal is the \emph{only} thing that can terminate a
step, proved so rather than arranged so.

\subsection{Mutation testing}

Our evaluation method is mutation testing \citep{demillo1978hints}, whose modern
practice is surveyed by \citet{jia2011analysis}: seed faults mechanically, and measure
what the checking apparatus catches. The classical difficulty is the equivalent mutant
--- syntactically changed, semantically identical, impossible to kill --- identified by
\citet{budd1982two} and undecidable in general.

The usual remedy is to detect equivalent mutants, by inspection against the
\emph{specification} and sometimes assisted by static analysis
\citep{offutt1997equivalent}, and exclude them from the denominator.
\Cref{thm:no-equivalent} shows that with labels fixed by the mutation operation rather
than by observed behaviour, no such exclusion step is needed to avoid selection bias.
\Cref{cor:exclusion-fatal} then makes a narrower point that is easy to get wrong here:
if equivalence were operationalised as ``the system under test did not notice'' ---
which is \emph{not} the standard remedy, but is the tempting automation of it --- then
for a loop whose gate is a test suite the discarded set is exactly the set of false
accepts, and the measured rate is zero regardless of gate quality. The subject under
test being a verifier, rather than a program with a separate suite, is what makes that
shortcut degenerate rather than merely imprecise.

What our construction does \emph{not} remove is the obligation to check labels against
the specification, and we discharge it only partially: we report one mutant we had
labelled defective where the gate was right (\cref{sec:eval-corrections}), and we have
not conducted a sampled spec review of the Tier-1 labels. On the broader question of
whether mechanical mutants stand in for real faults, the evidence is mixed but broadly
supportive \citep{andrews2005mutation,just2014mutants};
\citet{papadakis2019mutation} survey the modern position.

Our setting also differs in what is mutated. Classical mutation testing perturbs a
program and asks whether tests notice. Here the mutated object is frequently not a
program at all --- a contract clause, a commit-message history, a dependency manifest,
a clinical note --- and the ``test suite'' is a purpose-built mechanical checker. The
operators are correspondingly domain-shaped, and \cref{def:tiers} adds a second tier
precisely because mechanical operators cannot express a requirement that is a relation
\emph{between} artifacts.

\subsection{Vacuity, which is not our term}
\label{sec:related-vacuity}

The property we call vacuity (\cref{def:vacuous}) has a precise and much older meaning
in formal verification, and we adopt the existing name deliberately rather than coining
one. \citet{beer2001vacuity} identified that a temporal formula can pass on a model
\emph{for the wrong reason} --- because an antecedent is never satisfied, so the
interesting part of the specification is never exercised --- and gave an efficient
detection procedure. \citet{kupferman2003vacuity} formalised it as a formula holding
while some subformula does not affect the outcome, \citet{armoni2003enhanced}
strengthened detection for LTL, and \citet{gurfinkel2004vacuity} extended the notion
beyond single-occurrence subformulas. Vacuity checking has since become a standard
post-processing step in commercial model checkers, on the practical ground that a
passing property nobody has confirmed is \emph{meaningful} is worth very little.

The concept is theirs. What differs here is the setting: in model checking, vacuity is a
property of a specification against a model, detected by a tool built for that purpose.
In an agent harness the ``specification'' is a hand-written gate over an artifact, there
is no model checker, and the analogous check does not exist --- so the class recurs
unnoticed. What we add to it is a measured rate and a statement about an adversary
rather than about a specification; \cref{sec:novelty} says precisely what that is
relative to.

This also suggests a direction of transfer. Vacuity detection for temporal properties is
largely automatic, and nothing equivalent exists for the command- and
test-suite-shaped gates that dominate real harnesses. Borrowing that tooling --- an
automatic ``did this check actually constrain anything'' pass over a gate --- is the most
promising concrete import from formal methods into harness engineering that we can see,
and we do not attempt it here.

\subsection{Workflow verification}
\label{sec:related-workflow}

Termination and completion properties for graph-structured processes are not new
either. Workflow nets \citep{vanderaalst1997verification,vanderaalst1998petri} give a
Petri-net formulation of business processes with the same topology used here, a DAG
with splits and merges, and a notion of \emph{soundness}: every case can complete, no
task is dead. \Cref{thm:graph-termination,thm:soundness} are, in that vocabulary,
soundness results for a workflow net whose transitions are non-deterministic agents and
whose completion condition is adjudicated by a separate observer. The classical setting
assumes transitions are known behaviours; ours assumes nothing about the worker at all,
which is what makes the independent gate necessary rather than merely prudent.

That literature also fixes the shape of what can be claimed. A repair edge
restores budget that has already been spent, which in net terms is a \emph{reset} arc,
and \citet{dufourd1998reset} map where reset arcs push a net over the edge. The
frontier is not where intuition puts it: reachability and boundedness become
undecidable, while termination and coverability remain decidable --- a pattern the
well-structured-transition-system framework of \citet{finkel2001wsts} explains, since
reset nets retain the well-quasi-ordering that makes termination decidable. So we do
\emph{not} claim that termination of a repair-bearing graph is undecidable; for this
class it is not. What does become undecidable is workflow \emph{soundness}, which asks
for proper completion and is therefore a reachability property
\citep{vanderaalst2011soundness}.

Two consequences shape \cref{sec:termination}. First, the property we prove is
deliberately the weaker and decidable one --- the graph stops --- and not soundness in
the workflow-net sense. Second, decidability buys an operator nothing on its own:
knowing that some procedure could in principle answer \emph{whether} a graph halts does
not say \emph{how much work} it will do first. What an operator needs before
authorising a run is a number, in advance, in quantities already written down, which is
why \cref{thm:graph-termination} carries an explicit bound and
\cref{cor:manifest-readable} insists it be computable from the manifest. The
classification of faults by whether they are the component's own or its environment's is
likewise long-established in dependability engineering
\citep{avizienis2004taxonomy}; our \cref{thm:trichotomy} is a specialisation in which
the boundary is read off a declaration the system already holds.

\subsection{Runtime verification, and contracts}
\label{sec:related-rv}

A gate is a \emph{monitor}: it observes an execution and decides whether it satisfies a
property, at run time rather than by static analysis of every possible execution. That
is runtime verification as \citet{leucker2009rv} define it, and the trade they identify
is ours exactly --- a monitor decides about the run it observes and not about all runs,
which is why \cref{thm:soundness} is a statement about a particular run's ledger rather
than about the worker. Their survey also anticipates the application: they close on
runtime verification for \emph{contract enforcement}. What the RV setting assumes and
this one cannot is a property expressed in a formal logic over a well-defined trace
alphabet. Here the property is hand-written code over an artifact of arbitrary shape, no
monitor synthesis is available, and consequently, as \cref{sec:vacuity} shows, the
monitor itself becomes an object that needs measuring.

The word \emph{contract} is likewise borrowed rather than coined.
\citet{meyer1992contract} established the discipline of stating obligations that a
component must meet and having them checked rather than assumed. The difference is which
party the contract binds. In design by contract the supplier is code, its precondition
is checkable by the caller, and violation is a defect in a component whose behaviour is
in principle knowable. Here the supplier is a non-deterministic agent about which we
assume nothing at all, so the contract cannot be a precondition on the worker; it has to
be a postcondition adjudicated by a party the worker cannot reach. That is the whole of
\cref{def:independence}, stated in the older vocabulary.

\subsection{Anytime-valid inference}

Reporting a rate over a log that grows and is inspected repeatedly is a sequential
problem: a fixed-time interval loses its guarantee the moment a reader is permitted to
look whenever they like. We use the betting-based confidence sequences of
\citet{waudbysmith2023}, which are valid simultaneously at all sample sizes.
\Cref{rem:estimand} is explicit that the covered quantity is the mean over attempts
\emph{in the log}, an audit quantity, and not a forecast of a population --- a
distinction we get wrong in an earlier internal write-up and report in \cref{sec:eval}
rather than quietly correcting.

\section{Preliminaries}
\label{sec:prelim}

This section fixes the objects the rest of the paper reasons about. Each definition
names something a running system already has, so that every later theorem is a
statement about an artifact a reader can inspect rather than about an idealisation.
Where a definition is narrower than the informal term it borrows, we say so at the
point of narrowing.

\subsection{Workspaces, artifacts, and ownership}

\begin{definition}[Workspace]
\label{def:workspace}
A \emph{workspace} $\Omega$ is a finite set of files under a single root
directory, together with their contents. We write $\Omega_t$ for the workspace
at step $t$ and treat $\Omega$ as the complete mutable state visible to a
worker: a worker that cannot write outside $\Omega$ cannot affect anything
else.
\end{definition}

\begin{definition}[Artifact]
\label{def:artifact}
An \emph{artifact} $\Art$ is a distinguished subset of the workspace that a
gate reads in order to decide. A loop may have several; \cref{sec:eval} reports
a class of loops whose requirement is a \emph{relation between} two artifacts
and which single-artifact mutation therefore cannot express.
\end{definition}

The next definition does the most work later, and it is the part of the design most
often left implicit in deployed harnesses.

\begin{definition}[Ownership partition]
\label{def:ownership}
Each loop declares a finite set of glob patterns $F$, its \emph{forbid set}.
This induces a partition of the workspace into
\[
  \Omega \;=\; \Omega^{\mathrm{own}} \;\sqcup\; \Omega^{\mathrm{anchor}},
  \qquad
  \Omega^{\mathrm{anchor}} \;=\; \{\, f \in \Omega \;:\; f \text{ matches some pattern in } F \,\}.
\]
Files in $\Omega^{\mathrm{own}}$ are \emph{worker-owned}: producing them is the
worker's job, and anything wrong with one of them is the worker's fault. Files
in $\Omega^{\mathrm{anchor}}$ are \emph{loop-shipped anchors}: the gate's own
code, its ground-truth reference data, and any test-collection configuration.
The worker is never permitted to write them.
\end{definition}

$F$ is written by the loop's author in the manifest, before any run. It is not
derived from observing behaviour. This matters for \cref{thm:trichotomy}: the
question "is this failure the worker's fault?" --- which a harness must answer
every time a check fails --- has already been answered, statically, by the
author who knew. A harness that instead decides this at failure time is
guessing about intent at the worst possible moment.

\subsection{Workers, gates, and the independence requirement}

\begin{definition}[Worker]
\label{def:worker}
A \emph{worker} $\Worker$ is any procedure that, given a specification and a
workspace, may modify the workspace and returns a transcript. Formally
$\Worker : (\mathrm{Spec}, \Omega) \to (\Omega', \mathrm{transcript})$. A worker
is permitted to assert that it is finished. Nothing in this paper ever reads
that assertion.
\end{definition}

\begin{definition}[Gate]
\label{def:gate}
A \emph{gate} $\Gate$ is a procedure that reads a workspace and returns a
verdict, $\Gate : \Omega \to \mathcal{V}$, where $\mathcal{V}$ is the verdict
set of \cref{def:verdict}. A gate does not modify the workspace and does not
read the worker's transcript.
\end{definition}

Two words are used throughout and are not synonyms. The \emph{gate} is the role
--- the object occupying $\Gate$ in the triple, whatever implements it. A
\emph{checker} is one concrete implementation of that role: a standalone
program, invoked as a command, that exits zero or non-zero. Most gates in our
catalogue are checkers; a schema validation or a test suite is a gate and is not
a checker. Statements about gates are claims about the contract, statements
about checkers are claims about a particular kind of code, and where the paper
reports a defect it matters which is meant.

\begin{definition}[Independence]
\label{def:independence}
A pair $(\Worker, \Gate)$ is \emph{independent} when $\Worker$ and $\Gate$ are
distinct objects, $\Gate$'s input excludes $\Worker$'s transcript, and
$\Gate$'s code lies in $\Omega^{\mathrm{anchor}}$ so that $\Worker$ cannot
write it.
\end{definition}

This is \emph{structural} independence, not statistical: different objects with
disjoint write authority, with no claim that their errors are uncorrelated
(\cref{sec:future}).

\begin{remark}[Why a separate validation step in the same codebase is not this]
\label{rem:separate-function}
The common production pattern --- a workflow with a worker step followed by a
validation step, in one repository, deployed together --- satisfies the first clause of
\cref{def:independence} and neither of the others. The validator is reachable by
whatever wrote the work, so ``the checker was not modified'' is a property of the
deployment's discipline and not of the run: nothing in the execution path prevents the
same commit from adjusting the artifact and the assertion that judges it. And the
validator typically receives the worker's output \emph{as its input} --- including, in
agent settings, the transcript or a summary of it --- so a worker that describes its
work persuasively is talking directly to its own judge.

The third clause closes both. $\Gate$'s code lies in the anchored region a worker
cannot write, and $\Gate$'s input is the artifact rather than the account of it. That
is checkable statically, before the run, from the loop's own declaration, which is why
\cref{thm:soundness} can be a theorem instead of a recommendation.
\end{remark}

\subsection{The verdict set}

\begin{definition}[Verdict]
\label{def:verdict}
The verdict set is the three-element set
\[
  \mathcal{V} \;=\; \{\, \Pass,\; \Reject,\; \Incap \,\}
\]
with the following intended readings. $\Pass$: the gate read the artifact and
the stated condition holds. $\Reject$: the gate read the artifact and the
stated condition does not hold. $\Incap$: the gate could not form an opinion at
all.
\end{definition}

\begin{definition}[Incapacity set]
\label{def:incapacity}
A gate's \emph{incapacity set} $I(\Gate)$ is the set of conditions under which
it is permitted to return $\Incap$. We require
\[
  I(\Gate) \;\subseteq\; \{\text{malformed invocation}\} \;\cup\;
  \{\text{a file in } \Omega^{\mathrm{anchor}} \text{ is absent or unreadable}\},
\]
and nothing else. In particular, an empty, vacuous, or malformed
\emph{worker-owned} artifact is not in $I(\Gate)$.
\end{definition}

\Cref{def:incapacity} is the paper's sharpest normative commitment and
\cref{sec:verdict} defends it at length. Informally: \emph{``there was nothing
to check'' is an answer, not an inability to answer.} A gate that files it as
an inability discards a diagnosis it already holds.

\subsection{Bounded loops}

\begin{definition}[Budget]
\label{def:budget}
A \emph{budget} $\Budget$ is a tuple of declared, finite ceilings, of which
this paper uses two: $\maxatt \in \N_{\geq 1}$, the maximum number of attempts,
and a no-progress window $w \in \N_{\geq 1}$. The implementation declares nine
such bounds plus a kill switch; the remaining seven are resource ceilings whose
only role in the theorems below is that they are finite.
\end{definition}

\begin{definition}[Bounded loop]
\label{def:loop}
A \emph{bounded loop} is a triple $\Loop = (\Worker, \Gate, \Budget)$ with
$(\Worker,\Gate)$ independent in the sense of \cref{def:independence}, executed
by the procedure of \cref{alg:loop}: repeatedly run the worker, then the gate,
stopping at the first $\Pass$ or when the budget is exhausted.
\end{definition}

\begin{definition}[Attempt]
\label{def:attempt}
An \emph{attempt} (equivalently a \emph{lap}) is one execution of $\Worker$
followed by one execution of $\Gate$. We index attempts $k = 1, 2, \dots$ and
write $\Verdict_k$ for the verdict of the $k$-th.
\end{definition}

\begin{definition}[Loop outcome]
\label{def:outcome}
A loop terminates in exactly one of five states. $\Done$: some
$\Verdict_k = \Pass$. $\Halt$: the budget was exhausted with no $\Pass$.
$\textsc{pause}$: a declared approval is outstanding. $\textsc{killed}$: the
kill switch fired. $\textsc{error}$: some $\Verdict_k = \Incap$. Only $\Done$
asserts anything about the artifact.
\end{definition}

\begin{remark}[$\Halt$ and $\textsc{error}$ are not the same event]
\label{rem:halt-vs-error}
$\Halt$ says the worker was given its declared number of chances and did not
succeed --- a statement about the work. $\textsc{error}$ says the harness could
not tell --- a statement about the harness. Collapsing them is the failure mode
\cref{sec:vacuity} is about, and \cref{sec:eval} reports what it cost us when
we collapsed them ourselves.
\end{remark}

\subsection{Receipts}

\begin{definition}[Receipt and ledger]
\label{def:receipt}
A \emph{receipt} $r_k$ is a record of the $k$-th attempt containing at minimum
the verdict, the decision taken, the budget consumed, and the hash
$h(r_{k-1})$ of its predecessor. The \emph{ledger} is the append-only sequence
$r_1, r_2, \dots, r_n$. The ledger is the sole durable state: every surface
(command line, monitor, arena, embedding API) is a projection of it and none
holds state of its own.
\end{definition}

Hash-chaining makes a silent edit of history detectable on replay
(\cref{thm:chain}). The property is tamper-\emph{evidence} against partial edits and
truncation, which is what an audit of one's own runs actually needs; it is not a
signature, and \cref{rem:final-row} is exact about the two things it cannot do alone.

\subsection{Bounded-loop graphs}

\begin{definition}[Bounded-loop graph]
\label{def:graph}
A \emph{bounded-loop graph} is a directed acyclic graph
$\Graph = (\Nodes, \Edges)$ in which each node $v \in \Nodes$ carries a bounded
loop $\Loop_v = (\Worker_v, \Gate_v, \Budget_v)$, and each edge $(u,v)$ asserts
that $v$ may not begin until $u$ has terminated. Acyclicity is checked at
compile time.
\end{definition}

\begin{definition}[Graph outcome]
\label{def:graph-outcome}
A run of $\Graph$ terminates in $\Done$ iff every $v \in \Nodes$ terminates in
$\Done$ under $\Loop_v$ in the sense of \cref{def:outcome}. It terminates in
$\Halt$ if execution stops with some node not in $\Done$ and no approval
outstanding; in $\textsc{pause}$ if some node is in $\textsc{pause}$; and in
$\textsc{killed}$ or $\textsc{error}$ if any node is.
\end{definition}

\Cref{def:outcome} defines the five outcomes for a single loop, and it is easy to
proceed as though the graph-level meaning follows. It does not: ``the graph is
done'' is a separate predicate, and \cref{thm:soundness} quantifies over it. A
draft of this paper omitted this definition, which left that theorem's hypothesis
formally empty --- the proof silently supplied the missing content in a sentence
of prose. An external proof audit caught it.

\begin{definition}[Repair edge and repair round]
\label{def:repair}
A \emph{repair relation} $\Edges_r \subseteq \Nodes \times \Nodes$ is a set of
pairs $(v, u)$ with $u$ an ancestor of $v$ in $\Graph$. It is \textbf{disjoint
from $\Edges$ and is not part of the graph}: $\Graph = (\Nodes, \Edges)$ remains
acyclic by \cref{def:graph}, and no element of $\Edges_r$ is ever traversed as
dataflow. A pair $(v,u) \in \Edges_r$ is \emph{taken} only when $v$ terminates in
$\Halt$, and taking it ends the current pass and begins a new one in which $u$
and its descendants are re-admitted. The number of passes beyond the first is the
\emph{repair round} $\rho$, bounded by a declared $\rounds \in \N_{\geq 0}$.
\end{definition}

The separation is not bookkeeping. If $\Edges_r$ were added to $\Edges$ the result
would contain a cycle, $\Graph$ would not be a DAG, and the compile-time acyclicity
check of \cref{def:graph} could not be satisfied by any graph declaring repair.
Modelling repair instead as a \emph{bounded outer loop} over passes of an unchanged
DAG keeps each pass's structure exactly the acyclic one, with only the pass index
changing across a boundary. The implementation makes the same choice --- a repair
boundary is an explicit receipt in the ledger and nothing is revived in place --- so a
verifier replaying a run sees precisely where and why per-pass monotonicity
restarted.

What survives a boundary and what does not is the whole difficulty.
\Cref{sec:termination} needs a two-variable induction because within a pass each
node's attempt count is bounded, while crossing a boundary restores exactly the
quantity the inner induction decreases; \cref{thm:graph-termination} therefore inducts
on the pair $(\rounds - \rho,\ \text{remaining budget})$ under a lexicographic order.
No ordering of $\Nodes$ survives at all, which is why \cref{sec:soundness} does not
argue along one (\cref{rem:no-topological-induction}).

\begin{figure}[t]
\centering
\begin{tikzpicture}[
  >={Stealth[round]}, node distance=9mm,
  box/.style={draw, rounded corners=2pt, minimum height=8mm, inner xsep=4pt, font=\small},
  worker/.style={box, fill=blue!6},
  gate/.style={box, fill=orange!10},
  art/.style={box, fill=black!4},
  dec/.style={draw, diamond, aspect=2.1, inner sep=1pt, font=\small\itshape, fill=green!6},
  lbl/.style={font=\scriptsize\itshape, inner sep=1pt},
]
\node[worker] (w) {$\Worker$ \; worker};
\node[art, right=14mm of w] (a) {$\Omega^{\mathrm{own}}$};
\node[gate, right=14mm of a] (g) {$\Gate$ \; gate};
\node[dec, right=14mm of g] (d) {$\Verdict_k$};
\node[box, right=13mm of d, fill=green!12] (done) {$\Done$};
\node[box, below=11mm of d, fill=red!8] (halt) {$\Halt$};

\draw[->] (w) -- node[lbl, above] {writes} (a);
\draw[->] (a) -- node[lbl, above] {reads} (g);
\draw[->] (g) -- (d);
\draw[->] (d) -- node[lbl, above] {\textsc{pass}} (done);
\draw[->] (d) -- node[lbl, right] {budget out} (halt);

\draw[->, rounded corners] (d.north) -- ++(0,7mm) -- node[lbl, above] {\textsc{reject} $\;\to\;$ attempt $k{+}1$} ++(-52mm,0) -- (w.north);

\node[art, below=11mm of a, fill=orange!6] (anchor) {$\Omega^{\mathrm{anchor}}$ \;\scriptsize(gate code, ground truth)};
\draw[->] (anchor) -- (g);
\draw[->, dashed, red!70!black] (w.south) |- node[lbl, below, pos=0.72, text=red!70!black] {refused} (anchor.west);

\begin{scope}[on background layer]
  \node[draw, dashed, gray, rounded corners, fit=(w)(a)(g)(d)(anchor), inner sep=6pt] {};
\end{scope}
\end{tikzpicture}
\caption{One bounded loop. The worker writes only $\Omega^{\mathrm{own}}$; the
gate reads the artifact and the anchors, never the worker's transcript, and the
worker's write to an anchor is refused rather than merely discouraged
(\cref{def:independence,def:ownership}). The verdict, not the worker, selects
the outgoing edge. \textsc{incapacity} is omitted here and treated separately
in \cref{fig:trichotomy}, because the whole argument of \cref{sec:verdict} is
about how narrow that third edge must be.}
\label{fig:loop}
\end{figure}
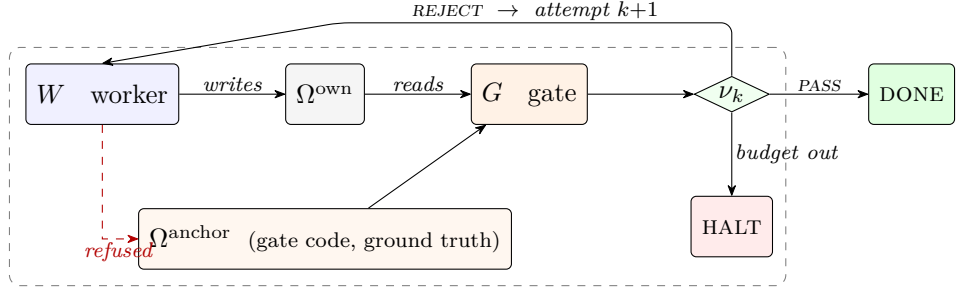


\begin{figure}[t]
\centering
\begin{tikzpicture}[
  font=\small,
  box/.style={draw, rounded corners=2pt, align=center, inner sep=4pt, minimum height=7mm},
  ctrl/.style={box, fill=black!4, minimum width=26mm},
  agent/.style={box, fill=black!10, minimum width=26mm},
  gate/.style={box, fill=black!4, minimum width=26mm, very thick},
  store/.style={box, fill=white, minimum width=26mm, draw=black!60},
  decl/.style={box, fill=white, dashed, minimum width=22mm},
  flow/.style={-{Stealth[length=2mm]}, thick},
  weak/.style={-{Stealth[length=2mm]}, densely dotted},
]

\node[decl] (manifest) {\texttt{bounds.yaml}\\\scriptsize $m$, $W$, $r$, $w$, tokens};
\node[decl, right=5mm of manifest] (spec)  {\texttt{loop.yaml}\\\scriptsize spec, gate, \texttt{forbid}};

\node[ctrl, below=9mm of $(manifest.south)!0.5!(spec.south)$] (controller)
  {\textbf{controller}\\\scriptsize admits or refuses each attempt};

\node[agent, below left=11mm and 6mm of controller] (worker)
  {\textbf{worker}\\\scriptsize agent process, one turn};
\node[store, below=8mm of worker] (workspace)
  {workspace\\\scriptsize agent-writable};

\node[gate, below right=11mm and 6mm of controller] (gatebox)
  {\textbf{gate}\\\scriptsize mechanical check, no worker input};
\node[store, below=8mm of gatebox] (ledger)
  {ledger\\\scriptsize append-only, hash-chained};

\node[box, fill=black!4, below=8mm of $(workspace.south)!0.5!(ledger.south)$, minimum width=44mm]
  (receipt) {receipt: status, attempts spent, \textbf{head digest}};

\draw[flow] (manifest) -- (controller);
\draw[flow] (spec) -- (controller);
\draw[flow] (controller) -- node[left, font=\scriptsize, pos=0.45, xshift=-0.8mm] {prompt + budget} (worker);
\draw[flow] (worker) -- node[left, font=\scriptsize] {writes} (workspace);
\draw[flow] (workspace.east) -- node[above, font=\scriptsize, pos=0.5] {reads} (gatebox.west);
\draw[flow] (gatebox) -- node[right, font=\scriptsize] {verdict} (ledger);
\draw[flow] (controller) -- node[right, font=\scriptsize, pos=0.45] {invoke} (gatebox);
\draw[flow] (ledger) -- (receipt);
%
\draw[weak] (worker.east) to[bend left=12] (controller.south);

\begin{scope}[on background layer]
  \node[draw=black!55, dashed, rounded corners=3pt, inner sep=6pt,
        fit=(worker)(workspace),
        label={[font=\scriptsize\itshape, anchor=south west, yshift=0.3mm]north west:untrusted}]
        (untrusted) {};
\end{scope}

\end{tikzpicture}
\caption[End-to-end architecture of a bounded loop]{%
\textbf{F1 --- a bounded loop, end to end, with the trust boundary drawn.}
Everything above the controller is \emph{declared} and readable before the run
(\cref{cor:manifest-readable}). Everything inside the dashed region is written by the
party under judgement. The two solid paths into the ledger are the only ones that can
produce a completion: the gate reads the workspace and records a verdict. The dotted
path is the worker's own account of itself --- it is stored, and it reaches no guard in
the procedure (\cref{lem:node-level}). The ledger and the receipt sit \emph{outside}
the untrusted region: they are written beside the loop package, not inside the
workspace, so an agent cannot edit the record that judges it.}
\label{fig:architecture}
\end{figure}
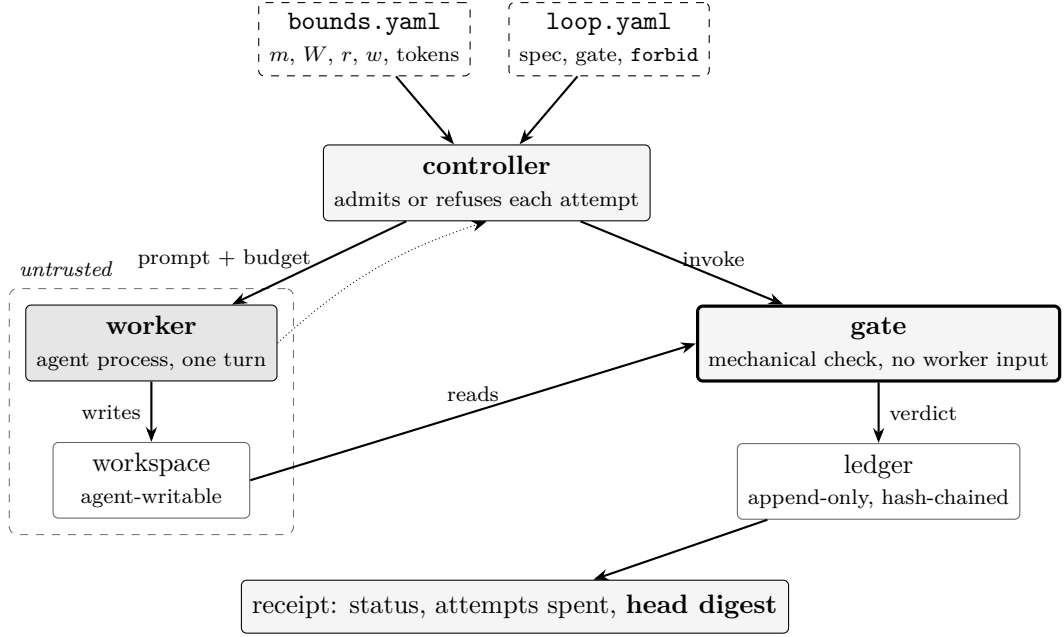

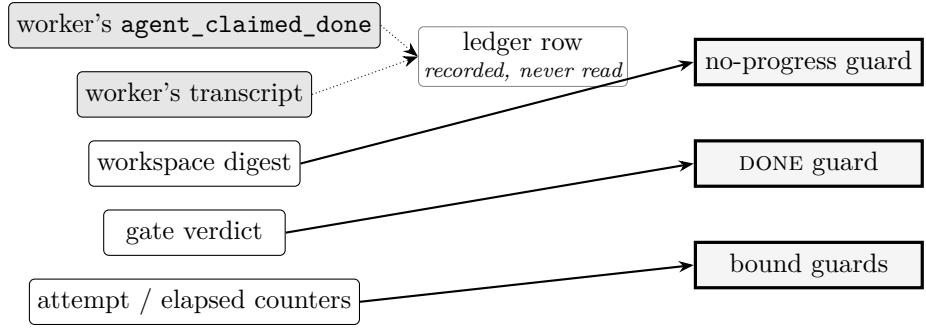
\begin{figure}[t]
\centering
\begin{tikzpicture}[
  font=\small,
  src/.style={draw, rounded corners=2pt, align=center, inner sep=3pt, minimum width=24mm, minimum height=6mm},
  guard/.style={draw, align=center, inner sep=3pt, minimum width=30mm, minimum height=6mm, very thick},
  yes/.style={-{Stealth[length=2mm]}, thick},
  no/.style={densely dotted, -{Stealth[length=2mm]}},
]
\node[src, fill=black!10] (claim)      {worker's \texttt{agent\_claimed\_done}};
\node[src, fill=black!10, below=3mm of claim] (transcript) {worker's transcript};
\node[src, fill=white,    below=3mm of transcript] (digest) {workspace digest};
\node[src, fill=white,    below=3mm of digest] (verdict) {gate verdict};
\node[src, fill=white,    below=3mm of verdict] (counters) {attempt / elapsed counters};

\node[guard, right=52mm of digest, fill=black!4] (done) {\textsc{done} guard};
\node[guard, above=7mm of done, fill=black!4] (progress) {no-progress guard};
\node[guard, below=7mm of done, fill=black!4] (bound) {bound guards};

\node[src, right=14mm of transcript, fill=white, draw=black!45, minimum width=22mm,
      align=center, inner sep=2pt, yshift=5mm]
  (record) {ledger row\\[-2pt]{\scriptsize\itshape recorded, never read}};

\draw[yes] (verdict.east)  -- (done.west);
\draw[yes] (digest.east)   -- (progress.west);
\draw[yes] (counters.east) -- (bound.west);
\draw[no]  (claim.east)      -- (record.west);
\draw[no]  (transcript.east) -- (record.west);

\end{tikzpicture}
\caption[Information flow into the guards]{%
\textbf{M3 --- which facts can decide a completion.} Shaded boxes are authored by the
party under judgement. Solid arrows are the only inputs any guard in the procedure
reads; the worker's self-report and its transcript reach the ledger and stop there.
This is the drift guarantee (\cref{thm:soundness}) as a picture: an agent cannot route
to \textsc{done} by asserting it, because no branch consults the assertion. The
adversarial test for this claim inverts it --- a runner that always claims completion,
paired with a gate that always refuses, must reach \textsc{halt} and never
\textsc{done}.}
\label{fig:information-flow}
\end{figure}

\section{The Verdict Contract}
\label{sec:verdict}

Every harness that checks its own work must eventually answer a question that
looks like plumbing and is not: when the check fails, \emph{whose failure was
it?} Either the artifact is wrong --- in which case the worker should be told
and given another attempt --- or the checking apparatus is broken, in which case
retrying is pointless and the run should stop and say so.

Deployed harnesses answer this question implicitly, usually with an exit code
and usually by accident. This section makes the answer explicit, proves the
resulting classification is total and disjoint, and argues that the boundary
belongs where \cref{def:ownership} already put it: at ownership, declared by the
loop's author before the run.

\subsection{Two verdicts are not enough}

The tempting design is binary --- a check passes or it does not --- and it is
wrong in a specific way. Consider a loop whose worker must emit a JSON document
and whose gate validates it against a schema, and suppose the worker emits
malformed JSON. Under a binary contract the gate has two options and both
misreport:

\begin{itemize}[leftmargin=1.4em, itemsep=1pt]
  \item Return \emph{fail}, and the run is indistinguishable from a schema
        violation --- recoverable, so the loop retries. Correct behaviour,
        arrived at by conflating two situations.
  \item Return \emph{error}, and the loop stops. But the gate knew exactly what
        was wrong and could have said so, and the thing it knew was actionable.
\end{itemize}

Now change one detail: the file that will not parse is the gate's own
ground-truth reference data, which the worker was never allowed to touch. The
correct behaviour flips. Retrying cannot help, because no action available to
the worker changes that file. The run should stop.

The two situations are byte-for-byte identical at the point of failure --- a
file that will not parse --- and require opposite responses. No amount of
inspection at failure time distinguishes them, because the distinguishing fact
is not about the bytes. It is about \emph{who was supposed to produce them}.

\subsection{The trichotomy}

\begin{theorem}[Verdict trichotomy]
\label{thm:trichotomy}
Fix a gate $\Gate$ that terminates on every input and complies with
\cref{def:incapacity}, and let $\Omega$ be any workspace. Then exactly one of
$\Pass$, $\Reject$, $\Incap$ applies to the execution $\Gate(\Omega)$, and the
classification is decidable from $\Omega$ and the declared forbid set $F$ alone
--- without reference to the worker, its transcript, or the history of the run.
\end{theorem}

\begin{proof}
\emph{Totality.} Either $\Gate$ completes an evaluation of its stated condition
or it does not. If it completes, the condition either holds ($\Pass$) or fails
($\Reject$). If it does not complete, the obstruction is by
\cref{def:incapacity} either a malformed invocation or an absent or unreadable
file $f \in \Omega^{\mathrm{anchor}}$, giving $\Incap$. No fourth case exists,
because the three antecedents --- completes-and-holds, completes-and-fails,
does-not-complete --- exhaust the possibilities.

\emph{Disjointness.} $\Pass$ and $\Reject$ are disjoint since the condition
cannot both hold and fail. Both are disjoint from $\Incap$ since both require
the evaluation to have completed and $\Incap$ requires that it did not.

\emph{Decidability from $\Omega$ and $F$.} The only input to the case analysis
beyond $\Omega$ is membership of the obstructing file in
$\Omega^{\mathrm{anchor}}$, which by \cref{def:ownership} is determined by
matching the file's path against $F$. $F$ is fixed in the manifest before the
run. Neither the worker nor the transcript nor the attempt index appears in the
analysis.
\end{proof}

\begin{corollary}[Vacuity is a rejection]
\label{cor:vacuity-rejects}
Let $\Omega^{\mathrm{own}}$ be empty, or contain an artifact that is empty or
malformed. Then $\Gate(\Omega) \ne \Incap$.
\end{corollary}

\begin{proof}
By \cref{def:incapacity}, $I(\Gate)$ contains no condition concerning
$\Omega^{\mathrm{own}}$. The obstruction is therefore not in the incapacity set,
and by the totality argument of \cref{thm:trichotomy} the verdict is $\Pass$ or
$\Reject$ --- an answer about the work either way, which is all the corollary
claims.
\end{proof}

The verdict is not necessarily $\Reject$, and the other case is the interesting one: a
gate whose condition is satisfied by the empty artifact \emph{completes} and returns
$\Pass$. That gate is vacuous in the
sense of \cref{def:vacuous}, and \cref{thm:vacuity} is exactly the statement that
it then certifies an artifact someone deleted. So the corollary establishes the
weaker and correct thing, that vacuity never reaches the harness as an inability to
answer, and \cref{sec:vacuity} handles the case where the answer it does reach
is the wrong one.

The corollary is small: \emph{``there is nothing here to check'' is a finding about
the work, not a confession by the harness.} A gate that reports it as incapacity has
performed the diagnosis and then thrown it away, and the loop, which was about to
hand the worker exactly that message, stops instead.

\begin{figure}[t]
\centering
\begin{tikzpicture}[
  >={Stealth[round]}, node distance=7mm,
  q/.style={draw, rounded corners=2pt, align=center, font=\small, inner sep=4pt},
  v/.style={draw, rounded corners=2pt, font=\small\scshape, minimum width=20mm, inner sep=4pt},
  lbl/.style={font=\scriptsize\itshape, inner sep=1.5pt, fill=white},
]
\node[q] (start) {gate runs on $\Omega$};
\node[q, below=9mm of start] (complete) {did the evaluation\\complete?};
\node[q, below left=11mm and 12mm of complete] (own) {is the obstructing file\\in $\Omega^{\mathrm{own}}$?};
\node[q, below right=11mm and 6mm of complete] (holds) {does the stated\\condition hold?};

\node[v, fill=green!14, below=11mm of holds] (pass) {pass};
\node[v, below=26mm of complete, fill=yellow!22] (reject) {reject};
\node[v, fill=red!12, below=11mm of own] (incap) {incapacity};

\draw[->] (start) -- (complete);
\draw[->] (complete) -- node[lbl, right, pos=0.35] {yes} (holds);
\draw[->] (complete) -- node[lbl, left, pos=0.35] {no} (own);
\draw[->] (holds) -- node[lbl, right] {yes} (pass);
\draw[->] (holds.west) -- node[lbl, above] {no} (reject.east);
\draw[->] (own.east) -- node[lbl, above] {yes} (reject.west);
\draw[->] (own) -- node[lbl, left] {no --- it is an anchor} (incap);

\node[align=left, font=\scriptsize, below=4mm of reject, text width=0.86\linewidth]
  {\textbf{The branch that matters.} Both no-completion paths look identical at the point
   of failure. Only the \emph{declared} ownership of the file separates them, and it was
   declared before the run.};
\end{tikzpicture}
\caption{The verdict trichotomy of \cref{thm:trichotomy}. The single
distinguishing test is ownership, read off the loop's own \texttt{forbid} set
rather than inferred from the failure. Every path into \textsc{incapacity}
passes through a file the worker was never permitted to write, which is what
makes \cref{cor:vacuity-rejects} hold.}
\label{fig:trichotomy}
\end{figure}

\subsection{Why ownership, and not severity}

An alternative boundary suggests itself: classify by how bad the failure looks.
Parse errors and crashes become incapacity; assertion failures become rejection.
This is what a harness does when nobody has decided, and it fails on the example
above --- the same parse error lands on both sides of the line depending on
whose file it is.

Severity is also the wrong \emph{kind} of question. It asks the gate to estimate, at
failure time, whether the worker could plausibly fix this, and the gate has no basis
for that estimate: it sees a byte stream. The loop's author knew precisely which files
the worker is responsible for, well enough to have written them into $F$ so the
anti-tamper guard could enforce it. \Cref{def:incapacity} does no more than read a
decision already made and recorded, for a different purpose, by someone who had the
information.

\begin{remark}[This is the tool-error convention, applied one level up]
\label{rem:industry-convention}
The contract is not novel in spirit. Anthropic's guidance for tool-using models is
that a tool which fails should return its error \emph{to the model} rather than raise
past it, precisely so the model can react; the major cloud SDKs separate transient
from permanent faults so that a caller knows whether a retry is meaningful; and
\texttt{terraform plan} returns three exit codes rather than two, because ``succeeded
with changes'' and ``failed'' are genuinely different answers. What is new is not the
trichotomy but the claim that the boundary is \emph{already declared} by the
anti-tamper configuration a careful loop author writes anyway, so the harness need not
guess and can be proved not to (\cref{thm:trichotomy}).
\end{remark}

\subsection{What it costs to get this wrong}

The contract is stated crisply because we did not have it. Before it was written
down, twenty-four checkers in our own catalogue returned the incapacity code for an
empty, vacuous, or malformed \emph{worker-owned} artifact, and both consequences were
live at once:

\begin{itemize}[leftmargin=1.4em, itemsep=2pt]
  \item \textbf{In the product.} A worker that emitted broken JSON ended the run
        instead of being told to fix it. The information needed for the retry
        existed, in hand, at the moment the run was abandoned.
  \item \textbf{In the measurement.} $84$ of $233$ held-out mutants were recorded
        as \emph{not judged} and left the denominator entirely, while seven loops
        produced nothing but incapacity and the surviving rate looked perfect
        (\cref{sec:eval}).
\end{itemize}

The second is the more instructive. A verdict misfiled as an error does not merely
lose a retry; it silently removes the case from the evidence base. The rate stayed
clean because the hard cases stopped counting.


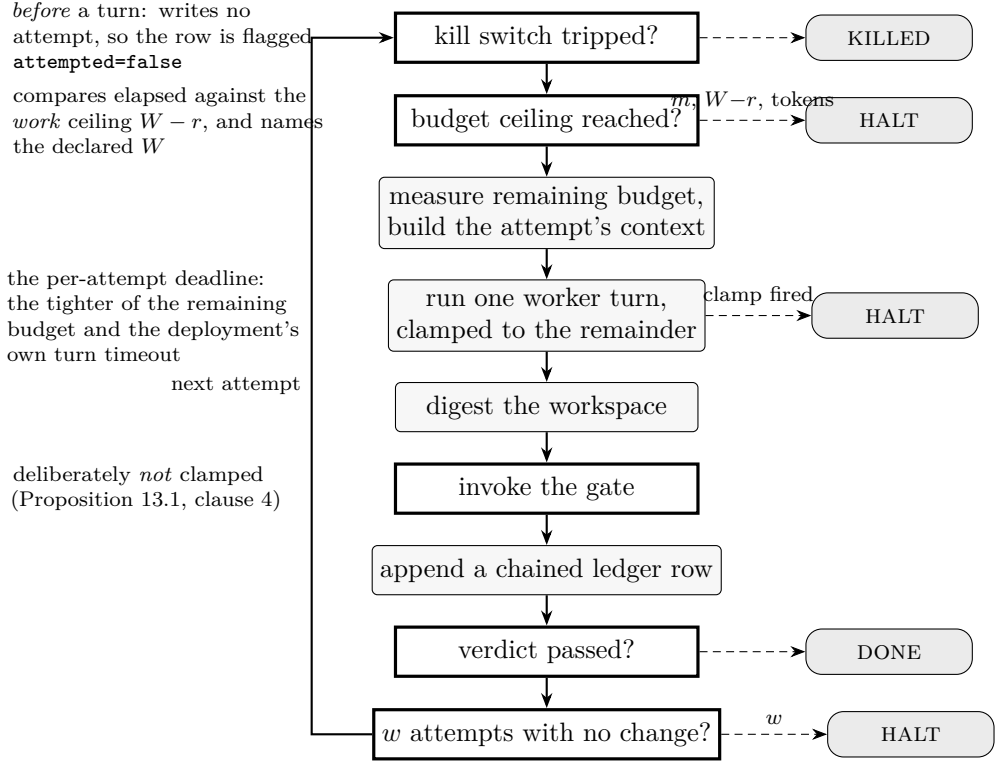
\begin{figure}[t]
\centering
\begin{tikzpicture}[
  font=\small,
  stagebox/.style={draw, rounded corners=2pt, align=center, inner sep=3pt,
               minimum width=40mm, minimum height=6.5mm, fill=black!3},
  guardbox/.style={draw, align=center, inner sep=3pt, minimum width=40mm,
                minimum height=6.5mm, very thick, fill=white},
  termbox/.style={draw, rounded corners=6pt, align=center, inner sep=3pt,
               minimum width=22mm, minimum height=6mm, fill=black!8},
  fwd/.style={-{Stealth[length=2mm]}, thick},
  side/.style={-{Stealth[length=2mm]}, densely dashed},
  ann/.style={font=\scriptsize, align=left, text width=41mm},
]

\node[guardbox] (kill)    {kill switch tripped?};
\node[guardbox, below=4mm of kill]  (budget)  {budget ceiling reached?};
\node[stagebox,  below=4mm of budget] (ctx)    {measure remaining budget,\\build the attempt's context};
\node[stagebox,  below=4mm of ctx]   (turn)    {run one worker turn,\\clamped to the remainder};
\node[stagebox,  below=4mm of turn]  (digest)  {digest the workspace};
\node[guardbox, below=4mm of digest] (gate)   {invoke the gate};
\node[stagebox,  below=4mm of gate]  (record)  {append a chained ledger row};
\node[guardbox, below=4mm of record] (pass)   {verdict passed?};
\node[guardbox, below=4mm of pass]  (stall)   {$w$ attempts with no change?};

\node[termbox, right=14mm of kill]   (killed) {\textsc{killed}};
\node[termbox, right=14mm of budget] (halt1)  {\textsc{halt}};
\node[termbox, right=14mm of turn]   (halt2)  {\textsc{halt}};
\node[termbox, right=14mm of pass]   (done)   {\textsc{done}};
\node[termbox, right=14mm of stall]  (halt3)  {\textsc{halt}};

\draw[fwd] (kill) -- (budget);
\draw[fwd] (budget) -- (ctx);
\draw[fwd] (ctx) -- (turn);
\draw[fwd] (turn) -- (digest);
\draw[fwd] (digest) -- (gate);
\draw[fwd] (gate) -- (record);
\draw[fwd] (record) -- (pass);
\draw[fwd] (pass) -- (stall);
\draw[side] (kill)   -- (killed);
\draw[side] (budget) -- node[above, font=\scriptsize] {$m$, $W{-}r$, tokens} (halt1);
\draw[side] (turn)   -- node[above, font=\scriptsize] {clamp fired} (halt2);
\draw[side] (pass)   -- (done);
\draw[side] (stall)  -- node[above, font=\scriptsize] {$w$} (halt3);
\draw[fwd] (stall.west) -- ++(-8mm,0) |- node[left, font=\scriptsize, pos=0.25] {next attempt} (kill.west);

\node[ann, left=8mm of kill]   {\emph{before} a turn: writes no attempt, so the row is
                                flagged \texttt{attempted=false}};
\node[ann, left=8mm of budget] {compares elapsed against the \emph{work} ceiling
                                $W - r$, and names the declared $W$};
\node[ann, left=8mm of turn]   {the per-attempt deadline: the tighter of the remaining
                                budget and the deployment's own turn timeout};
\node[ann, left=8mm of gate]   {deliberately \emph{not} clamped
                                (\cref{prop:spend-bound}, clause 4)};

\end{tikzpicture}
\caption[Control flow with every bound's enforcement point]{%
\textbf{F2 --- one attempt, with every declared bound shown where it is actually
enforced.} Heavy boxes are the points that can end a run; the annotations name what each
reads. Two of the five checks happen \emph{before} a turn, which is why the ledger flags
whether an attempt was spent --- without it, a ceiling halt at $m$ writes an
$(m{+}1)$-th row and utilisation computed from the receipt reads $1.1$ for a bound that
held exactly (\cref{def:bound-utilisation}). The clamp on the worker turn is the
one enforcement point that was added last; every other box on this diagram predates it.}
\label{fig:control-flow}
\end{figure}

\begin{figure}[t]
\centering
\begin{tikzpicture}[font=\small, x=1.05mm, y=1mm]

\fill[black!6]  (0,0)   rectangle (74,7);
\fill[black!18] (74,0)  rectangle (100,7);
\fill[black!2]  (100,0) rectangle (118,7);
\draw (0,0) rectangle (100,7);
\draw[dashed] (100,0) rectangle (118,7);
\draw[very thick] (74,-1.2) -- (74,8.2);

\node at (37,3.5)  {\small work: attempts admitted here};
\node at (87,3.5)  {\small reserve $r$};
\node at (109,3.5) {\small gate tail $G$};

\draw[-{Stealth[length=2mm]}] (0,-6) -- (122,-6) node[right, font=\scriptsize] {elapsed};
\foreach \x/\l in {0/{$0$}, 74/{$W-r$}, 100/{$W$}, 118/{$W+G$}} {
  \draw (\x,-5) -- (\x,-7) node[below, font=\scriptsize] {\l};
}

\node[font=\scriptsize, align=left, anchor=north west] at (0,-11)
  {no attempt \emph{begins} after $W-r$, and none \emph{continues} past it};
\node[font=\scriptsize, align=left, anchor=north west] at (74,-11)
  {at most one\\wind-down turn};
\node[font=\scriptsize, align=left, anchor=north west] at (100,-11)
  {one gate\\in flight};

\node[font=\scriptsize, anchor=north west, align=left, text width=128mm] at (0,-21)
  {The reserve \textbf{partitions} the declared ceiling; it never extends it.
   $(W-r) + r = W$, so every bound expressed in $W$ is unchanged and nothing had to be
   re-proved to make a hard stop leave usable work behind
   (\cref{cor:reserve-preserves-bounds}). Granting an extra turn \emph{on top} of $W$
   would have been the obvious design, and would have made $W$ mean ``$W$ plus however
   long a summary takes'' --- the one quantity an operator cannot read from the
   manifest.};

\end{tikzpicture}
\caption[The declared wallclock, partitioned]{%
\textbf{F4 --- the spend contract on a time axis.} The solid box is what the manifest
declares. $G$ is dashed because it is the one term outside the declaration: a gate is
never cut short, since terminating a check part-way yields no verdict and a check that
could not run must never be recorded as having judged. A verdict costing a few seconds
past the ceiling is worth strictly more than no verdict, so we pay it and publish it
rather than buying a tidier inequality with an unusable outcome.}
\label{fig:spend-axis}
\end{figure}

\section{Termination Under Repair, and Why the Budget Must Be Global}
\label{sec:termination}

A harness that retries must say when it stops. This section proves that it does,
and --- more usefully --- gives the stopping point as a closed expression in
quantities the author declared before the run, so that an operator can read the
worst case off the manifest rather than discovering it from a bill. Neither theorem
below says the process halts \emph{usefully}; \cref{rem:termination-scope} is
explicit about that limit.

\subsection{The loop procedure}

\begin{figure}[t]
\centering
\fbox{\begin{minipage}{0.92\linewidth}
\small
\textbf{Procedure} $\textsc{RunLoop}(\Loop = (\Worker,\Gate,\Budget),\ \Omega_0)$
\begin{enumerate}[leftmargin=2.2em, itemsep=1pt, topsep=3pt]
  \item $k \gets 0$; \; $\Omega \gets \Omega_0$
  \item \textbf{while} $k < \Budget.\maxatt$:
  \begin{enumerate}[leftmargin=1.6em, label=\arabic{enumi}.\arabic*, itemsep=1pt]
    \item $k \gets k + 1$
    \item $\Omega \gets \Worker(\mathrm{Spec}, \Omega)$ \hfill \textit{worker may write only $\Omega^{\mathrm{own}}$}
    \item \textbf{if} $\Omega^{\mathrm{anchor}}$ changed \textbf{then} \textbf{return} $\textsc{killed}$
    \item $\Verdict_k \gets \Gate(\Omega)$ \hfill \textit{gate never sees the worker's transcript}
    \item append receipt $r_k$ to the ledger
    \item \textbf{if} $\Verdict_k = \Pass$ \textbf{then} \textbf{return} $\Done$
    \item \textbf{if} $\Verdict_k = \Incap$ \textbf{then} \textbf{return} $\textsc{error}$
    \item \textbf{if} no progress in the last $w$ attempts \textbf{then} \textbf{return} $\Halt$
  \end{enumerate}
  \item \textbf{return} $\Halt$
\end{enumerate}
\end{minipage}}
\caption{The loop procedure. Line 2.6 is the only place $\Done$ is produced,
and its guard reads a gate verdict. There is no branch anywhere in this
procedure that reads the worker's claim about its own completion; that is the
invariant \cref{thm:soundness} formalises.}
\label{alg:loop}
\end{figure}

\subsection{A single loop}

\begin{theorem}[Loop termination]
\label{thm:loop-termination}
For any worker $\Worker$, any gate $\Gate$, and any budget $\Budget$ with
$\maxatt = m \in \N_{\geq 1}$, $\textsc{RunLoop}$ terminates after at most $m$
attempts, provided each individual invocation of $\Worker$ and $\Gate$
terminates.
\end{theorem}

\begin{proof}
Let $k_t$ denote the value of the counter $k$ after $t$ iterations of the
\textbf{while} body. We show by induction on $t$ that $k_t = t$ and that the
loop guard is evaluated with $k = t$.

\emph{Base case.} Before the first iteration $k_0 = 0$ by line 1, and
$0 < m$ since $m \geq 1$, so the body is entered and line 2.1 sets $k_1 = 1$.

\emph{Inductive step.} Suppose $k_t = t$ and the body has been entered $t$
times. The guard is evaluated with $k = t$. If $t \geq m$ the procedure returns
$\Halt$ at line 3 and has terminated after $t \le m$ attempts. If $t < m$, line
2.1 sets $k_{t+1} = t + 1$, establishing the hypothesis at $t+1$. No other line
modifies $k$, and no line decreases it.

Since $(k_t)_{t \geq 0}$ is the strictly increasing sequence $k_t = t$ in $\N$
and the guard fails as soon as $k_t \geq m$, the body executes at most $m$
times. Each execution performs \emph{at most} one attempt in the sense of
\cref{def:attempt}: a body execution returning at line 2.3 has run $\Worker$ but
not $\Gate$, and \cref{def:attempt} requires both, so it is not an attempt.
Every other body execution reaches line 2.4 and is a full attempt. Early returns
at lines 2.6, 2.7 and 2.8 occur after the attempt completes and only shorten the
run. Hence at most $m$ body executions occur and at most $m$ attempts, so
termination occurs after at most $m$ attempts.
\end{proof}

\begin{proposition}[The no-progress window can only tighten the bound]
\label{prop:no-progress}
Let $w \in \N_{\geq 1}$ be the no-progress window and suppose progress is
measured by a function $\pi$ of the workspace that the gate's evidence
determines. Write
$j^{\ast} = \max\bigl(\{0\} \cup \{k : \pi \text{ changed at attempt } k\}\bigr)$.
Then $\textsc{RunLoop}$ terminates after at most $\min\{\, m,\; w + j^{\ast} \,\}$
attempts.
\end{proposition}

\begin{proof}
Immediate from line 2.8: if $\pi$ last changed at attempt $j^{\ast}$, the guard
fires at attempt $j^{\ast} + w$ at the latest, and line 3 caps the run at $m$
regardless. The $\{0\}$ matters: without it, a $\pi$ that never changes makes the
bound a maximum over the empty set, which is undefined in $\N$. The intended
value there is $\min\{m, w\}$, and $j^{\ast} = 0$ delivers it.
\end{proof}

``The loop stops when it stops making progress'' is true only relative to a
definition of progress, and here that definition is supplied by the gate's own
evidence rather than by the worker's report of effort.

\subsection{The graph procedure}
\label{sec:rungraph}

\Cref{thm:graph-termination} is a statement about an execution, so the execution
has to be written down. \Cref{alg:graph} is the scheduler, at the level of detail
the proof uses; it is the procedure the released engine runs, and every guard
below corresponds to one in the implementation.

\begin{figure}[t]
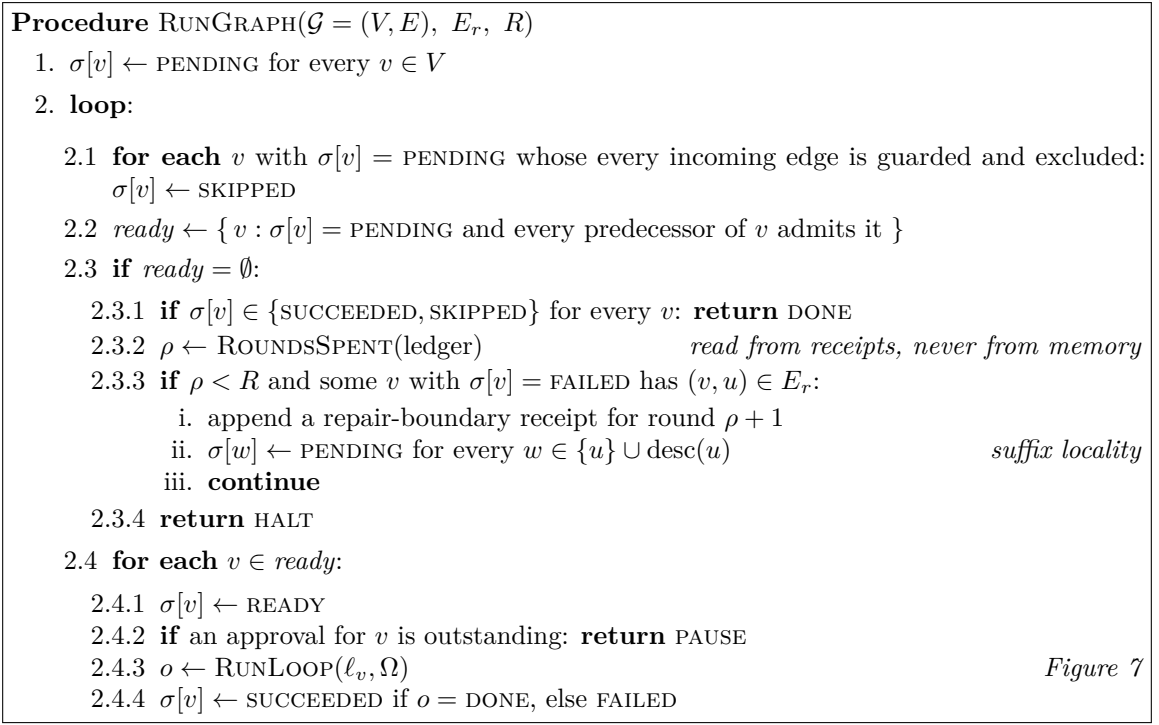

\centering
\fbox{\begin{minipage}{0.94\linewidth}
\small
\textbf{Procedure} $\textsc{RunGraph}(\Graph = (\Nodes,\Edges),\ \Edges_r,\ \rounds)$
\begin{enumerate}[leftmargin=2.2em, itemsep=1pt, topsep=3pt]
  \item $\sigma[v] \gets \textsc{pending}$ for every $v \in \Nodes$
  \item \textbf{loop}:
  \begin{enumerate}[leftmargin=1.6em, label=\arabic{enumi}.\arabic*, itemsep=1pt]
    \item \textbf{for each} $v$ with $\sigma[v] = \textsc{pending}$ whose every
          incoming edge is guarded and excluded: $\sigma[v] \gets \textsc{skipped}$
    \item $\mathit{ready} \gets \{\, v : \sigma[v] = \textsc{pending}$ and every
          predecessor of $v$ admits it $\}$
    \item \textbf{if} $\mathit{ready} = \emptyset$:
    \begin{enumerate}[leftmargin=1.8em, label=\arabic{enumi}.\arabic{enumii}.\arabic*, itemsep=1pt]
      \item \textbf{if} $\sigma[v] \in \{\textsc{succeeded}, \textsc{skipped}\}$
            for every $v$: \textbf{return} $\Done$
      \item $\rho \gets$ \textsc{RoundsSpent}(ledger)
            \hfill\textit{read from receipts, never from memory}
      \item \textbf{if} $\rho < \rounds$ and some $v$ with
            $\sigma[v] = \textsc{failed}$ has $(v,u) \in \Edges_r$:
      \begin{enumerate}[leftmargin=1.8em, label=\roman*., itemsep=0pt]
        \item append a repair-boundary receipt for round $\rho + 1$
        \item $\sigma[w] \gets \textsc{pending}$ for every
              $w \in \{u\} \cup \mathrm{desc}(u)$ \hfill\textit{suffix locality}
        \item \textbf{continue}
      \end{enumerate}
      \item \textbf{return} $\Halt$
    \end{enumerate}
    \item \textbf{for each} $v \in \mathit{ready}$:
    \begin{enumerate}[leftmargin=1.8em, label=\arabic{enumi}.\arabic{enumii}.\arabic*, itemsep=0pt]
      \item $\sigma[v] \gets \textsc{ready}$
      \item \textbf{if} an approval for $v$ is outstanding: \textbf{return}
            $\textsc{pause}$
      \item $o \gets \textsc{RunLoop}(\Loop_v, \Omega)$ \hfill\textit{\cref{alg:loop}}
      \item $\sigma[v] \gets \textsc{succeeded}$ if $o = \Done$, else
            $\textsc{failed}$
    \end{enumerate}
  \end{enumerate}
\end{enumerate}
\end{minipage}}
\caption{The graph procedure. Three guards carry \cref{thm:graph-termination}
and each is worth locating. Line 2.3 opens a repair round \emph{only when the
ready set is empty} --- a pass runs to completion before any boundary can open,
so repair never pre-empts ordinary progress. Line 2.3.2 reads the round count
from the ledger rather than from a variable, so a process that crashes and
resumes cannot repair without limit. Line 2.3.3.ii resets exactly the target and
its descendants, which is \cref{ass:bounded-repair}(1) in executable form.}
\label{alg:graph}
\end{figure}

Each guard is pinned by a released test, and --- because a passing test is exactly
what this paper argues proves nothing on its own --- each pin was checked by
removing the guard and confirming the test notices.

Writing this figure is what produced the check. Enumerating the guards revealed
that only one of the three, suffix locality, had a test; the other two were
properties the code had and the suite did not protect. Deleting the ready-set
condition fails its test with
a repair boundary appearing in the ledger ahead of a sibling branch that was
still runnable. Deleting the ledger-derived round count does not fail a test:
\textbf{it does not terminate}. A counter held in memory is reset by the restart
it is supposed to survive, so the run repairs without limit and
\cref{eq:graph-bound} bounds nothing. That the second mutant hangs rather than
failing is the sharper evidence, and it is the reason \cref{alg:graph} line
2.3.2 reads \textsc{RoundsSpent} from receipts rather than from a variable.

\textbf{Where \textsc{pause} comes from.} \Cref{def:outcome} admits $\textsc{pause}$ and \cref{alg:loop} never returns it,
which would leave \cref{thm:graph-termination}'s ``no node terminates in
$\textsc{pause}$'' hypothesis excluding a case the formal model cannot produce.
Line 2.4.2 is where it comes from: an outstanding human approval suspends the
\emph{graph}, not the loop, and the procedure returns without sealing the run.
The hypothesis is therefore about a reachable state, and what it excludes is
real --- a run waiting on a person does not terminate, and no bound over
declared budgets could make it. Termination here is a statement about the
machine's work, with the human left explicitly outside it.

\subsection{A graph of loops}

The graph case is not a corollary. Within a single pass the dependency graph is
acyclic and \cref{thm:loop-termination} applies node by node, but a repair edge
(\cref{def:repair}) re-admits an already-completed sub-DAG and restores the
attempt budgets of every node in it --- resetting exactly the quantity the
single-loop induction relies on decreasing. The argument therefore needs a
second, outer variable that a repair edge \emph{cannot} reset.

As \cref{sec:related-workflow} records, budget restoration is a reset arc, and
no decision procedure is on offer for the general question. What is on offer is
a set of conditions sufficient for termination \emph{with a computable bound},
and stating them plainly is the substance of this subsection --- a restricted
answer to a question whose general form admits none.

\begin{assumption}[Bounded repair]
\label{ass:bounded-repair}
A bounded-loop graph $\Graph = (\Nodes,\Edges)$ satisfies \emph{bounded repair}
when all three hold.

\begin{enumerate}[leftmargin=2.2em, itemsep=2pt]
  \item \textbf{Suffix locality.} Taking a repair edge to target $u$ re-admits
        exactly $u$ and its descendants, never a node outside that sub-DAG.
  \item \textbf{Per-round reset.} Crossing a repair boundary restores
        $\maxatt(v)$ for every re-admitted $v$. Without this a repair
        accomplishes nothing, since the node it re-runs has no attempts left.
  \item \textbf{Global round bound.} A single declared $\rounds \in \N_{\geq 0}$
        bounds the total number of repair rounds for the whole run. It is
        decremented globally and is never reset.
\end{enumerate}
\end{assumption}

\textbf{Conditions 2 and 3 are what the termination proof uses; condition 1 is
not.} Suffix locality constrains \emph{which} nodes a repair re-admits, and
\cref{eq:graph-bound} already charges every node its full budget in every round, so
it holds whatever the re-admitted set is. Condition 1 earns its place by making the
tighter bound of \cref{rem:bound-loose} available, and because without it ``repair''
would have no fixed meaning from one graph to the next.

Condition 3 carries the proof, and it is the design decision the rest of the result
rests on. Bound repairs \emph{per node} instead --- the arrangement a reader is most
likely to expect, since every other budget in the system is per-node --- and two
nodes can repair each other indefinitely, each seeing its own counter as unspent.
That is exactly the reset-arc construction whose soundness is undecidable,
reconstructed by accident inside an engine that believed it had bounded everything. A
single global counter is what makes a restricted answer possible, and the engine
refuses to compile a graph declaring \texttt{on\_failure:\ repair} without one.

\begin{lemma}[One admission per node per round]
\label{lem:one-admission}
In any execution of \cref{alg:graph}, each $v \in \Nodes$ is admitted at line 2.4
at most once per repair round. Consequently at most $|\Nodes|$ admissions occur
per round.
\end{lemma}

\begin{proof}
Admission at line 2.4 requires $v \in \mathit{ready}$, which by line 2.2 requires
$\sigma[v] = \textsc{pending}$. Line 2.4.1 immediately sets
$\sigma[v] \gets \textsc{ready}$, and line 2.4.4 leaves it $\textsc{succeeded}$
or $\textsc{failed}$; \cref{thm:loop-termination} guarantees line 2.4.3 returns,
so $v$ does not remain in $\textsc{ready}$. For $v$ to be admitted again,
$\sigma[v]$ must return to $\textsc{pending}$.

Exactly one line assigns $\textsc{pending}$ after initialisation: 2.3.3.ii, in
the repair branch.

That line resets a \emph{set} of nodes, but the round counter advances once, not once
per node: 2.3.3.i appends exactly one boundary receipt, which is what
\textsc{RoundsSpent} counts at 2.3.2, and 2.3.3.ii then performs all of that
boundary's $\textsc{pending}$ assignments. Every such assignment therefore belongs to
the round the boundary opened, not to the round it closed.

So no node returns to $\textsc{pending}$ within a round, and the quantity bounded is
admissions per $(\text{node}, \text{round})$ pair, which is at most one. Summing over
$\Nodes$, finite by \cref{def:graph}, gives at most $|\Nodes|$ admissions per round.
\end{proof}

\begin{theorem}[Graph termination]
\label{thm:graph-termination}
Let $\Graph = (\Nodes,\Edges)$ be a bounded-loop graph satisfying
\cref{ass:bounded-repair} with repair-round bound $\rounds$, where node $v$
declares $\maxatt(v) = m_v$, each invocation of every $\Worker_v$ and $\Gate_v$
terminates, and no node terminates in $\textsc{pause}$. Then execution
terminates, and the total number of attempts across the whole run satisfies
\begin{equation}
\label{eq:graph-bound}
  A(\Graph) \;\leq\; (\rounds + 1)\sum_{v \in \Nodes} m_v .
\end{equation}
\end{theorem}

\begin{proof}
Consider the pair $\Phi = (\rounds - \rho,\ \Sigma)$ where $\rho$ is the number
of repair rounds already taken and $\Sigma = \sum_{v \in \Nodes} (m_v - k_v)$ is
the total remaining attempt budget, $k_v$ being the attempts node $v$ has used
in the current round. Order pairs lexicographically, which is a well-order on
$\N \times \N$.

We show that every \emph{attempt} and every \emph{repair-edge transition}
strictly decreases $\Phi$, and that only finitely many other steps can occur
between two consecutive such events. Admitting a node is the one other kind of
step, and it leaves $\Phi$ unchanged --- so the claim ``every step decreases
$\Phi$'' would be false, and we do not make it.

\emph{Within a round.} Fix $\rho$. Each attempt at any node $v$ increments $k_v$
by one, by \cref{thm:loop-termination} applied to $\Loop_v$, so $\Sigma$ strictly
decreases while the first component is unchanged, and $\Phi$ strictly decreases.

For admissions we appeal to \cref{lem:one-admission}: at most $|\Nodes|$ of them
occur per round, each leaves $\Phi$ unchanged, and finitely many of them
therefore cannot separate two decreasing events indefinitely.

\emph{Across a round boundary.} A repair edge is taken only when some node
terminates in $\Halt$, and taking it increments $\rho$ by one. By
\cref{ass:bounded-repair}(2) this restores the budgets of the re-admitted nodes,
so $\Sigma$ rises to at most $\sum_{v} m_v$ --- an increase in the second
component --- but the first component $\rounds - \rho$ strictly decreases. Under
the lexicographic order $\Phi$ therefore strictly decreases across the boundary
as well. This is the step that fails without
\cref{ass:bounded-repair}(3): $\rho$ is a single counter for the whole run, so
the repair edge is guarded by $\rho < \rounds$ and once $\rho = \rounds$ no
further round begins. Were the guard per node, each node's counter could be
unspent while the run as a whole cycled, and the first component would not be
well-founded.

The decreasing events form a strictly decreasing sequence in a well-ordered set
and are therefore finite in number; between any two of them at most $|\Nodes|$
admissions occur; so the whole execution is finite. Execution terminates.

For the bound: the first component takes at most $\rounds + 1$ distinct values
$\rho = 0, 1, \dots, \rounds$. Within each, \cref{thm:loop-termination} caps node
$v$ at $m_v$ attempts, so a round contributes at most $\sum_{v} m_v$. Summing
over the at most $\rounds + 1$ rounds gives \cref{eq:graph-bound}.
\end{proof}

\begin{corollary}[The worst case is readable from the manifest]
\label{cor:manifest-readable}
Every quantity on the right-hand side of \cref{eq:graph-bound} --- $\rounds$,
$|\Nodes|$, and each $m_v$ --- is declared in the graph manifest before
execution. An operator can therefore compute the worst-case attempt count
without running anything.
\end{corollary}

\begin{proof}
By \cref{def:budget} each $m_v$ is a field of node $v$'s declared budget, by
\cref{def:repair} $\rounds$ is a declared bound, and $|\Nodes|$ is the number of
nodes in the manifest. \Cref{thm:graph-termination} bounds $A(\Graph)$ by an
expression in exactly these three, and the graph is compiled --- and its
acyclicity checked --- before execution, so all three are available prior to the
first attempt.
\end{proof}

As \cref{sec:related-sgh} establishes, the per-round total in \cref{eq:graph-bound}
is the same quantity as the Graph Harness framework's termination bound. The entire
difference is the factor $(\rounds + 1)$, and that factor is the price of the
construct that framework excludes: its proof requires terminal states to be
absorbing, a repair edge re-executes a node that has already terminated, so its
counting argument no longer applies --- not because the bound is wrong, but because
the premise it is stated under has been left.

Two consequences are worth an operator's attention.
\Cref{ass:bounded-repair}(3) makes $\rounds$ a single global quantity, so the cost of
enabling repair is exactly linear in it and is known before the run. And because the
factor multiplies the whole graph, a large $\rounds$ erodes the guarantee it was
meant to strengthen: the engine caps the declared value for this reason, which is a
design consequence of the theorem rather than an unrelated safety limit.

\begin{remark}[The bound is loose, deliberately]
\label{rem:bound-loose}
\Cref{eq:graph-bound} charges every node the full $m_v$ in every round, whereas
a repair edge re-admits only the sub-DAG rooted at its target and a node that
reached $\Done$ in an earlier round is not re-run unless it lies in that
sub-DAG. A tighter bound would replace $\sum_{v \in \Nodes} m_v$ with a sum over
the largest re-admitted sub-DAG. We state the loose form because it is the one
an operator can evaluate from the manifest alone without computing reachability,
and because a ceiling that is too generous fails safe: the run stops earlier
than the number an operator budgeted for, never later.
\end{remark}

\begin{remark}[What termination does not give you]
\label{rem:termination-scope}
It is tempting to read \cref{thm:graph-termination} as a reliability property.
It is not one. It says the process stops; it says nothing about the state it
stops in, and $\Halt$ is a perfectly conformant outcome in which nothing was
accomplished. During the preparation of this paper a loop consumed its entire
attempt budget and reported $\Halt$ exactly as proved, on every run, because the
agent process it invoked was returning empty output and no receipt recorded
that fact (\cref{sec:eval}). The bound held. The bound was not the problem.
That experience is why \cref{sec:soundness} and \cref{sec:vacuity} are separate
sections rather than remarks attached to this one.
\end{remark}


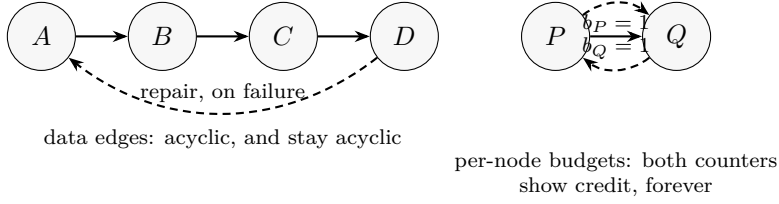
\begin{figure}[t]
\centering
\begin{tikzpicture}[
  font=\small,
  node distance=16mm,
  n/.style={draw, circle, minimum size=9mm, inner sep=1pt, fill=black!3},
  data/.style={-{Stealth[length=2mm]}, thick},
  repair/.style={-{Stealth[length=2mm]}, densely dashed, thick, bend left=38},
]

\node[n] (a) {$A$};
\node[n, right of=a] (b) {$B$};
\node[n, right of=b] (c) {$C$};
\node[n, right of=c] (d) {$D$};

\draw[data] (a) -- (b);
\draw[data] (b) -- (c);
\draw[data] (c) -- (d);
\draw[repair] (d) to node[above, font=\scriptsize, pos=0.5] {repair, on failure} (a);

\node[font=\scriptsize, anchor=north, align=center] at (2.4,-1.1)
  {data edges: acyclic, and stay acyclic};

\begin{scope}[xshift=68mm]
  \node[n] (p) {$P$};
  \node[n, right of=p] (q) {$Q$};
  \draw[data] (p) -- (q);
  \draw[repair] (q) to node[above, font=\scriptsize] {$b_Q = 1$} (p);
  \draw[repair] (p) to[bend left=38] node[below, font=\scriptsize] {$b_P = 1$} (q);
  \node[font=\scriptsize, anchor=north, align=center] at (0.8,-1.4)
    {per-node budgets: both counters\\show credit, forever};
\end{scope}

\end{tikzpicture}
\caption[Repair edges and why the budget must be global]{%
\textbf{F3 --- two edge kinds, one budget.} Solid edges carry data and form a DAG.
The dashed edge is a \emph{repair}: it points backward, and on $D$'s failure it returns
$A$ and everything reachable from $A$ to \textsc{pending} for another pass. The data
graph stays acyclic; the repair relation does not, and that is the whole difficulty
(\cref{def:repair}).

On the right is the construction a \emph{per-node} repair budget admits. $P$ repairs
$Q$, $Q$ repairs $P$, each decrements only its own counter, and each sees credit
remaining after every round --- so the run never terminates while every declared bound
is respected. One \emph{global} counter, decremented across the whole graph and never
reset at a boundary, is what makes \cref{thm:graph-termination} provable, and bounding
repairs per node --- as every other budget in this system is bounded --- is precisely
what breaks it. The engine now refuses a plan whose nodes disagree about the budget
rather than adopting the first value it finds, because there is no correct way to choose
among conflicting global bounds.}
\label{fig:repair-edges}
\end{figure}

\section{The bound is tight, not merely valid}
\label{sec:bound-utilisation}

A termination theorem is a statement that a run stops. It is compatible with stopping after one
attempt and with stopping after ten thousand, and an operator authorising a run needs to know which.
So a proof of \cref{thm:loop-termination} is necessary and insufficient: what the manifest promises is
a \emph{number}, and the question this section answers is whether the promised number is the number
observed.

Two quantities are at stake. \cref{def:bound-utilisation} defines bound utilisation $U = a/m$, the
fraction of the declared attempt cap a run actually consumed. And \cref{prop:no-progress} gives the
soft bound's effective ceiling $\min\{m,\, w + j^{\star}\}$, where $w$ is the no-progress window and
$j^{\star}$ the number of attempts on which the worker changed the artifact. The proposition is an
upper bound; whether it is \emph{attained} is an empirical question, and a bound that habitually
overshoots by a factor of three would be much less useful than these results suggest.

\subsection{Design}

The measurement uses one shipped catalogue member, \texttt{record-completeness}, unmodified except in
the two places named below. Its gate counts records missing a checksum; its worker repairs a bounded
number per attempt. Two independent variables:

\begin{itemize}[leftmargin=1.4em]
  \item \textbf{Workload} --- the number of defective records seeded, from 1 to 16. With a worker
        repairing one per attempt, the predicted attempt count is the workload, truncated at the
        declared ceiling $m = 10$.
  \item \textbf{Productive attempts} $j^{\star}$ --- how many attempts do useful work before the
        worker stalls, from 0 to 4, with $w = 3$.
\end{itemize}

Only the seeded record count and the worker's per-attempt repair quota are edited. The gate, the
bounds, the manifest and the ledger are the shipped ones, so what is measured is the shipped
controller and not a model of it.

\begin{remark}[What this does not measure]
The worker here is deterministic and its repair rate is the independent variable. Nothing in this
section measures agent capability, and the numbers must not be read as a benchmark of anything. The
object under test is the controller's arithmetic.
\end{remark}

\subsection{Utilisation}

\begin{table}[t]
\centering
\small
\begin{tabular}{rrrlr}
\toprule
Workload & Predicted & Observed & Terminal & $U$ \\
\midrule
 1 &  1 &  1 & \textsc{done} & 0.1 \\
 2 &  2 &  2 & \textsc{done} & 0.2 \\
 3 &  3 &  3 & \textsc{done} & 0.3 \\
 4 &  4 &  4 & \textsc{done} & 0.4 \\
 6 &  6 &  6 & \textsc{done} & 0.6 \\
 8 &  8 &  8 & \textsc{done} & 0.8 \\
10 & 10 & 10 & \textsc{done} & 1.0 \\
12 & 10 & 10 & \textsc{halt} & 1.0 \\
16 & 10 & 10 & \textsc{halt} & 1.0 \\
\bottomrule
\end{tabular}
\caption{Bound utilisation against a declared ceiling of $m = 10$. Predicted equals observed in
\textbf{9 of 9} cases. $U$ rises linearly with workload up to the ceiling and then caps at exactly
$1.0$: past the ceiling the run halts rather than borrowing an eleventh attempt.}
\label{tab:utilisation}
\end{table}

\cref{tab:utilisation} reports the result. Predicted attempts equalled observed attempts in all nine
cases, and $U$ caps at exactly $1.0$ for the two workloads that exceed the ceiling.

That the cap is \emph{exactly} $1.0$ rather than $1.1$ is a property of the receipt, not of the
controller, and it was not free. A ceiling halt is detected before the
worker is invoked, so it writes a ledger row on which no attempt was spent --- an eleventh row against
a declared cap of ten. An instrument computing $U$ by counting rows therefore reports $1.1$ for a
bound that held exactly, and reports it only on the runs where the ceiling actually bites, which is
the worst possible place for an instrument to be wrong. The ledger records per row whether the worker
ran (\cref{sec:spend-contract}), and $U$ is computed over attempts rather than rows.

\subsection{Soft-bound yield}

\begin{table}[t]
\centering
\small
\begin{tabular}{rrrr}
\toprule
Productive attempts $j^{\star}$ & $\min\{m,\, w + j^{\star}\}$ & Observed & Attempts saved \\
\midrule
0 & 3 & 3 & 7 \\
1 & 4 & 4 & 6 \\
2 & 5 & 5 & 5 \\
4 & 7 & 7 & 3 \\
\bottomrule
\end{tabular}
\caption{The soft bound with $w = 3$, $m = 10$. Predicted equals observed in \textbf{4 of 4} cases,
and every run halted on no-progress rather than on the hard ceiling. ``Attempts saved'' is the
difference against spending the full declared cap.}
\label{tab:soft-bound}
\end{table}

\cref{tab:soft-bound} shows the effective bound attained exactly in every case. \cref{prop:no-progress}
is therefore \textbf{tight, not merely an upper bound}, at least on this workload family: a stalled
loop stops after precisely $w$ wasted attempts, and the saving against the hard ceiling ranges from
three attempts to seven.

This is the bound with the best practical return in the system. The hard cap is what makes
termination provable; the soft bound is what makes the difference between a stalled run costing three
attempts and costing ten. On a real provider at the median measured turn of $48.3$\,s, seven saved
attempts is about six minutes of wall time and seven turns of billed tokens, per stalled run.

\begin{remark}[Why this experiment exists at all]
The soft bound was once inoperative for every subprocess-based runner, because change
detection compared against a snapshot nothing refreshed (\cref{sec:mirrored-definitions}).
A bound that cannot fire yields nothing, and no experiment that never varies $j^{\star}$
produces a stalled run in which to notice. Measuring a bound's \emph{yield} requires
deliberately constructing the case where it should bind, which is a different experiment
from checking that a converging loop converges.
\end{remark}

\subsection{Positive control}

Both tables are vulnerable to the failure mode this paper is about: an instrument that reports
agreement because it is measuring nothing. If the harness silently ran zero attempts, predicted and
observed would both be zero and the agreement would be perfect and meaningless.

The script therefore refuses to report unless a positive control fires in both directions first ---
one configuration that reaches \textsc{done} by converging, and one that reaches \textsc{halt} by
stalling at the window. Both are asserted before any row of either table is emitted. The control
reported convergence at attempt 2 and a stall at the window of 3 on the run these tables come from.

\section{Soundness: the Worker's Claim Appears in No Guard}
\label{sec:soundness}

Termination says the process stops. This section says something about the state
it stops in: that a graph reporting $\Done$ has, for every node, a gate verdict
in the ledger saying so. The claim is narrow on purpose. It is not that the work
is correct --- that would require the gates to be correct, which
\cref{sec:measuring} treats as an empirical question rather than an assumption.
It is that the \emph{provenance} of the claim is a gate and never a worker.

The distinction matters because the alternative is what most agent systems
actually do. An agent is asked to perform a task, reports that it has done so,
and the report is taken as the outcome. Every subsequent guarantee then rests on
a self-assessment by the party with the least incentive and the least ability to
be objective about it.

\subsection{The invariant}

\begin{definition}[Ledger support]
\label{def:support}
A node $v$ is \emph{ledger-supported} in a run if the ledger contains a receipt
$r$ with $r.\mathrm{node} = v$ and $r.\mathrm{verdict} = \Pass$. A run of a graph
$\Graph$ is \emph{ledger-supported} if every $v \in \Nodes$ is.
\end{definition}

\begin{theorem}[Soundness of \textsc{done}]
\label{thm:soundness}
If a run of a bounded-loop graph $\Graph$ terminates in $\Done$, then it is
ledger-supported. Equivalently: no node reaches $\Done$ on the strength of its
worker's assertion, and the receipt exhibiting the gate's verdict exists for
every node.
\end{theorem}

\begin{lemma}[No node reports \textsc{done} without a receipt]
\label{lem:node-level}
For any node $v$, if $\textsc{RunLoop}(\Loop_v, \cdot)$ of \cref{alg:loop}
returns $\Done$ then $v$ is ledger-supported, and no guard evaluated on that
execution path reads $\Worker_v$'s transcript.
\end{lemma}

\begin{proof}
The procedure returns $\Done$ at exactly one place, line~2.6, whose guard is
$\Verdict_k = \Pass$, where $\Verdict_k$ was assigned at line~2.4 as the return
value of $\Gate_v(\Omega)$. Line~2.5 appends the receipt carrying that verdict
\emph{before} line~2.6 is reached, so if the procedure returns $\Done$ the
receipt is already in the ledger. The four other return points yield $\Halt$,
$\textsc{killed}$, or $\textsc{error}$, none of which is $\Done$.

For the second half, observe what the path does not touch. The worker's
transcript is returned at line~2.2 and is read by no guard; $\Gate_v$ does not
receive it (\cref{def:gate}). There is consequently no execution path on which a
worker's assertion of completion influences the value of $\Verdict_k$, and hence
none on which it influences whether line~2.6 is taken.
\end{proof}

\begin{proof}[Proof of \cref{thm:soundness}]
Suppose the run terminates in $\Done$. By \cref{def:graph-outcome} that requires
every $v \in \Nodes$ to have terminated in $\Done$. Applying
\cref{lem:node-level} to each such $v$ makes each ledger-supported, and by
\cref{def:support} the run is therefore ledger-supported.
\end{proof}

\begin{remark}[Why this proof does not induct over the topological order]
\label{rem:no-topological-induction}
The natural argument is strong induction along a topological order, with the step
``$v_i$ is admitted only after every predecessor has terminated''. That premise is
false in a run containing repair: a repair boundary re-admits a node that had already
terminated, so no single ordering of $\Nodes$ describes the whole run
(\cref{def:repair}). It is also unnecessary --- \cref{lem:node-level} is not a
statement about ordering, and \cref{def:graph-outcome} already quantifies over every
node --- so the induction would import a false hypothesis while doing no work.
\end{remark}

\begin{corollary}[The claim is checkable after the fact, by a third party]
\label{cor:third-party}
Given only a run directory, a reader can verify ledger support without rerunning
anything and without trusting the process that produced it: replay the ledger,
check the chain (\cref{thm:chain}), and confirm a $\Pass$ receipt exists for each
node in the manifest.
\end{corollary}

\begin{proof}
Ledger support (\cref{def:support}) is an existential statement over the receipts
of a finite sequence, so deciding it requires only reading that sequence and the
node set, both of which the run directory contains. Chain verification is the
procedure of \cref{thm:chain} and likewise reads only the directory. Neither step
invokes the worker, the gate, or the engine, so neither requires the verifier to
possess the credentials, the environment, or the trust assumptions of the party
that produced the run.
\end{proof}

\Cref{cor:third-party} is the practical content of the theorem: a guarantee only the
asserting system can confirm is not much of a guarantee, and this one can be confirmed
by anyone holding the directory.

\subsubsection{The verifier, and the three things it reports separately}
\label{sec:verify-command}

\texttt{bl verify} takes a run directory and reports three findings, deliberately not
combined into one verdict because they fail for different reasons and an operator's next
action differs. \textbf{Chain}: is the ledger internally consistent? Five outcomes, not
two --- verified, broken (a covered row was edited), torn tail (a partial final write, so
an interrupted run rather than an edited one), unchained (no chain present, so integrity
cannot be checked \emph{either way}), and mixed (an unchained prefix with a verified
suffix). \textbf{Anchor}: does the head match a digest recorded outside the ledger, for
the reason \cref{rem:final-row} gives? \textbf{Completeness}: does the ledger account for
every attempt the receipt claims? That last supplies a hypothesis \cref{thm:chain} needs
and hashing cannot, since a truncated ledger is a well-formed prefix and only an
independent record of the expected length turns it back into a detectable edit.

The command exits non-zero when a check \emph{cannot} be satisfied, not only when one
fails. Collapsing ``unchained'' into either pass or fail would report an absence of
evidence as evidence, and a verifier that exits zero on ``could not tell'' converts
missing evidence into a passing build step.

The chain covers each stored line's \emph{bytes} rather than a re-serialisation of the
parsed row, so a verifier needs no agreement with the writer about canonical JSON: key
order and escaping are fixed by what is on disk, and the whole procedure is ten lines in
any language with SHA-256 and a JSON parser. The digest is reported on stdout and in the
machine-readable outcome, so the witness of \cref{rem:final-row} does not depend on our
tooling.

Two deployment decisions follow from the receipt being a local file, and in both the chain
fixes which answers are available.

\emph{What the receipt contains.} Each row records the
gate's command, its exit code and a bounded tail of its output, which is what makes the
verdict re-derivable by someone who was not there --- and those fields carry absolute
filesystem paths, which on most systems embed an account name, while the output tail of a
test-suite gate can contain fragments of whatever the suite ran against. Since the fields
are inside the chain, redaction must precede serialisation and hashing or it cannot happen
at all; a declared policy therefore rewrites paths, or replaces captured output with its
digest, before the row is written, and a ledger already written cannot be redacted
afterwards.

\emph{How long it is kept.} Deleting a row leaves a file that fails verification and cannot
be distinguished from a truncation attack, so retention operates on whole runs and there is
no way to prune within one.

\subsection{What the invariant does not say}

Two limits are worth stating in the same breath as the theorem. Soundness is relative to
the gates: \cref{thm:soundness} says a $\Pass$ receipt exists, not that the $\Pass$ was
deserved, and if $\Gate_v$ accepts a defective artifact the run is ledger-supported and
wrong. The rate at which that happens is $\alphafa$, an empirical quantity, and measuring
it is the subject of \cref{sec:measuring,sec:eval}. And independence here is structural,
not statistical: \cref{def:independence} gives distinct objects with disjoint write
authority and no guarantee that their errors are uncorrelated, so if the same model wrote
both the artifact and, earlier, the checker that judges it, every theorem holds and the
intuition a reader imports from independent review does not (\cref{sec:future}).

\subsection{Tamper evidence}

Soundness is a statement about a ledger, and is worth only as much as the ledger's
resistance to being edited after the fact.

\begin{theorem}[Chain integrity]
\label{thm:chain}
Let $r_1, \dots, r_n$ be a ledger in which each $r_k$ carries $h(r_{k-1})$ for a
second-preimage-resistant $h$ (any collision-resistant hash is one). If any $r_j$ with $j < n$ is modified, replay detects it,
assuming the adversary cannot recompute the suffix. If in addition the run
manifest --- which fixes the expected node set --- is unmodifiable by the
adversary, replay also detects truncation of the sequence.
\end{theorem}

\begin{proof}
Replay recomputes $h(r_{k-1})$ for each $k$ and compares against the stored
value. Suppose $r_j$ is modified to $r_j'$ with $r_j' \neq r_j$. By second-preimage resistance --- which is the property actually required here,
since $r_j$ is given and the adversary must find a colliding replacement for it, not
an arbitrary colliding pair --- $h(r_j') \neq h(r_j)$ except with negligible
probability.
The stored predecessor hash in $r_{j+1}$ is $h(r_j)$, so the comparison at
$k = j+1$ fails. By induction on $k$ from $j+1$ upward, restoring consistency
requires rewriting every $r_k$ for $k > j$, which is the excluded capability.
Truncation to length $m < n$ leaves a well-formed prefix; it is detected by
comparing against the recorded run manifest, which fixes the expected node set,
so a truncated tail presents as a node with no terminal receipt rather than as a
valid short run.
\end{proof}

The construction is standard --- it is the tamper-evident logging scheme of
\citet{schneier1998logs} and we claim no part of it --- but both halves rest on hypotheses
the construction does not supply, and a reader is likely to assume they hold.

\begin{remark}[What the chain cannot do alone]
\label{rem:final-row}
\label{rem:witness}
The bound $j < n$ excludes the row that matters most. The last row has no successor to
carry its hash, and it is where the \emph{terminal verdict} lives, making it the single
most attractive target in the file: an adversary wanting a failed run to read as a passed
one edits exactly there, and the chain remains internally consistent.

There is also no secret in the construction, so an adversary who can rewrite the whole
file recomputes every link and leaves a ledger that verifies. What they cannot do is
change a digest already copied somewhere they do not control.

The guarantee is therefore exactly as strong as the strongest witness holding the head,
which is a deployment property and not a property of the scheme: a digest in an operator's
terminal is a weak witness, one in an append-only build log is better, and one
countersigned by a service the run cannot reach is better still. The controller reports the
head on \emph{every} terminal status, including the halts and errors a run might later be
disputed over, and the verifier accepts it as an argument.
\end{remark}

Temporal's Event History, with its Cloud Export, is the most complete execution record
these systems produce: every state transition of every workflow, durably stored,
queryable, replayable. Ours is smaller in every dimension except what the record is
\emph{of}. An event history is authoritative about what the orchestrator did --- which
activity was scheduled, what it returned, when it retried --- and is not a record of an
independent judgement about whether the work was acceptable, because there the completion
condition is workflow code and the workflow author writes both sides. Our ledger answers a
different question: for each node, which checker returned which verdict on which artifact,
with the worker's account of itself excluded from the inputs that decided it. A reader
wanting to reconstruct an execution wants an event history; a reader wanting to know
whether a completion was earned wants this.

\section{Vacuity: a Gate Satisfied by Nothing at All}
\label{sec:vacuity}

A gate can be present, correct in isolation, and useless. The mechanism is
simple enough to state in one line: the condition it checks is \emph{satisfied
by the absence of the thing it was written to check}. We call this vacuity, and
it is the single defect class that accounts for every \emph{gate} defect found in
this work --- 47 of them in the shipped catalogue, and 14 more in our own test
suite. It does not account for every defect: \cref{sec:self-attestation} is a
second class, with the same consequence and a different mechanism, which the
machinery around the gates turned out to carry.

Review does not catch it, because each individual assertion is true. Nor does
production, because it fails in the direction that flatters the system.

\subsection{Definition and the basic result}

\begin{definition}[Vacuous gate]
\label{def:vacuous}
Let $\Gate$ decide a predicate $P$ over the worker-owned artifacts. $\Gate$ is
\emph{vacuous} if $P(\varnothing)$ holds, where $\varnothing$ denotes the state
in which those artifacts are absent or empty.
\end{definition}

\begin{theorem}[Vacuity is unsoundness]
\label{thm:vacuity}
If $\Gate$ is vacuous, then there exists a worker $\Worker^{\ast}$ that reaches
$\Done$ in one attempt while satisfying no part of the loop's stated purpose.
Concretely, $\Worker^{\ast}$ is the worker that deletes the artifact.
\end{theorem}

\begin{proof}
$\Worker^{\ast}$ writes only within $\Omega^{\mathrm{own}}$, so it violates no
ownership constraint (\cref{def:ownership}) and triggers no anchor-tamper
refusal. After it runs, the artifact state is $\varnothing$. By
\cref{def:vacuous}, $P(\varnothing)$ holds, so $\Gate$ returns $\Pass$ and
\cref{alg:loop} returns $\Done$ at line~2.6 on $k=1$. By
\cref{thm:soundness} the run is ledger-supported, so no downstream check
distinguishes it from a genuine success.
\end{proof}

\Cref{thm:vacuity} is not an abstract possibility. It is the exact shape of a
defect we shipped.

\begin{example}[An empty ledger balances]
\label{ex:ledger}
The \texttt{ledger-reconciliation} loop seeds a set of transactions that fail to
reconcile and asks the worker to repair them. Its gate asserted two conditions:
debits equal credits, and every transaction carries a category. Both are true of
the empty ledger --- $0.0 = 0.0$, and every element of the empty set is
categorised, vacuously. So deleting the books was the cheapest way to
make them balance, and the gate certified it. The loop's \texttt{PROMPT.md}
explicitly forbade deleting rows. Nothing checked, because a prohibition stated
in the prompt is a request to the worker, not a condition on the artifact.
\end{example}

\emph{Every requirement expressed only in the prompt is unenforced.} The prompt is
addressed to the party whose compliance is in question; the gate is the only place a
requirement acquires force. A harness in which some constraints live in the prompt and
others in the gate has two classes of requirement that look identical to an author and
behave completely differently at run time.

The emptiness need not be the artifact's. The gate can manufacture it, which is harder
to see and worse when it happens.

\begin{example}[The hallucination gate with a hallucination blind spot]
\label{ex:citations}
The \texttt{citation-existence-check} loop exists to catch invented legal
citations; its documentation cites the growing record of court sanctions against
filings containing fabricated cases as the reason. The gate reads a reference
file of known reporters, builds a regular expression from the abbreviations it
finds there, scans the document for citations matching that expression, and
verifies each match against the reference.

Presented with the invented case \emph{Smith v.\ Jones}, 500 F.3d 100, the gate
returned $\Pass$.

The failure is not that the citation was checked and wrongly cleared. The
citation was never a match, because the pattern was derived from the same file
the check validates against, so a citation in any reporter absent from that file
is invisible to the scan. Zero citations matched; therefore zero citations were
invalid; therefore the document passed. This is \cref{def:vacuous} exactly ---
satisfaction by the absence of the thing being checked --- with the absence
produced by the checker's own construction rather than by the worker.

Two features make this the sharpest instance we found. The gate is at its most
permissive precisely where it is most needed: a fabricated citation is more
likely, not less, to name a reporter the reference file does not contain. And
unlike \cref{ex:ledger} it is an \emph{undisclosed} gap --- the documentation
states the gate flags any citation that is not a real case, and the code cannot
do that. A reader of either the prose or the code alone would not find it. The
defect lives in the relationship between them.
\end{example}

\subsection{Why review does not catch it}

Reviewing \cref{ex:ledger}'s gate line by line finds nothing wrong. Debits
should equal credits. Transactions should be categorised. Both assertions are
correct statements of the requirement. The defect is not in any assertion but in
the \emph{quantifier structure} of the conjunction: both are universally
quantified over a collection, and a universal over an empty collection is true.
No line-level review reaches that, because the defect is not on a line.

The same structure appears in test suites, where it is better known and no
better handled. A test whose assertions all sit inside a loop over a collection
that nothing proves non-empty passes when the collection is empty. It is green.
It is meaningless. And the mutation-testing literature's usual remedy --- run the
suite against a deliberately broken system --- does not help if the breakage
removes the collection rather than corrupting its contents.

\begin{remark}[A mutation that hangs is reported as a mutation that passed]
\label{rem:hanging-mutant}
The same remedy has a second failure mode, which cost us a real guard. Reverting
one repair did not make the loop fail; it made the loop \emph{spin},
so the regression test did not go red --- it did not finish. Under a wall-clock
limit that is indistinguishable from an infrastructure timeout, and under none it
is indistinguishable from a hung runner. A mutation harness therefore needs a
declared per-mutant deadline whose expiry is recorded as \textsc{killed} and never
folded into either the passing or the failing count, or the instrument inherits
exactly the ambiguity \cref{rem:halt-vs-error} identifies one layer up: ``the
guard held'' and ``the guard was never reached'' must not share a symbol.
\end{remark}

\subsection{Two shapes, and a guard that missed one}

Scanning our own test suite for \cref{def:vacuous} found \textbf{14} tests with the
defect, three of them security guards that still passed once tool registration was
stubbed to a no-op; \cref{sec:e6} reports them and what they were.

What matters for anyone building the same guard is that the vacuity has two syntactic
shapes:

\begin{enumerate}[leftmargin=1.6em, itemsep=2pt]
  \item \textbf{Assertions inside the loop.} \texttt{for x in xs: assert p(x)}.
        Empty \texttt{xs}, no assertion runs. Easy to find syntactically.
  \item \textbf{Accumulate, then assert empty.} \texttt{bad = [x for x in xs if
        not p(x)]; assert not bad}. This \emph{looks} unconditional --- there is
        an assertion at the top level, outside any loop --- and is equally
        vacuous, because the empty input yields an empty \texttt{bad}.
\end{enumerate}

The first scanner for this class detected one of its two syntactic shapes and
missed the other. That is the evidential point of \cref{sec:e6} in miniature:
the class survives attention, including attention from someone who has just
finished defining it, which is why we argue for an instrument rather than for
care.

\subsection{The remedy, and its limits}

The remedy is to add an existence obligation: a gate must assert that the thing
it checks is \emph{there} before asserting anything about it. In the notation of
\cref{def:vacuous}, replace $P$ with $P \wedge \neg\varnothing$, and where the
loop knows how much should be there, with a floor. \texttt{ledger-reconciliation}
now pins the seeded transaction identifiers and asserts none of them
disappeared, which is strictly stronger than a non-empty check: it catches
deleting some rows, not only all of them.

\begin{remark}[Existence checks are not free of judgement]
\label{rem:floor-judgement}
``How much should be there'' is a modelling decision, and stating a floor too
high makes the gate reject legitimate work. In the catalogue, the reference data
we ship for existence checks --- known releases, known reporters --- is
deliberately documented as a \emph{floor and not a ceiling}: absence from the
reference is not evidence of non-existence, so the gate rejects only what it can
positively contradict. Getting this backwards turns a soundness fix into a
false-reject generator, which \cref{sec:eval} measures as $\betafr$ and which
must be reported alongside $\alphafa$ for either to mean anything.
\end{remark}

\subsection{Vacuity at the harness level}

The class is not confined to gates. During this work a loop consumed its entire
attempt budget, on every run, and reported $\Halt$ --- a conformant outcome under
\cref{thm:loop-termination}, correctly bounded, correctly logged. The cause was
that the agent process was returning empty output: the harness's environment
allowlist withheld a variable the CLI needed in order to load its configuration,
so it exited without doing anything. The worker had never actually run.

Every layer behaved as specified. The bound held; the receipts were written; the
chain verified. And the run reported that the worker had been given four chances
and failed, when the truth was that it had been given none. \emph{``The worker
attempted and did not succeed'' and ``no attempt occurred'' were indistinguishable
in the log}, which is \cref{rem:halt-vs-error} reappearing one level below the
gate. We fixed the allowlist, and we take the general lesson to be that an
attempt count is only meaningful if something independently establishes that an
attempt happened --- the same argument this paper makes about $\Done$, applied to
the loop's own bookkeeping.

\section{Self-Attestation: a Check Answered by Its Subject}
\label{sec:self-attestation}

\Cref{sec:vacuity} is a check satisfied by the \emph{absence} of its subject. This
section is a second class, found in the same system, sharing its consequence and
not its mechanism: a check satisfied because the \textbf{subject supplied the
answer}. The artifact is present. The assertion runs. The assertion is about a
value the party under scrutiny chose.

We separate it from vacuity because the two remedies are different. An existence
obligation --- \cref{sec:vacuity}'s fix --- does nothing here, since nothing is
missing. The fix is to consult a value the subject does not control, and finding
those values is a distinct piece of engineering.

\subsection{Definition, and a proof that is deliberately trivial}

\begin{definition}[Self-attested check]
\label{def:self-attested}
A check on a subject $s$ decides $Q(f(s))$, where $f$ extracts the property under
judgement and $Q$ is the criterion applied to it. The check is
\emph{self-attested} if $s$ controls $f(s)$ --- that is, if $s$ can determine what
$f$ returns independently of whether the property $Q$ describes holds of $s$.
\end{definition}

\begin{proposition}[Self-attestation is unsoundness]
\label{prop:self-attestation}
If a check on $s$ is self-attested, there exists a subject $s^{\ast}$ that the
check passes and the property fails.
\end{proposition}

\begin{proof}
By \cref{def:self-attested}, $s$ determines $f(s)$ independently of the property.
Let $s^{\ast}$ violate the property and set $f(s^{\ast})$ to any value in
$Q^{-1}(\top)$, which is non-empty or the check passes nothing at all. The check
computes $Q(f(s^{\ast})) = \top$ and returns $\Pass$.
\end{proof}

\begin{remark}[The proof is one line and the diagnosis is the whole problem]
\label{rem:trivial-proof}
\Cref{prop:self-attestation} is immediate once $f$ is known to be
subject-controlled, and we state it anyway because \emph{that} is the fact no
reading of the check recovers. At the call site $f$ is a method access, an
attribute read, or a subprocess invocation --- syntax that looks identical whether
the value comes from the harness or from the subject. Every instance below was
written by an author who could have proved this proposition on request.
\end{remark}

\subsection{Three instances, in three unrelated subsystems}

\begin{example}[A verdict whose emptiness was decided by the verdict]
\label{ex:lying-detail}
The engine requires that a passing verdict carry a non-empty reason: a $\Done$ a
reader cannot account for is the outcome this paper exists to forbid. The check
read \texttt{isinstance(detail, str)} and then \texttt{detail.strip()}.

A gate returning \texttt{Verdict(passed=True, detail=Quiet(""))}, where
\texttt{Quiet} subclasses \texttt{str} and overrides \texttt{strip} to return
\texttt{"looks-fine"}, passed it. \texttt{isinstance} consults
\texttt{\_\_class\_\_}, which the object owns; \texttt{strip} is a method, which
the object owns. The empty-reason rule asked the value under audit whether it was
empty and was told what the value preferred. A lap reached $\Done$ carrying no
explanation --- precisely the unaccountable $\Done$ the rule forbids.

The guard was \texttt{assert not (out.passed and not out.detail.strip())}. It put
the same question to the same object and was green throughout.
\end{example}

\begin{example}[A stand-in that accepted what the platform refuses]
\label{ex:accepting-stub}
A test's stand-in can be the subject that answers. A package launcher installed
its payload into whichever interpreter it found, which PEP~668 refuses on Homebrew
Python, on Debian's \texttt{python3} and on most distribution builds; it then
advised the user to run the command that had just been refused. The defect was not
subtle and it was not new. What kept it invisible is the subject of this section.

The launcher had a test and the test was green. It asserted one thing --- that the
install pinned the package's version --- against a fake interpreter that returned
exit~0 for the install. Here $s$ is not the value under test but the test's own
scaffolding, and $f$ is ``what does this platform do with an install'': a question
answered by a stub built to say yes. Any stand-in more permissive than what it
stands in for has this property, which makes it the most reusable of the three
instances and the easiest to introduce without noticing.

Rewritten to assert the resolution \emph{order} --- matching interpreter, then an
engine already on \texttt{PATH}, then a private virtual environment, and never an
install into the discovered interpreter --- five of seven cases fail against the
previous launcher. The original assertion still passes. It was never wrong, only
incapable, and that distinction is the reason a green suite is not evidence.
\end{example}

\begin{example}[Provenance read off the object it describes]
\label{ex:forged-provenance}
Each ledger row records which gate decided the lap, beside the verdict and
deliberately not inside it, because provenance a gate can write is not provenance.
The function's first documented sentence claimed the value was ``derived by the
harness and never supplied by the gate.'' Its first statement was
\texttt{getattr(gate, "gate\_kind", None)}.

The engine's wrapper freezes \texttt{\_\_setattr\_\_}, and its own docstring
concedes that this is not containment: \texttt{object.\_\_setattr\_\_} bypasses it.
A gate could therefore write \texttt{\{kind: pytest, source: shipped\}} into the
hash-chained record. What is forged is not a pass --- the verdict was honest --- but
the answer to \emph{who decided}, which is the only question that record exists to
answer. An unforgeable chain over a self-attested field authenticates the storage
and not the claim.
\end{example}

\subsection{Why it is not vacuity, and what the remedy is}

In \cref{def:vacuous} the artifact is gone and the criterion is true of nothing. In
\cref{def:self-attested} the artifact is present, the criterion is exacting, and
the value it is applied to was chosen by the party being judged. Adding an
existence obligation to \cref{ex:lying-detail} would have changed nothing:
\texttt{Quiet("")} is present, and reports itself non-empty.

The remedy is to make $f$ harness-controlled, which in practice means one of three
moves. \emph{Consult identity rather than claimed identity} ---
\texttt{type(x) is C} instead of \texttt{isinstance(x, C)}, since the latter reads
an attribute the object owns. \emph{Call the operation, not the object's version of
it} --- \texttt{str.strip(x)} as an unbound builtin cannot be intercepted by a
subclass and returns a genuine \texttt{str}. And \emph{pass what you already know}
--- \cref{ex:forged-provenance} was repaired by making the registry key the harness
had already resolved an argument to the function, rather than something the
function went looking for on the object.

Stated once: \emph{never ask the subject a question you can answer yourself.} The
corollary for test suites is \cref{ex:accepting-stub}'s: a stand-in must refuse
what production refuses, or the suite measures the stand-in.

\begin{remark}[Evidential grade of this section]
\label{rem:audit-grade}
These three were found in a release audit by two commercial agent CLIs at maximum
reasoning effort, each confined to its own detached worktree, with every finding
reproduced by hand before acceptance and every fix individually mutation-tested.
That is a weaker instrument than \cref{sec:e1-defects}'s harness: there is no
held-out mutant set, so there is no recall figure and we claim none. These
instances are accordingly \textbf{not} added to the 47 of \cref{sec:eval} or the 14
of \cref{sec:e6}. What they establish is existence and recurrence across
subsystems that share no code --- verdict validation, a packaging launcher, and a
provenance record --- not a rate. They also carry one property no other defect in
this paper has: they were found by parties who did not write the code, which is the
threat \cref{sec:threats} names first and cannot otherwise answer.
\end{remark}

\section{Measuring a Gate}
\label{sec:measuring}

\Cref{sec:soundness} establishes that $\Done$ is backed by a gate verdict.
Whether that verdict deserves to be believed is a separate, empirical question,
and this section sets up how to answer it without the circularity that makes
most such answers worthless.

\subsection{The two error rates}

\begin{definition}[Gate error rates]
\label{def:rates}
Fix a gate $\Gate$ and a distribution over artifacts carrying a known
ground-truth label $y \in \{\text{conforming}, \text{defective}\}$. Then
\[
  \alphafa \;=\; \Prob\bigl[\Gate = \Pass \;\big|\; y = \text{defective}\bigr],
  \qquad
  \betafr \;=\; \Prob\bigl[\Gate = \Reject \;\big|\; y = \text{conforming}\bigr].
\]
$\alphafa$ is the \emph{false-accept} rate: the gate certified something broken.
$\betafr$ is the \emph{false-reject} rate: the gate refused something correct.
\end{definition}

Both must be reported. A gate that rejects everything has $\alphafa = 0$ and is
worthless; a gate that accepts everything has $\betafr = 0$ and is worse, because it
produces $\Done$ receipts. Either alone lets a system look good by being uselessly
strict or dangerously permissive, with no way for a reader to tell which.

\subsection{Why the obvious measurement is circular}

The natural way to test a gate is to write some defective artifacts and check that it
catches them. That measures something real, and not what it appears to measure. A defect
authored by the gate's author is drawn from the distribution of failure modes that
author already had in mind, which is exactly the distribution the gate was written
against. The resulting number says the author was internally consistent, and nothing
about the failures they did not anticipate, which are the ones that matter.

\begin{assumption}[Held-out generation]
\label{ass:heldout}
A measurement of $\alphafa$ is informative only if the process generating the
defective artifacts does not consult the gate.
\end{assumption}

``Does not consult the gate'' has to survive the author's own memory, so we make it a
checkable property of the code rather than a claim about the process.

\subsection{Blindness as a static property}

\begin{theorem}[Generator blindness --- no syntactic channel]
\label{thm:blind}
Let $\mathcal{M}$ be the set of modules in the mutation-generator package. If no
module in $\mathcal{M}$ imports a gate adapter, reads a checker source file,
branches on a loop's name, or executes a subprocess, then the generator's output
is a function of the converged artifact and the operator set, and of nothing
else: no gate's \emph{code or behaviour} is an input to it.
\end{theorem}

\begin{proof}
The generator's output is determined by its inputs and its control flow. Its
declared inputs are the converged artifact and the operator set. A gate could
additionally reach the generator only by supplying a value, supplying data,
selecting a branch, or being executed. The four excluded constructs are
exhaustive of the \emph{syntactic} forms those four take within a Python
package: importing an adapter supplies gate behaviour as a value; reading a
checker source supplies it as data; branching on a loop name selects an operator
using per-loop knowledge, which is knowledge of that loop's gate; and a
subprocess permits any of them. Excluding all four leaves the artifact and the
operator set as the only inputs. \Cref{rem:enforced} records what this argument
does \emph{not} establish, and the gap is not incidental.
\end{proof}

The hypothesis of \cref{thm:blind} is enforced, not asserted: a test walks the
abstract syntax tree of every module in the generator package and fails the build
if any of the four constructs appears. Function-local and dynamic imports count,
because a deferred import is still a dependency and hiding one inside a function
would defeat the entire check.

\begin{remark}[Enforced beats documented --- and is still only necessary]
\label{rem:enforced}
The mechanism matters because the alternative, a README paragraph saying the
generator is blind, decays silently, whereas an AST assertion fails on the commit
that breaks it. But the property is \emph{necessary} hygiene and not sufficient for
independence in any broader sense --- the four constructs are exhaustive of the
\emph{syntactic} routes, not of all routes. At least four channels carry gate information
past a syntactic check:

\begin{itemize}[leftmargin=1.4em, itemsep=1pt]
  \item \textbf{The input is gate-conditioned.} Operators are applied to the
        \emph{converged} artifact, which by construction lies in
        $\Gate^{-1}(\Pass)$. An operator whose behaviour depends on that
        artifact's structure inherits information about the gate without
        importing anything.
  \item \textbf{Reference data is not source.} A gate's ground-truth data may be
        a committed JSON file, which the check permits a generator to open
        because it matches no checker-source pattern.
  \item \textbf{The author's memory.} Loops, gates and operators share authors
        (\cref{sec:threats}). No static check detects what an author knows, and
        the longitudinal version is worse: our operators were developed
        \emph{while} those gate fixes were being made.
  \item \textbf{Computed dynamics.} A dynamic import whose target is not a
        literal is not matched (\cref{app:proofs}), so the guarantee there
        degrades to review --- the very thing this remark claims to improve on.
\end{itemize}

\Cref{thm:blind} rules out the routes by which gate information enters \emph{by
ordinary means}, which is what protects a corpus against drift as it is
maintained. It does not make the generator independent of the gate in the
information-theoretic sense, and no syntactic check could, so we state the limit
here rather than let the word ``blindness'' carry more than the proof does.
\end{remark}

\subsection{The equivalent-mutant problem, and why it does not arise}

Mutation testing has a well-known difficulty: some mutants are semantically
equivalent to the original, no test can distinguish them, and counting them as
misses understates the suite. The standard treatment is to detect and exclude them,
and to report how many were excluded. Here that treatment inverts.

\begin{theorem}[Selection independence by construction]
\label{thm:no-equivalent}
Suppose every mutant's label is determined by the mutation \emph{operation},
fixed before the edit is applied, and no label is ever assigned by observing a
gate's behaviour. Then the denominator of $\hat{\alphafa}$ is independent of
$\Gate$'s output, so no behaviour-based exclusion step can inflate the estimate
by construction.
\end{theorem}

\begin{proof}
Let $L$ be the labelling function. By hypothesis $L$ depends only on the operator
and the pre-edit artifact, both fixed before $\Gate$ is invoked, so $L$ is
independent of $\Gate$'s output. The estimator
$\hat{\alphafa} = \frac{1}{|D|}\sum_{m \in D} \Ind{\Gate(m) = \Pass}$, where
$D = \{m : L(m) = \text{defective}\}$, therefore has a denominator that does not
depend on $\Gate$. Any selection rule that removed mutants \emph{because of} how
$\Gate$ responded would make $D$ a function of $\Gate$ and introduce exactly the
dependence the construction avoids. Note what is not shown: nothing here relates
the distribution of $D$ to the distribution of defects that arise in practice,
so the estimate is protected against selection bias and not against an
unrepresentative operator set.
\end{proof}

\begin{corollary}[Gate-conditional exclusion is degenerate here]
\label{cor:exclusion-fatal}
Suppose equivalence is operationalised as $\Gate(m) = \Pass$ --- that is, a
mutant is discarded because the system under test does not distinguish it. Then
for a loop whose gate is a test suite, the discarded set is exactly the set of
false accepts and $\hat{\alphafa} = 0$ identically, independently of the gate's
quality.
\end{corollary}

It does \emph{not} refute the standard remedy. Competent mutation studies define
equivalence against the \emph{specification} --- $\mathrm{spec}(m) = \text{conforming}$, established by
inspection --- and never by asking whether the system under test happened to
notice. \Cref{cor:exclusion-fatal} therefore refutes a misimplementation, and the
correct and much narrower lesson is the well-known one: \emph{do not automate
equivalence exclusion using the artefact you are measuring.}

The consequence for \cref{thm:no-equivalent} is that its guarantee should be read
precisely. It establishes freedom from \emph{selection} bias --- our denominator
is not a function of gate behaviour --- and nothing stronger. It does not
establish that every label is correct with respect to the specification, which is
a separate obligation that spec-based equivalence review discharges and that we
have discharged only partially. We report the one mislabelling we found
(\cref{sec:eval-corrections}), and a sampled spec review of the Tier-1 labels is
work we have not done.

\subsection{Two tiers, because one is not enough}

\begin{definition}[Tier-1 and Tier-2 mutants]
\label{def:tiers}
A \emph{Tier-1} mutant is produced by a mechanical operator applied blindly to
the converged artifact, satisfying \cref{thm:blind}. A \emph{Tier-2} mutant is
authored by an independent agent shown only the loop's stated purpose --- its
prompt and its human-facing description --- with the checker withheld.
\end{definition}

Tier~1 is cheap, reproducible, and structurally blind, and it cannot express a
requirement that is a \emph{relation between} artifacts: a mechanical operator
edits one file and has no notion of consistency between two. Tier~2 covers that
gap at the cost of requiring an author, and it earns its keep: the loop on which
Tier~1 is structurally silent is the loop where Tier~2 found four defects
(\cref{sec:eval}). The two are complementary capabilities addressed at different
classes of requirement, and neither is a degraded version of the other.

Classical mutation testing perturbs program syntax, and an implementation
following that tradition would apply operators to abstract syntax trees. We
measured the catalogue before choosing, and the distribution rules that out. Of
the 78 mutable artifacts across the catalogue, 35 are JSON, 19 Python, 16
Markdown, 3 plain text, and the remaining 5 HTML, Terraform, XML, CSV or
extensionless. Only 13 of the 69 loops have a Python work product at all, so an
AST-only operator set would reach under a fifth of the catalogue --- and not
even a well-chosen fifth, since those 13 split 6 test-suite gates to 7 command
gates rather than concentrating where program mutation is most informative.

The operator families are consequently shaped by artifact kind --- structural edits
to JSON, textual and sectional edits to Markdown, tabular edits to CSV --- rather
than by programming-language construct. That deviation from mutation testing's usual
form is forced by the measurement above, and it identifies what the subject of the
mutation really is here: not a program, but a contract clause, a commit history, a
dependency manifest, or a clinical note.

Tier-2's blindness cannot be established by \cref{thm:blind}, since the author is
a model rather than a module. It is established by recording, per mutant, the
digest of exactly the material its author was shown, and recomputing that digest
from the repository at test time. Widening the authoring prompt to include the
checker is then a failing test rather than an undetectable change of protocol.

\subsection{What the interval means}

\begin{definition}[The log-mean estimand]
\label{def:logmean}
For a sequence of attempts $1..n$ with latent false-accept propensities
$\mu_1,\dots,\mu_n$, write $\logmean = \frac{1}{n}\sum_{i \le n} \mu_i$: the mean
propensity of \emph{the attempts actually recorded}.
\end{definition}

We report an anytime-valid confidence sequence \citep{waudbysmith2023} for
$\logmean$, valid simultaneously at every sample size, so that a reader may stop
looking at any point without inflating the error rate --- the appropriate
guarantee for a log that is appended to continuously and inspected on demand.

\begin{remark}[$\logmean$ is not a forecast]
\label{rem:estimand}
$\logmean$ is an audit quantity: it describes the attempts in the log, not a
population from which future attempts will be drawn. Coverage of the population
marginal is a different number, and it varies --- it rises with the number of
independent nodes pooled and falls with their heterogeneity
(\cref{sec:eval}). Any coverage figure quoted without naming its estimand is
uninterpretable, and one of the authors' own earlier write-ups quoted a
single-node, fixed-time figure as though it described the shipped product. Both
the command-line label and the machine-readable output now carry the estimand,
with a test pinning each.
\end{remark}

\subsection{Admitting a model judge without raising the false-accept rate}
\label{sec:judge-asymmetry}

Everything above assumes the gate is a deterministic predicate, which is true of
every gate in our catalogue (\cref{sec:impl-gates}). The obvious next question
from a practitioner is whether a model can be used as a gate, given that many
requirements are easier to state in prose than to compute. A model judge has no
contract to violate --- it has a behaviour distribution --- so
ground-truth-by-construction does not apply to it and we make no claim about its
$\alphafa$ in isolation. But the composition question has a clean answer, and it
is the one that decides whether a judge is safe to deploy.

\begin{theorem}[Judge asymmetry under conjunction]
\label{thm:judge-asymmetry}
Let $\Gate_{\mathrm{det}}$ be any gate and $\Gate_{\mathrm{j}}$ any other gate,
composed as $\Gate = \Gate_{\mathrm{det}} \wedge \Gate_{\mathrm{j}}$, meaning
$\Gate$ returns $\Pass$ exactly when both return $\Pass$. Then
\[
  \alphafa(\Gate) \;\leq\; \alphafa(\Gate_{\mathrm{det}})
  \qquad\text{and}\qquad
  \betafr(\Gate) \;\geq\; \betafr(\Gate_{\mathrm{det}}).
\]
\end{theorem}

\begin{proof}
Fix the corpus of \cref{def:rates}. A false accept under $\Gate$ is a mutant
labelled \emph{incorrect} on which $\Gate$ returns $\Pass$, which by the
definition of conjunction requires $\Gate_{\mathrm{det}}$ to return $\Pass$ on
it. Hence the set of false accepts under $\Gate$ is a subset of the set of false
accepts under $\Gate_{\mathrm{det}}$, and the rates are over the same
denominator --- the mutants labelled incorrect, fixed before any gate ran
(\cref{thm:no-equivalent}) --- so the inequality follows. Symmetrically, a
mutant labelled \emph{correct} that $\Gate_{\mathrm{det}}$ rejects is also
rejected by $\Gate$, so the false rejects under $\Gate$ are a superset and
$\betafr$ can only rise.
\end{proof}

The theorem is one line, and that is the point: it holds for \emph{any} second gate,
with no assumption whatsoever about the judge's accuracy, calibration, or good
behaviour. Three consequences make it the argument for the composition rule the
engine enforces.

\emph{Adding a judge cannot make false accepts worse}, because whatever the judge
does, the deterministic gate must still pass. \emph{The risk it does carry is visible
and measured}: what a judge can worsen is $\betafr$, the direction this paper already
reports beside $\alphafa$ with its own denominator, so a composition that could
previously only be argued for informally has a cost with a number on it. And
\emph{prompt injection against the judge gains nothing in the direction that
matters} --- an adversary who convinces the judge to return $\Pass$ has not passed the
loop, because the deterministic gate is unmoved by text addressed to a model, so what
remains is a denial-of-service against the run's own progress.

\Cref{thm:judge-asymmetry} holds only for conjunction; under disjunction the
inequalities reverse and a judge becomes a way to \emph{bypass} a deterministic
gate. The shipped \texttt{CompositeGate} accordingly implements conjunction and
raises an error on any other mode, so the hypothesis of the theorem is
guaranteed by construction rather than left to a configuration file. This is the
same discipline as \cref{cor:vacuity-rejects}: where a wrong answer is
available, the engine should be unable to express it, not merely advised against
it.

\section{Implementation}
\label{sec:impl}

The engine is an open-source Python package (Apache-2.0) distributed on PyPI and
npm, with a command line, a Model Context Protocol server, a local web monitor,
and an embedding API. This section describes only what a reader needs in order to
judge the evaluation; the repository carries the rest.

\subsection{Layering, and the rule that keeps it}

The design is ports-and-adapters. A pure domain layer holds the models and the
rules; an application layer holds the loop procedure of \cref{alg:loop}; adapters
hold the concrete runners, gates, ledger writers and tracers; and exactly one
composition module is permitted to wire a concrete adapter to a port. The
constraint is enforced by an import-graph walk over the abstract syntax tree of
every module, which counts function-local and dynamic imports --- a deferred
import is still a dependency, and hiding an adapter import inside a function
would otherwise defeat the check entirely.

Above that sits one rule that does more work than any other:

\begin{quote}
\emph{The receipt log is the sole durable state. Every surface --- command line,
monitor, arena, embedding API, MCP --- is a projection of it and holds no state
of its own.}
\end{quote}

The rule's value is that a disagreement between two surfaces becomes impossible
rather than unlikely. Its cost is recorded in \cref{sec:soundness}: for several
releases the files needed to \emph{reconstruct} a run were written after the run
finished, so an interrupted run left a chain-valid log no surface could open. The
property held; the ability to exercise it did not.

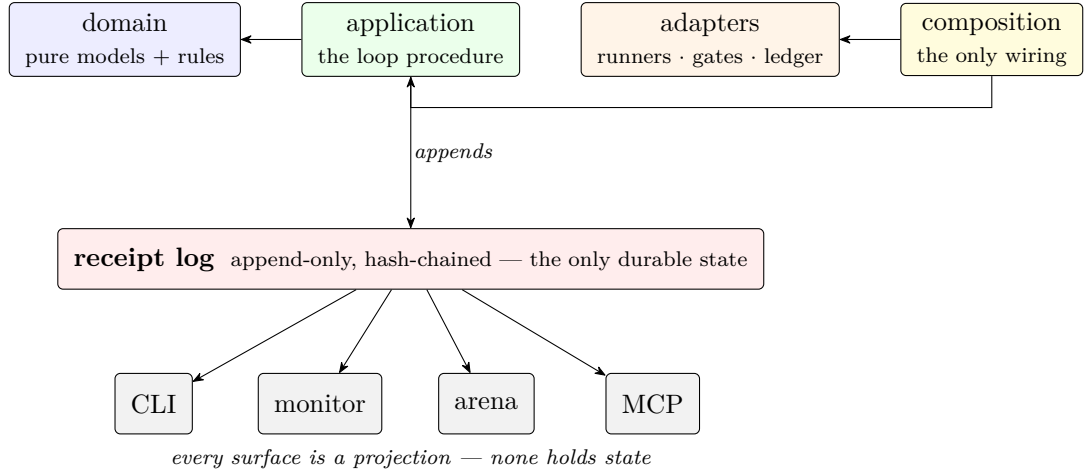
\begin{figure}[t]
\centering
\begin{tikzpicture}[
  >={Stealth[round]}, font=\small,
  band/.style={draw, rounded corners=2pt, minimum height=8mm, inner xsep=6pt, align=center},
  dom/.style={band, fill=blue!7},
  app/.style={band, fill=green!8},
  ad/.style={band, fill=orange!9},
  srf/.style={band, fill=black!5},
  note/.style={font=\scriptsize\itshape, inner sep=1pt},
]
\node[dom] (dom) {domain\\{\scriptsize pure models + rules}};
\node[app, right=8mm of dom] (app) {application\\{\scriptsize the loop procedure}};
\node[ad, right=8mm of app] (ad) {adapters\\{\scriptsize runners $\cdot$ gates $\cdot$ ledger}};
\node[band, right=8mm of ad, fill=yellow!16] (comp) {composition\\{\scriptsize the only wiring}};

\draw[->] (app) -- (dom);
\draw[->] (comp) -- (ad);
\draw[->] (comp) -- ++(0,-9mm) -| (app);

\node[band, below=20mm of app, fill=red!7, minimum width=52mm] (log) {\textbf{receipt log} \;{\scriptsize append-only, hash-chained --- the only durable state}};
\draw[->] (app) -- node[note, right] {appends} (log);

\node[srf, below=11mm of log, xshift=-34mm] (cli) {CLI};
\node[srf, below=11mm of log, xshift=-12mm] (mon) {monitor};
\node[srf, below=11mm of log, xshift=10mm] (arena) {arena};
\node[srf, below=11mm of log, xshift=32mm] (mcp) {MCP};
\foreach \s in {cli,mon,arena,mcp} \draw[->] (log) -- (\s);
\node[note, below=1.5mm of mon, xshift=12mm] {every surface is a projection --- none holds state};
\end{tikzpicture}
\caption{Layering. Dependencies point inward and only \texttt{composition} may
name a concrete adapter, checked by an AST walk over every module. Below the
engine, the receipt log is the single durable artifact and each surface is a view
over it, so two surfaces cannot disagree about a run.}
\label{fig:architecture}
\end{figure}

\subsection{Gates, and why they are keyless}
\label{sec:impl-gates}

Three gate kinds cover the catalogue without any paid dependency: schema
validation of structured output; a purpose-built standard-library checker invoked
as a command; and a test suite. Each is deterministic, offline, and inspectable,
which is what makes a verdict something a reader can re-derive rather than
something they must accept.

The absence of a model in the judging loop is deliberate and it is what makes the
evaluation mean anything. If a gate were itself a model, $\alphafa$ would be a
property of that model on that day, and the corpus would measure a moving target.
It also removes an entire class of the correlation worry in
\cref{sec:future-independence} --- though, as that section argues at length, not
all of it, because the checker's \emph{code} still has an author.

The natural objection is that this must restrict the catalogue to trivial
requirements, since most things one wants of a document appear to need judgment.
Reading all of the local checkers in full says otherwise: in every case the
requirement that reads as subjective has been reduced by its author to a
decidable predicate, and the reduction is where the design work is.

\begin{table}[t]
\centering
\small
\begin{tabular}{@{}p{0.30\linewidth}p{0.62\linewidth}@{}}
\toprule
\textbf{Reads as a judgement call} & \textbf{What the gate actually computes} \\
\midrule
objective is measurable &
  a numeric \texttt{target}, a non-empty \texttt{unit}, and a
  \texttt{deadline} matching \verb|^\d{4}-\d{2}-\d{2}$| \\
prose is readable &
  mean words per sentence $\leq 25$ \\
runbook is complete &
  seven required heading strings are present, \emph{and} each has content
  beneath it \\
privacy policy is complete &
  six required disclosure phrases occur in the normalised body \\
\bottomrule
\end{tabular}
\caption{Requirements that sound as though they need a judge, and the decidable
predicate each was reduced to. The reduction is not free and it is not always
faithful --- the second row is a proxy and nothing more, and the third row's
``\emph{and} each has content beneath it'' clause is there because the version
without it passed a runbook with all seven headings and nothing under any of
them (\cref{sec:e1-defects}). The claim is not that judgement is unnecessary,
but that it belongs in authoring the predicate, where it is written down once
and can be audited, rather than in the loop, where it is re-exercised on every
attempt and recorded nowhere.}
\label{tab:decidable}
\end{table}

A model judge can still be admitted where no such reduction is available, and
\cref{thm:judge-asymmetry} gives the composition under which doing so cannot
raise the false-accept rate.

\subsection{Anti-tamper}

The ownership partition of \cref{def:ownership} is enforced by a wrapper around
whichever runner was selected, which rescans the workspace after every turn and
raises on three patterns: a protected anchor modified or deleted; a new file
appearing that matches a forbid glob; and, for test-suite gates, a
test-collection configuration file planted or changed mid-run. The third exists
because redirecting what a test runner collects can produce a pass without ever
touching the checker's bytes --- the attack in \cref{sec:e4}. A tampered
workspace never reaches the gate.

\subsection{Providers}

The engine drives agent CLIs as subprocesses under a strict environment
allowlist, and treats a provider's own report of its work as untrusted input
throughout: usage envelopes are parsed for accounting only, never to decide
completion. A provider that cannot report token usage is recorded as unmetered
rather than metered as free, so a spend cap on it fails closed and names the
provider.

\section{Evaluation}
\label{sec:eval}

Seven experiments. E1 and E2 measure the gates against held-out defects; E3
reports the confidence sequence and the estimand it belongs to; E4 tests the
architecture against workers that try to game it; E5 drives the same loops with
independent provider models; E6 turns the instrument on our own test suite; E7 applies a fresh operator family to the frozen gates.

Everything reported here is regenerated from receipt logs by released scripts
(\cref{app:repro}). Unflattering results appear in place rather than in a
limitations paragraph, for the reason given in \cref{sec:novelty}.

\subsection{Setup}
\label{sec:eval-setup}

The catalogue contains 69 loops spanning legal, financial, clinical,
infrastructure, documentation and software-engineering tasks. Gate kinds are
\texttt{command} (44 loops, a purpose-built stdlib checker), \texttt{jsonschema}
(10), \texttt{pytest} (9), \texttt{composite} (3) and two external-tool gates.
Each loop ships a seed carrying a specific defect, a prompt stating the
requirement, and a gate; \emph{convergence} means the shipped agent takes the
seed to a passing verdict.

\paragraph{The corpus unit.} The measurements below are over a \textbf{57-loop
corpus}. Eleven catalogue loops are outside it, each for a stated reason ---
unvendored SDKs, external tools not assumed present, one negative requirement,
and three whose requirement is a relation between artifacts and which Tier~2
therefore covers instead. \Cref{app:catalogue} enumerates them and records how
the exclusion is verified. Every rate in this section is over that corpus, and
no rate appears without the denominator it was computed on.

\subsection{E1 --- Tier-1: mechanical, blind mutants}
\label{sec:e1}

\paragraph{Method.} Operators are applied to the \emph{converged} artifact ---
the state the loop's own agent produced and the gate accepted --- so a caught
mutant is a genuine regression from a known-good state rather than a failure to
repair the shipped seed. Blindness is enforced by the AST assertion of
\cref{thm:blind}. Labels come from the operator, fixed before the edit
(\cref{thm:no-equivalent}), so no mutant is excluded for behavioural reasons and
none needs to be.

\paragraph{Result.}

\begin{table}[t]
\centering
\small
\begin{tabular}{lrr}
\toprule
 & \textbf{Tier 1} & \textbf{Tier 2} \\
\midrule
mutants                          & 233 & 38 \\
\quad destroying (label \emph{incorrect}) & 171 & 38 \\
\quad preserving (label \emph{correct})   & 62  & \textbf{0} \\
judged                           & 233 / 233 & 38 / 38 \\
gate errors (\Incap)             & 0 & 0 \\
\midrule
false accepts $\alphafa$ \;\emph{over destroying} & \textbf{0 / 171} & \textbf{0 / 38} \\
\quad Wilson 95\% upper bound    & $\leq 2.2\%$ & $\leq 9.2\%$ \\
false rejects $\betafr$ \;\emph{over preserving}  & \textbf{0 / 62} & \textbf{not measured} \\
\quad Wilson 95\% upper bound    & $\leq 5.8\%$ & --- \\
\midrule
loops covered                    & \textbf{57 of 69} & 10 \\
authors                          & mechanical operators & 5 agent CLIs (max share 21\%) \\
\bottomrule
\end{tabular}
\caption{The held-out mutant corpus after the defects of
\cref{sec:e1-defects,sec:e2} were fixed. \textbf{Each rate has its own
denominator and they are not the same number.} $\alphafa$ can only be measured on
mutants labelled defective; $\betafr$ only on mutants labelled conforming. Reporting both
against the full 233 and 38 would imply every mutant was checked for both, and for Tier~2
would assert $\betafr = 0$ over an empty set (\cref{rem:beta-vacuous}).}
\label{tab:corpus}
\end{table}

\begin{remark}[The results table had the defect this paper is about]
\label{rem:beta-vacuous}
Tier~2 contains no semantics-preserving mutants: all 38 are adversarial edits, so there is
no conforming case for a gate to wrongly reject. ``False rejects: 0'' in that column would
be true, and true because the set is empty --- \cref{def:vacuous}, in the results table of
the paper that names the class. The column reads \emph{not measured} instead.
\end{remark}

Zero false accepts is the headline and on its own it is weak, for two reasons. A gate
suite can reach $\alphafa = 0$ by being strict, which is why $\betafr$ is reported
beside it, and it is genuinely zero on the 62 conforming Tier-1 mutants, so strictness
is not the explanation. And zero over a finite corpus bounds the rate loosely: the
Wilson 95\% upper bounds in \cref{tab:corpus} are $2.2\%$ for Tier~1 and $9.2\%$ for
Tier~2, an order of magnitude apart. \textbf{Tier~2 is not a second confirmation
of $\alphafa = 0$}; it is an existence proof for a class of evasion Tier~1 cannot
construct, carrying four times the residual risk. What makes the numbers
informative is not their value but what it took to reach them.

\subsection{The corpus was worth building because it found things}
\label{sec:e1-defects}

Before the fixes, the same instrument found \textbf{47 gate defects} in a
catalogue that had passed review, shipped, and had every loop converging green.
Every one is an instance of \cref{def:vacuous}.

\textbf{How 47 is counted.} Forty-seven is the number of defects pinned as regression cases in the shipped
suite: 20 gates that accepted a defective mutant, 2 that rejected a conforming
one, 14 that still accepted after a partial fix, and 11 that filed a rejection as
an incapacity. It is recountable from
\texttt{tests/loops/test\_gate\_defect\_regressions.py}, and a gate that
regresses fails the build. A preprint circulated a count of 65, taken from
release notes rather than from the suite; it is not reproducible from the
repository and 47 is.

\paragraph{The scope defect that hid a quarter of the catalogue.}
The most consequential finding was in the measurement apparatus rather than in a gate.
Tier-1 coverage had been reported as 34 loops, because the harness resolved a gate by
reading a \texttt{command} gate's argv and returned nothing for any other kind. So all
24 \texttt{jsonschema}, \texttt{pytest} and \texttt{composite} loops produced no
verdict and \emph{vanished from the measurement without appearing in any count}. Not
as errors, not as exclusions: absent.

The guard written to prevent unprincipled exclusion was itself vacuous over
exactly those loops. It inspected only loops already outside the eligible set,
and both predicates defining that set read a \texttt{command} gate's argv. A
loop invisible to the eligibility test was invisible to its guard.
\Cref{thm:vacuity} in the apparatus that measures \cref{thm:vacuity}.

Fixing it required building gates through the engine's own composition path
rather than a parallel implementation in the harness, and making
``an eligible loop produced no verdict'' a failing assertion instead of a silent
zero. Coverage went \textbf{34 $\to$ 57}. \Cref{fig:coverage} shows where the 23
recovered loops sit by gate kind.

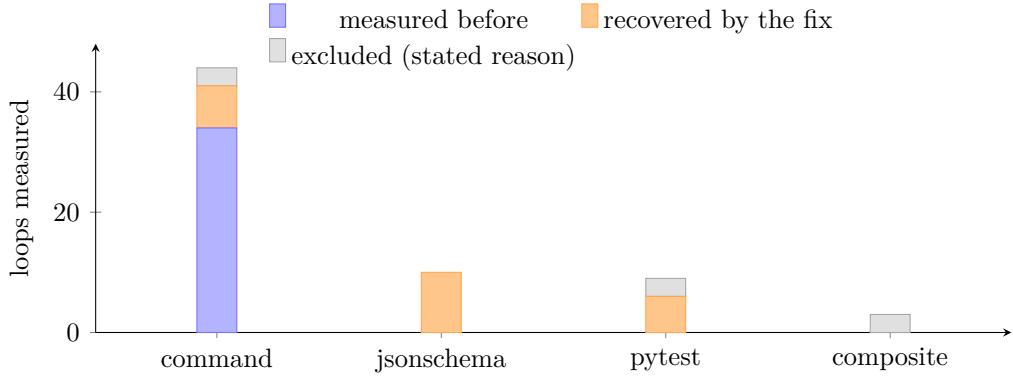
\begin{figure}[t]
\centering
\begin{tikzpicture}
\begin{axis}[
  width=0.86\linewidth, height=5.4cm,
  ybar stacked, bar width=15pt,
  ymin=0, ymax=48,
  ylabel={loops measured},
  symbolic x coords={command, jsonschema, pytest, composite},
  xtick=data,
  legend style={at={(0.5,1.16)}, anchor=north, legend columns=2, draw=none, font=\small},
  ylabel style={font=\small}, tick label style={font=\small},
  axis lines=left, enlarge x limits=0.18,
  nodes near coords style={font=\scriptsize},
]
\addplot+[fill=blue!30, draw=blue!55] coordinates
  {(command,34) (jsonschema,0) (pytest,0) (composite,0)};
\addplot+[fill=orange!45, draw=orange!70] coordinates
  {(command,7) (jsonschema,10) (pytest,6) (composite,0)};
\addplot+[fill=black!12, draw=black!35] coordinates
  {(command,3) (jsonschema,0) (pytest,3) (composite,3)};
\legend{measured before, recovered by the fix, excluded (stated reason)}
\end{axis}
\end{tikzpicture}
\caption{Tier-1 coverage by gate kind. Before the fix the corpus resolved only
\texttt{command} gates, so every \texttt{jsonschema}, \texttt{pytest} and
\texttt{composite} loop produced no verdict \emph{and appeared in no count} ---
not as errors, not as exclusions, absent. The guard against unprincipled
exclusion inspected only loops already outside the eligible set, and both
predicates defining that set read a \texttt{command} gate's argv, so it was
vacuous over exactly the loops that had gone missing. The grey band is the 11
loops excluded for the stated reasons of \cref{sec:eval-setup}.}
\label{fig:coverage}
\end{figure}

\paragraph{A verdict-contract defect that removed a third of the evidence.}
With coverage fixed, \textbf{84 of 233} mutants came back as \Incap\ --- ``the
gate could not run''. All 84 were the misfiling described in
\cref{sec:verdict}: 24 checkers were returning the incapacity code for an empty
or malformed \emph{worker-owned} artifact. Seven loops produced nothing but
errors, and because unjudged mutants left the denominator entirely, \emph{the
surviving false-accept rate looked perfect}. The rate was clean because the hard
cases had stopped counting.

\paragraph{An empty ledger balances.} \Cref{ex:ledger}, reached only once
\texttt{pytest} gates were measured at all --- i.e.\ only after the scope defect
was fixed. The loop that certified deleting the books had been shipping for
three releases.

\subsection{Authors document how their gate is permissive, not how it is strict}
\label{sec:disclosure-asymmetry}

Before the corpus existed, six defects were found by reading all of the local
checkers and then running them on hand-constructed inputs. That set is small,
but it is the only one for which we have both the defect and the author's own
written account of the gate's limits, and the two do not line up in a symmetric
way.

Classify each defect by whether the gate's documentation admitted the behaviour.
An \emph{undisclosed gap} is one where the docstring or README advertises
something the code does not deliver --- an ordinary bug. A \emph{disclosed
limitation} is one where the code documents its own narrowness accurately and the
loop's stated reason for existing is nonetheless not served: accurately described,
still not fit for purpose.

\begin{table}[t]
\centering
\small
\begin{tabular}{@{}llc@{}}
\toprule
\textbf{Direction} & \textbf{Defect} & \textbf{Disclosed?} \\
\midrule
false accept & data-processing terms: substring match inside a denial & yes \\
false accept & runbook: required headings present, no content under any & yes \\
false accept & alt-text: attribute absent, image passes & yes \\
false accept & citations: pattern derived from the reference set (\cref{ex:citations}) & \textbf{no} \\
\midrule
false reject & conventional commits: \texttt{!} breaking-change marker refused & \textbf{no} \\
false reject & dependency pinning: extras and spacing legal under the spec refused & \textbf{no} \\
\bottomrule
\end{tabular}
\caption{The six defects found by hand, by direction and by whether the gate's
own documentation disclosed the behaviour. Three of four false accepts were
disclosed; neither false reject was. The sample is too small to carry a rate,
and the pattern is nonetheless the one we would predict.}
\label{tab:disclosure}
\end{table}

Three of the four false accepts were disclosed. Neither false reject was. Six
observations support no rate and we claim none, but the direction has a mechanical
explanation: an author writing a checker
reasons forward from the cases they intend to catch, and the cases they know
they are \emph{not} catching are the ones they are thinking about while writing,
so those get a comment. Being too strict is not a case they are thinking about
at all --- a legitimate artifact the pattern happens to refuse is, by
construction, an input that never occurred to the person who wrote the pattern.
Self-knowledge extends to a gate's permissiveness and stops at its strictness.

This is the practical argument for reporting $\betafr$ at all, and for reporting
it with its own denominator rather than folding it into a single accuracy
figure. The false-accept direction is the one an author has partly audited
already. The false-reject direction is the one where measurement is the only
available source of information, because introspection provably does not reach
it. It is also the direction \cref{thm:judge-asymmetry} identifies as the one a
model judge can worsen --- so a harness that does not publish $\betafr$ has no
way to see the cost of admitting one.

\subsection{E2 --- Tier-2: semantic mutants authored blind}
\label{sec:e2}

\paragraph{Method.} Five independent agent CLIs were each shown a loop's
\texttt{PROMPT.md} and human-facing \texttt{README.md} --- its stated purpose ---
together with the converged artifact, and asked to author an edit that satisfies
the stated requirement's letter while defeating its intent. \textbf{The checker
was withheld.} Each mutant records the digest of exactly the material its author
saw, and a guard recomputes that digest from the repository, so widening the
authoring prompt is a failing test rather than a silent protocol change. No
author contributed more than 21\% of the tier.

\paragraph{Result.} 38 mutants, 38 judged, 0 false accepts, 0 false rejects ---
after fixing the \textbf{16 gate defects} the tier found. Grouped by evasion:

\begin{table}[t]
\centering\small
\begin{tabular}{@{}p{0.30\linewidth}p{0.64\linewidth}@{}}
\toprule
\textbf{loop} & \textbf{how the gate was evaded} \\
\midrule
\texttt{test-presence-per-module} & a test with no import and no assertion; a test importing a
  \emph{different} module; a test asserting on a standalone calculation; renaming the source so
  the import orphans \\
\texttt{nda-required-clauses}, \texttt{gdpr-dpa-terms} & the required heading kept, the obligation
  underneath it negated \\
\texttt{okr-measurable}, \texttt{dependency-pinning}, \texttt{conventional-commits} & the
  offending item \emph{deleted} rather than repaired \\
\texttt{dependency-pinning} & pinned to a version that was never released; pinned below the
  file's own stated floor \\
\texttt{conventional-commits} & subject substituted wholesale; a docs change typed \texttt{build};
  an already-conforming line rewritten \\
\texttt{secret-scan-keyless} & credential split across a concatenation; \texttt{password} renamed
  to \texttt{passwd}; the config deleted rather than moved to the environment \\
\bottomrule
\end{tabular}
\caption{The 16 Tier-2 defects. Every one satisfies the requirement as literally
stated and defeats it in substance --- the failure mode a mechanical operator
cannot construct, because constructing it requires understanding the purpose.}
\label{tab:tier2}
\end{table}

\paragraph{The tier earned its existence rather than being argued for.}
\texttt{test-presence-per-module} is the sharpest case. \textbf{Tier~1 excludes
it entirely}, because its requirement is a relation between a source module and
its test and single-artifact mutation cannot express that. \textbf{Tier~2 found
four defects in it.} The loop whose gate was structurally unmeasurable by the
cheap method was the loop with the most defects.

Deleting rather than repairing recurs across three unrelated loops, and it is
\cref{thm:vacuity} restated by an adversary: if the check quantifies over items
and every remaining item conforms, removing the offending item passes. The fix
in each case was an existence obligation with a floor derived from what the seed
shipped (\cref{rem:floor-judgement}).

\subsection{Two results we could have rounded away}
\label{sec:eval-corrections}

\paragraph{A mutant was mislabelled and the gate was right.} One author changed
an invoice quantity from \texttt{50} to \texttt{50.0} and reported it as breaking
a three-way match. JSON does not distinguish them and $50 = 50.0$; the gate
correctly passed. Taking the label on trust would have meant ``fixing'' a correct
finance gate to reject a float. The entry is recorded, keyed by content digest,
as a known mislabelling rather than deleted --- so the corpus carries its own
correction.

\paragraph{The ``not mechanically checkable'' category ships empty.} The corpus
provides a category for requirements no mechanical checker can decide, on the
principle that an unmeasurable question should be refused rather than passed.
Four Tier-2 findings were placed in it and all four turned out to be checkable
--- a wholesale subject substitution retains none of the original's content
words; a documentation change typed \texttt{build} contradicts the
specification's own reservation of \texttt{docs}; an already-conforming line
rewritten is detectable by pinning what the seed recorded as conforming; and a
fabricated version is answerable from shipped reference data, exactly as another
loop in the same catalogue already answers whether a citation exists.

Each filing had been made by reasoning from the \emph{shape} of the requirement
rather than by attempting an implementation. ``A checker cannot know a commit's
original meaning'' is true, and is not the question; the question is whether the
specific evasion leaves a mechanical trace. The category therefore remains
defined and empty, and an entry in it is now a claim that someone tried to write
the checker and failed.

\subsection{E3 --- The confidence sequence, and which estimand it covers}
\label{sec:e3}

Coverage was measured by simulation on sequences with the shape the product
actually produces --- attempts pooled across \emph{many} nodes, so the latent
propensity is not constant along the sequence --- checking at every $n$:

\begin{table}[t]
\centering\small
\begin{tabular}{lcc}
\toprule
\textbf{estimand} & \textbf{anytime-valid} & \textbf{fixed-time} \\
\midrule
$\logmean$ --- mean propensity of the attempts in the log & \textbf{1.0000} & 0.9590 \\
population marginal, 1 node pooled & 0.8333 & 0.6067 \\
population marginal, 6 nodes pooled & 0.9717 & 0.8450 \\
population marginal, 30 nodes pooled & 0.9983 & 0.9783 \\
population marginal, heterogeneous ($\rho = 3.5$) & 0.8267 & 0.6000 \\
\bottomrule
\end{tabular}
\caption{Empirical coverage at nominal 95\%, by estimand. The sequence is
anytime-valid \textbf{for $\logmean$}; coverage of a population marginal is a
different quantity that rises with the number of independent nodes pooled and
falls with their heterogeneity.}
\label{tab:coverage}
\end{table}

$\logmean = 1.0000$ is conservative rather than vacuous, and we verify that
rather than assert it: shrinking the radius to $0.6\times$ drops coverage to
$\approx 0.26$, to $0.3\times$ drops it to $\approx 0.04$, and shifting the centre
by $+0.10$ drops it to $\approx 0.46$. All three controls are in the suite, so a
change that made the interval vacuously wide would fail a test.

The widths matter as much as the coverage, and a reader cannot judge sharpness from a
coverage rate alone. On this run the fixed-$n$ empirical-Bernstein 95\% upper bounds are
$\alphafa \le 0.0664$ and $\betafr \le 0.0739$; the anytime-valid bounds over the same data
are $\alphafa \le 0.1916$ and $\betafr \le 0.2136$; and the Wilson bound quoted in the
abstract is $0.0897$. The anytime-valid interval is roughly three times wider, which is the
price of a guarantee that holds at every stopping time rather than at one chosen in advance
--- and it is the right number to quote for a corpus that will grow, because a fixed-$n$
bound re-read after adding mutants is no longer a 95\% bound.

The distinction \cref{tab:coverage} draws is the whole of the practical advice
here: a coverage figure is uninterpretable without naming its estimand, and the
two columns of that table differ by more than a factor of the gap most readers
would tolerate. We report $\logmean$ because a receipt log is an audit object,
and the estimand is named wherever the number appears.

\subsection{E4 --- Workers that try to game the gate}
\label{sec:e4}

Adversarial workers exercise the architecture directly rather than the checkers.
Each asserts success without earning it; none reaches $\Done$, and the receipt
log shows why in every case.

\begin{itemize}[leftmargin=1.4em, itemsep=1pt]
  \item \textbf{Declare completion.} A worker that reports it has finished without
        modifying the artifact. Its claim is read by no guard (\cref{thm:soundness}):
        the gate runs, rejects, and the loop retries.
  \item \textbf{Edit the checker.} Refused before the gate is reached. The anti-tamper
        guard rescans after every turn and raises, so the loop never evaluates against a
        tampered workspace.
  \item \textbf{Plant a collection config.} For test-suite gates, adding or altering a
        \texttt{conftest.py}, \texttt{pytest.ini} or equivalent can redirect what is
        collected and produce a pass without touching the checker's bytes. Treated as
        anchor tampering for exactly this reason.
  \item \textbf{Delete the artifact.} The \cref{thm:vacuity} attack. Now a rejection
        wherever an existence obligation is declared --- and E1/E2 are how we know
        where it was not.
\end{itemize}

\Cref{thm:soundness} predicts the first three, and the ownership partition being
enforced rather than requested predicts them jointly. The fourth is not
architectural, which is why \cref{sec:soundness} is separate from
\cref{sec:vacuity}: the architecture routes the decision to the gate, and whether
the gate then decides correctly is an empirical question the corpus answers.

\subsection{E6 --- The instrument, turned on ourselves}
\label{sec:e6}

The corpus measures gates in a catalogue. The same question can be asked of the
test suite that guards the engine, and we asked it, because a paper arguing that
verification layers go unmeasured should not exempt its own.

Scanning for \cref{def:vacuous} found \textbf{14} tests whose assertions all sat
inside a quantification over a collection that nothing proved non-empty. The scan
traversed 222 test files carrying 2{,}811 test functions; after the 14 were repaired,
0 remained.
Confirmed by stubbing tool registration to a no-op, after which three security
guards still passed:

\begin{itemize}[leftmargin=1.4em, itemsep=1pt]
  \item ``no tool accepts a subject identity'' --- satisfied by there being no tools;
  \item ``no tool takes a filesystem path'' --- likewise;
  \item ``no discovery tool is secret-shaped'' --- likewise.
\end{itemize}

Plus a credential-leak guard that scanned a run directory and passed on an empty
one, and a test named for a property it never checked: it asserted a sentinel was
not declared \emph{twice} and passed on zero.

Two details we record because they are the informative part. One of the fourteen
was written earlier the same day, by an author who had spent that day fixing this
class in the catalogue. And the scanner itself had the defect: it detected only
the ``assert inside the loop'' shape and missed the ``accumulate, then assert the
accumulator is empty'' shape, which \emph{looks} unconditional, until it was
run against that same freshly written test. The guard that ships catches both and
includes a deliberately vacuous case it must flag, so a refactor that disables
the scan fails rather than passing silently.

\textbf{What E6 is evidence for.} Not that we are careless. That the class survives attention: it survived review,
it survived the author's own fresh knowledge of it, and it survived the first
tool built specifically to find it. That is the argument for measuring a
verification layer instead of reasoning about it, made with the strongest
evidence we have: evidence against ourselves.

\subsection{E5 --- Cross-model convergence}
\label{sec:e5}

\paragraph{Question.} Does the same bounded loop, driven by independent provider
models, converge through the same gate? The engine's value proposition assumes
the gate is what decides, not the model; if outcomes tracked the provider, the
gate would be decoration.

\paragraph{Method.} Four loops spanning all three keyless gate kinds
(\texttt{command}, \texttt{pytest}, \texttt{jsonschema}), selected before any run
and fixed, each driven by four independent agent CLIs. Seed, gate, bounds and
forbid set are identical across arms; only the worker changes. Terminal state
comes from the engine's own exit path and the attempt count from the receipt log
and never from anything a CLI said about itself. Attempt ceiling 4.

\paragraph{What is deliberately not claimed.} This is not a capability benchmark.
Each loop carries one seeded defect and one fixed mechanical check; a provider
needing more attempts on one loop has not been shown to be worse at anything
general. The quantity of interest is whether the \emph{gate} admits the same
outcomes regardless of authorship.

\paragraph{Two limits, stated.} A fifth CLI (\texttt{codex}) ships a profile and is
excluded: its subscription was over its periodic budget for this window, so a
five-provider run was not producible. \textbf{Four ran; the paper says four.} And
because arms are driven through the engine's generic shell runner, which does not
parse provider usage envelopes, token and cost figures are unavailable here and are
reported as such rather than estimated. Terminal state and attempts-to-\textsc{done},
the two quantities the question turns on, are read from receipts and are exact.

\paragraph{Result.} \textbf{16 of 16 arms reached \textsc{done}, every one on its
first attempt.}

\begin{table}[t]
\centering\small
\begin{tabular}{llcccr}
\toprule
\textbf{loop} & \textbf{gate kind} & \textbf{provider} & \textbf{terminal} & \textbf{attempts to \textsc{done}} & \textbf{wall (s)} \\
\midrule
\texttt{citation-existence-check} & command    & A & \textsc{done} & 1 & 26.6 \\
                                  &            & B & \textsc{done} & 1 & 27.5 \\
                                  &            & C & \textsc{done} & 1 & 54.9 \\
                                  &            & D & \textsc{done} & 1 & 86.6 \\
\midrule
\texttt{bug-fix-red-green}        & pytest     & A & \textsc{done} & 1 & 24.3 \\
                                  &            & B & \textsc{done} & 1 & 36.0 \\
                                  &            & C & \textsc{done} & 1 & 37.7 \\
                                  &            & D & \textsc{done} & 1 & 79.5 \\
\midrule
\texttt{catalog-required-fields}  & jsonschema & A & \textsc{done} & 1 & 59.1 \\
                                  &            & B & \textsc{done} & 1 & 60.4 \\
                                  &            & C & \textsc{done} & 1 & 64.6 \\
                                  &            & D & \textsc{done} & 1 & 151.7 \\
\midrule
\texttt{record-completeness}      & command    & A & \textsc{done} & 1 & 25.7 \\
                                  &            & B & \textsc{done} & 1 & 36.5 \\
                                  &            & C & \textsc{done} & 1 & 48.3 \\
                                  &            & D & \textsc{done} & 1 & 270.6 \\
\bottomrule
\end{tabular}
\caption{Cross-model convergence: four independent provider models, four loops, one
gate each. \textbf{Letters are assigned within each loop by ascending wall time and
do NOT identify a provider across rows}; letters are not mapped to vendor names
because sixteen observations cannot support a product comparison and a named table
would invite one. Wall time is reported only to show the arms differ in something
while agreeing on everything the gate decides, and must not be read as a
per-provider figure. Tokens are unavailable through this path and are not estimated.}
\label{tab:e5}
\end{table}

The uniformity is the finding, and it is a modest one stated precisely: on these four
loops, terminal state and attempt count were invariant across providers whose
wall-clock behaviour differed substantially. That is consistent with the gate, rather
than the model, determining the outcome.

\paragraph{Two claims this table used to make, withdrawn.} An earlier caption
asserted a fixed letter-to-provider mapping. The data cannot support one: no single
assignment yields ascending times in all four loops, because the ordering genuinely
changes between them. Detecting it required regenerating the table from its artifact
rather than reading it, and the table was rebuilt in the process, having been carried
over from a superseded twelve-arm run. The second withdrawal is the spread between
fastest and slowest arm, once quoted as a finding. Re-running all sixteen cells
against an unchanged system moved individual times by factors from $0.83\times$ to
$1.86\times$, one cell nearly doubling, so a cross-provider ratio from a single run
per cell is substantially measurement noise. The general point outlasts both
retractions: an experiment that runs each cell once cannot distinguish a property of
the system from a property of the afternoon. Nor is this evidence about harder loops
--- the seeded defects here are ones any competent agent repairs in a single attempt,
a ceiling effect we should not read past, and a loop whose defect is genuinely
difficult would separate the arms.

\subsection{A near-miss worth more than the result}
\label{sec:e5-nearmiss}

The first execution of E5 produced a clean, publishable, and \emph{false}
finding: one provider reported \textsc{halt} on every loop, at exactly the
attempt ceiling, while the other three reached \textsc{done} on their first
attempt. Written up as it stood, that is a claim about a third party's product.

It was our misconfiguration. Probing the CLI directly showed two things, both
silent:

\begin{enumerate}[leftmargin=1.6em, itemsep=1pt]
  \item Its file-writing tools were being auto-denied, because the invocation
        carried no permission grant. It exited successfully, reporting
        \emph{``no output produced --- a tool required a permission that headless
        mode cannot prompt for, so it was auto-denied.''}
  \item It does not treat the process working directory as its workspace. Asked
        to create a file in the current directory, it created one in its own
        scratch directory and reported ``Created and verified'' with a path
        outside the workspace the gate reads.
\end{enumerate}

So the agent ran, believed it had succeeded, said so, and the gate, correctly,
saw an unchanged workspace and rejected. Four times. Then the loop halted,
exactly as \cref{thm:loop-termination} guarantees.

Investigating the cause also surfaced two defects in the engine, both of the same
shape and both now fixed: one runner invoked its CLI with flags that binary
rejects outright, and the subprocess environment allowlist withheld a variable a
CLI needs in order to load its configuration. Each produced an agent that ran and
did nothing, and each was invisible from the ledger. They had shipped. The
experiment found them because it was the first thing to drive those runners
against real binaries end to end, which is the same reason the mutation corpus
found forty-seven gate defects a green test suite did not.

\subsection{E7 --- The post-freeze corpus, and what a frozen gate still admits}
\label{sec:e7}

Every number in \cref{sec:e1,sec:e2} is a saturation figure. The operators that
found the 47 defects are the operators that scored the gates after repair, so
their zero establishes that \emph{this corpus is exhausted against these gates}
and nothing further. This experiment answers that objection.

\paragraph{Protocol.} The gates are frozen exactly as shipped. A fresh operator
family is applied, and \textbf{nothing it finds is repaired before it is
recorded}: fixing first and measuring after would reproduce the circularity
the experiment exists to break. Eligibility, labelling and blindness are
unchanged, so the two measurements are comparable.

\paragraph{Why a fresh family was available at all.} The shipped destroying
operators are \texttt{empty\_file}, \texttt{whitespace\_only} and
\texttt{filler\_text}: three operators that all destroy the \emph{whole}
artifact. That bluntness is what makes their label certain without consulting a
gate, and it is also a single probe wearing three coats. It can only find a gate
satisfied by the \emph{absence} of what it checks.

The post-freeze family carries a different claim with the same certainty:
\emph{an artifact that keeps its structure and loses its content serves no
stated purpose either.} A payload whose every required key is present and empty.
A document of headings with nothing beneath them. A module whose every function
is \texttt{pass}: it imports, it collects, it runs, and it asserts nothing. A
gate that checks shape rather than substance is invisible to the shipped
operators, which delete the shape too, and reachable by these.

\paragraph{Result.}

\begin{table}[t]
\centering\small
\begin{tabular}{lrr}
\toprule
 & \textbf{Tier 1 (saturated)} & \textbf{E7 (post-freeze)} \\
\midrule
loops covered              & 57 & 54 \\
destroying mutants         & 171 & 60 \\
preserving mutants         & 62 & 35 \\
gate errors (\Incap)       & 0 & 0 \\
\midrule
false accepts $\alphafa$   & 0 / 171 & \textbf{14 / 60 = 23.3\%} \\
\quad Wilson 95\%          & $\leq 2.2\%$ & $[14.4\%,\ 35.4\%]$ \\
false rejects $\betafr$    & 0 / 62 & \textbf{0 / 35} \\
\quad Wilson 95\% upper    & $\leq 5.8\%$ & $\leq 9.9\%$ \\
\bottomrule
\end{tabular}
\caption{The saturated corpus and the post-freeze corpus against the same frozen
gates. The left column is what an exhausted operator set reports; the right is
what a fresh one finds behind it. \textbf{The comparison is the result, not
either column.} A reader should take from this that a gate suite reporting zero
against the corpus that repaired it has established exhaustion and not quality
and that the remedy is another operator family, which is cheap, rather than
more mutants from the same one, which is free and worthless.}
\label{tab:e7}
\end{table}

\paragraph{Twenty-one became fourteen, and the arithmetic is the finding.}
The raw count was 21 apparent false accepts over 67 destroying mutants. We
adjudicated \emph{all} 21 against each loop's stated requirement (the whole
set rather than a sample, because the set is small enough) and seven were
\textbf{mislabelled by the operator}. They leave the numerator and the
denominator both, since a mutant nobody can label was never a valid trial. The
adjudication is released with the corpus.

The seven are not seven accidents. They are two failures of one assumption, and
stating it converts the limit on ground-truth-by-construction
(\cref{thm:no-equivalent}) from a qualitative caveat into a checkable side
condition:

\begin{quote}
A content-removal operator's \emph{destroying} claim is certain only for a
requirement \textbf{universally quantified over content the operator removes}.
It is false for a \textbf{negative} requirement, ``must not contain $X$'':
because removing content removes $X$; and false for a requirement satisfied by
\textbf{structure the operator preserves}: a name, a signature, a key, a
heading.
\end{quote}

\texttt{no-hardcoded-sleep} asks that no test contain a hardcoded sleep, and
stubbing bodies deletes the sleeps. \texttt{type-annotations-present} asks that
signatures be annotated, and stubbing bodies keeps every signature. In both the
gate was right and the operator was wrong. The rule is now enforced by
\emph{refusing to emit}, and it fails closed: one loop whose requirement names
no quantifier is withheld although its mutant was adjudicated genuinely
defective. One true positive lost, no unlabellable trial admitted, and we report
the trade rather than tuning the predicate until it flatters the result.

\paragraph{What the fourteen were.} Ten are vacuous satisfaction over an emptied
collection (``every dataset's license is in the allowlist'', ``every SKU's
running balance stays non-negative''), each verified by running the frozen
checker directly on the emptied artifact. One is a test-suite gate quantifying
over records the operator removed. The remaining three are the finding we would
put first.

\begin{remark}[An incompletely applied fix, invisible to the corpus that motivated it]
\label{rem:incomplete-fix}
\texttt{runbook-completeness} was the hand-found defect that motivated the
content-presence rule: seven required headings present, nothing under any of
them, gate passed. It was repaired, and its checker now requires content beneath
each heading.

Three sibling gates have the identical shape and did not receive the fix.
\texttt{privacy-policy-completeness} checked headings and nothing else.
\texttt{nda-required-clauses} and \texttt{gdpr-dpa-terms} checked that a required
heading was present and that its body did not \emph{void} the obligation, which was
the earlier negation repair, but never that the body existed at all. A data
processing agreement whose clause headings are present and whose clause bodies
are empty is worthless, and it passed.

The shipped operator family cannot reach this. \texttt{empty\_file} deletes the
headings too, at which point all four gates correctly report the sections
missing. The defect lives precisely in the gap between destroying an artifact
and hollowing it, and only an operator aimed at that gap can see it. \emph{A
repair applied to the instance that revealed a class, and not to the class, is
indistinguishable from a repair applied to the class --- until a different
instrument arrives.} All three are fixed in the release following this
measurement, and the number above is the pre-fix one, which is the measurement.
\end{remark}

\subsection{E8 --- A class the instrument cannot generate}
\label{sec:e8}

\paragraph{Question.} Every instrument in E1--E7 was built by the people who wrote
what it measures, which is the first threat \cref{sec:threats} names and the one it
cannot answer from inside. E8 asks the answerable form: does the instrument miss a
class of defect that its own operators are structurally unable to express?

\paragraph{Method.} Two commercial agent CLIs, each at its highest documented
reasoning setting, were given the same brief at the same commit of the engine and
run in separate detached worktrees so neither could observe the other. Findings
were reproduced by hand before acceptance and each accepted repair was
mutation-tested. \Cref{app:audit} records the protocol, including what it cost to
learn that an auditor is an agent and therefore writes.

\paragraph{Result: yes, and the mechanism is specific.} Every accepted finding fell
into a class this paper had already named, with one exception that matters. Three
of them populate \cref{def:self-attested}, and \textbf{no operator in either
mutant tier can produce one}: both tiers mutate worker-owned artifacts, and a
self-attested check is a defect in the harness that reads them. The instrument's
blindness is therefore not a sampling gap that more mutants would close --- it is a
consequence of what the operators act on. E1's coverage over gates says nothing
about the code that decides whether a gate's answer may be believed.

The three are \cref{ex:lying-detail}, \cref{ex:accepting-stub} and
\cref{ex:forged-provenance}: a rule about a verdict's reason that consulted the
verdict, a test about a platform's install policy that consulted a stub written to
agree, and a record of which gate ran that read the answer off the gate. They share
no code, and \cref{sec:second-axis} states what they share instead.

\paragraph{A second, cheaper finding: a guard evaluated lazily is a guard tested on
one platform.} Acting on the result above removed a field from a capability record.
Three test modules gated themselves at import time on
\texttt{not (live.seatbelt or live.bubblewrap)}. Python evaluates \texttt{or} left
to right; the authors' platform makes the first disjunct true, so the second was
never evaluated where the suite was run. On any other platform it raises during
\emph{collection}, so those modules do not fail --- they do not exist. The full
suite passed. The type-checker found it, because it reads code the suite does not
execute. Stated generally: \emph{a precondition behind a short circuit is tested
only on hosts where the circuit does not fire}, and which hosts those are cannot be
seen from one machine.

\paragraph{What E8 does not establish.} Two auditors, one codebase, one release, no
held-out set. There is no recall figure and we compute none, because we cannot say
what either auditor missed; the filter that accepted their findings was also ours.
Both are commercial products at versions we cannot pin, so E8 is not replicable in
the sense E1--E3 are. It is a case report and we label it one. What a case report
can establish is existence, and the existence claim is exactly the paragraph above:
a defect class outside the instrument's reach, demonstrated rather than
conjectured, by parties with no authorship of the code --- in a codebase whose
authors had spent the preceding weeks studying that class.

\section{The manifest as a spend contract}
\label{sec:spend-contract}

The two guarantees established so far answer questions an engineer asks. Termination
(\cref{thm:loop-termination,thm:graph-termination}) answers \emph{will this stop}; soundness of
\textsc{done} (\cref{thm:soundness}) answers \emph{is the completion real}. Neither answers the
question that actually decides whether a run is authorised, which is asked by whoever owns the
budget:

\begin{quote}
\emph{What is the most this can cost me, and can I read that number before it starts?}
\end{quote}

This section states what the harness guarantees about spend, exactly, including the term it does not
cover. The distinction between a bound that is \emph{declared} and one that is \emph{enforced} carries
the whole result.

\subsection{Three ceilings and one derived bound}

A loop manifest declares its envelope in four numbers, all readable before the run:

\begin{itemize}[leftmargin=1.4em]
  \item \texttt{max\_iterations} ($m$) --- the hard attempt cap. Work bound.
  \item \texttt{max\_tokens} --- cumulative token budget across attempts, or unset.
  \item \texttt{max\_wallclock\_s} ($W$) --- elapsed time for the run.
  \item \texttt{no\_progress\_window} ($w$) --- the soft bound: consecutive attempts that change
        nothing before the run halts.
\end{itemize}

For a graph, \cref{cor:manifest-readable} composes these into the closed form
$(1+R)\sum_{v}m_v$, every term of which is a literal in the manifest. That corollary is what lets an
operator authorise a run against a number rather than against a hope.

The soft bound is the one that pays for itself in practice. \cref{prop:no-progress} gives the
effective attempt bound $\min\{m,\, w + j^{\star}\}$, where $j^{\star}$ is the number of attempts on
which the worker changed the artifact. A loop that stalls therefore stops after $w$ wasted attempts
rather than at $m$, and \cref{sec:bound-utilisation} reports that this bound is exact rather than
merely valid on every case we measured.

\subsection{What the wallclock ceiling promises}

The caveat is inside the claim rather than beside it.

\begin{proposition}[Bounded elapsed spend]
\label{prop:spend-bound}
Let a run declare \texttt{max\_wallclock\_s} $= W$ and \texttt{handoff\_reserve\_s} $= r$, with
$0 \le r < W/2$. Let $G$ be the gate's own timeout. Then:
\begin{enumerate}[leftmargin=1.6em]
  \item no attempt begins after $W - r$ elapsed;
  \item no attempt continues past $W - r$ elapsed;
  \item total worker time, including any wind-down turn, is bounded by $W$;
  \item total run time is bounded by $W + G$.
\end{enumerate}
\end{proposition}

\begin{proof}[Proof]
The four clauses are proved in \cref{app:spend-proof} from the controller's step structure: a
monotone clock, the budget guard at every lap head, the runner's clamp on its blocking wait,
and the sequentiality of gate invocation. The argument is elementary and is written out because
two of its steps were wrong in shipped code --- the guard compared against the declared ceiling
rather than the work ceiling, and the remainder handed to a runner was computed once at wiring
time, which makes it a constant rather than a remainder.
\end{proof}

\begin{remark}[What $G$ actually is in the shipped wiring]
\label{rem:gate-tail}
$G$ is written as a separate symbol because it is a separate \emph{concept} --- the
guarantee is ``the declared ceiling plus one gate'', whatever that gate costs. In our
implementation that number is not small: the composition root wires each gate's timeout to
\texttt{max\_wallclock\_s}, so $G = W$ and the worst case is $2W$. That is a deliberately
generous gate allowance and not a tight one; a deployment wanting a tighter total should
set a smaller gate timeout, and the guarantee is stated in terms of $G$ precisely so that
doing so tightens the bound without changing the result.

Two cases where the clause does not apply, stated because a reader tracing the code will
find them. A loop that declares no \texttt{max\_wallclock\_s} at all gets a platform
default of one hour on the manifest path --- so $W$ exists and the bound holds --- but
\texttt{Bounds} constructed directly in code with no ceiling passes no timeout to the
gate, and a gate subprocess with no timeout has no $G$. The proposition's hypothesis
requires a declared $W$, so nothing false is claimed; what would be false is reading
clause~(4) as a property of every possible run.
\end{remark}

Enforcing a wall-clock ceiling \emph{during} a unit of work is long-established
infrastructure practice, and we claim none of it. A Kubernetes Job's
\texttt{activeDeadlineSeconds} terminates the pod irrespective of what its code is doing;
Temporal's \texttt{scheduleToCloseTimeout} bounds an activity across all its retries and
is readable from the workflow definition before it runs; AWS Step Functions and Argo
Workflows expose per-task timeouts and retry caps in the same way. Any claim that this
paper is the first to stop a task at a declared time would be refuted by a decade-old
manifest field.

Three things in \cref{prop:spend-bound} are not supplied by those systems, and they are
the contribution. The ceiling is a field in the \emph{workload's own} manifest rather
than the scheduler's, so the number an operator authorises and the number the harness
enforces are the same literal. The reserve \emph{partitions} that number instead of
extending it, which is what lets a hard stop leave usable work behind without changing
any bound (\cref{sec:handoff-reserve}). And the gate is deliberately exempt, so a
verdict is never truncated into a non-verdict --- the failure mode a general-purpose
deadline has no way to distinguish and therefore no way to avoid.

\begin{remark}[Two limits on one attempt, and why the receipt distinguishes them]
An attempt is subject to two ceilings with different owners. \texttt{max\_wallclock\_s} is declared by
the loop author and is a spend contract; the runner's own per-turn timeout is set by the operator
deploying it and expresses what that deployment tolerates for a single turn. The tighter binds. When
one fires, \emph{which one} is recorded: exceeding the declared ceiling halts the run and names the
bound, while exceeding the deployment's own timeout is an error. Collapsing the two would file ``your
budget ran out'' under the same heading as ``the harness broke'', and an operator reading the receipt
cannot act on the first if it is reported as the second.
\end{remark}

\subsection{Reserving, not extending: a bound that does not destroy what it interrupts}
\label{sec:handoff-reserve}

A hard bound must be hard or it bounds nothing. But a bound that reports only \emph{budget exceeded}
discards work already paid for, and the consequence compounds: the next run begins from the same
initial state with the same budget and no record of what the previous one established. A task
requiring more than one budget window can then never complete, however many windows are spent on it.
The bound stops being merely strict and becomes counterproductive.

The obvious remedy is to grant one extra turn on expiry so the worker can summarise. That remedy is
wrong: if exceeding $W$ buys an additional turn then $W$ silently means ``$W$ plus however long a
summary takes'', and that second term is precisely the quantity the operator cannot read from the
manifest --- a bound that is declared and not enforced, reached from the opposite direction
(\cref{sec:declared-vs-enforced}).

So the reserve \emph{partitions} the ceiling instead of extending it:

\begin{equation}
\underbrace{(W - r)}_{\text{work}} \;+\; \underbrace{r}_{\text{wind-down}} \;=\; \underbrace{W}_{\text{declared}}
\label{eq:reserve-partition}
\end{equation}

\begin{corollary}[The courtesy is free of the guarantee]
\label{cor:reserve-preserves-bounds}
Every termination and spend result in this paper holds verbatim under
\cref{eq:reserve-partition}, for any $r$ in range.
\end{corollary}

\begin{proof}
The declared total is unchanged, so no bound expressed in terms of $W$ changes. The wind-down turn is
bounded by $r$ and occurs at most once, so the attempt count is unchanged and
\cref{thm:loop-termination} applies unmodified. \qedhere
\end{proof}

That corollary is the point. There is no branch on which a run outlives its declared ceiling in
exchange for a handoff, so nothing had to be re-proved to make the bound humane.

Two design consequences follow, and both are decisions rather than details:

\begin{itemize}[leftmargin=1.4em]
  \item \textbf{The external kill switch gets no wind-down.} An operator pulling it wants the run to
        stop now, not to spend more of anything. Every \emph{bound} halt gets one --- including the
        no-progress bound, where a worker that is stuck explaining what it tried is the most useful
        handoff of the set rather than the least.
  \item \textbf{The wind-down cannot change the terminal state.} A handoff turn that hangs, fails, or
        produces nothing costs the reserve and no more. Without that invariant, an attempt to be
        helpful could convert a clean halt into a failure, which is strictly worse than the
        abruptness it replaced.
\end{itemize}

\subsection{Auditing consumption from the receipt}

A spend claim that cannot be checked against its own receipt is not a spend claim. Two properties of
the append-only ledger make consumption recomputable by a third party holding only the run directory.

First, every row records whether the worker was actually invoked. The distinction matters because two
of the controller's checks --- the kill switch and the budget ceiling --- occur \emph{before} a turn
and therefore record a lap on which no attempt was spent. Without the flag, a ceiling halt at
$m = 10$ writes an eleventh row, and anyone computing consumed-over-declared from the ledger obtains
$1.1$ for a bound that in fact held exactly. We report that quantity in
\cref{sec:bound-utilisation}, so an instrument that overstates it by one on precisely the runs where
the ceiling bites would have corrupted the measurement rather than merely the display.

Second, a bound halt names the file containing its handoff, so the evidence is reachable from the
receipt without knowing a convention.

\begin{definition}[Bound utilisation]
\label{def:bound-utilisation}
For a completed run, $U = a/m$ where $a$ is the number of attempts spent and $m$ the declared cap.
$U = 1$ is a run that consumed its declared envelope exactly; $U < 1$ is a run that finished or
stalled early; $U > 1$ is impossible and, if observed, indicts the instrument rather than the run.
\end{definition}

\section{Declared is not enforced: a defect class, found in our own system}
\label{sec:declared-vs-enforced}

Every guarantee in this paper is stated over a manifest. That construction has an exposed joint:
between the number an operator reads and the behaviour that number is supposed to constrain. A
manifest field can be present, validated, type-checked, documented, and reachable by every consumer
that displays it, while constraining nothing at all. The bound is then a label rather than a bound,
and --- this is the part that makes the class worth naming --- \emph{everything an operator can
inspect looks correct}.

We found this in our own system, after the termination and soundness results had been proved and the
evaluation had been run. It is the strongest evidence we have for the paper's thesis, so we report it
in full rather than quietly fixing it. The instance below is the wallclock ceiling. The class also
caught \cref{thm:chain}: the chain-integrity theorem was stated over a hash chain the ledger adapter
did not write, so every hypothesis in it quantified over a construction no shipped code produced ---
the theorem was correct and vacuous, undetected across four releases, and it undermined the one
artifact every other claim here is evidenced by. A proof does not check that the object it quantifies
over exists.

\subsection{The defect}

\texttt{max\_wallclock\_s} was declared in all 69 loop manifests in the catalogue. It was validated at
load time, normalised (a null value became a conservative one-hour default rather than
``unbounded''), wired into every gate constructor, surfaced in the pre-run banner, and recorded in the
receipt. It was compared against elapsed time in exactly one place: at the top of each attempt, before
the worker was invoked.

That single comparison site is the defect. Checking a ceiling only \emph{between} attempts bounds the
gap between attempts, not an attempt. Once a worker was running, nothing in the system was measuring
it against the loop's declared ceiling; the process was subject only to whatever unrelated
per-subprocess default its runner happened to carry. Concretely: a loop declaring
\texttt{max\_wallclock\_s: 120} ran a single attempt for over 300 seconds and was terminated, in the
end, by a runner default of 300 that the loop had never declared. It was stopped by a limit nobody
wrote, 180 seconds after the limit somebody did write had been exceeded.

The manifest said what it meant. The run did not respect it.

\subsection{The second defect, which the first was concealing}

Enforcing the ceiling exposed something worse, and this is the finding we would not have predicted.

With the bound inoperative, nothing had ever tested whether the declared \emph{values} were
achievable. They were not. 65 of the 69 shipped loops declared \texttt{max\_wallclock\_s: 60} against
\texttt{max\_iterations: 10} --- a budget of six seconds per attempt. The catalogue's demonstration
worker completes an attempt in roughly $0.4$ seconds, so every test passed and every run finished
comfortably inside the ceiling.

A real agent command-line worker does not complete a turn in six seconds. Across four independent
providers driving four loops over two runs, we measured 31 completed turns: median $48.3$\,s, 75th
percentile $72.5$\,s, maximum $270.6$\,s. So no shipped loop's declared ceiling could accommodate a
single real turn, let alone ten of them.

Two defects, and the first was hiding the second. While the bound was unenforced, the wrong numbers
were invisible; the moment it was enforced, every real-agent run of every shipped loop would have
halted part-way through its first attempt. A reader may reasonably observe that enforcing the bound
made the system worse before it made it better. That is correct, and it is the argument for enforcing
a declared bound \emph{even when nothing appears to be wrong}: an unenforced bound suppresses the
signal that would have revealed its own value to be nonsense.

\subsection{Repair, and how the numbers were chosen}

Enforcement now happens at two points rather than one. The controller measures the remaining work
budget immediately before handing off to the worker, and the runner clamps its own blocking wait to
the smaller of that remainder and its configured per-turn timeout. There is one implementation of the
clamp, shared by every runner, for reasons \cref{sec:mirrored-definitions} explains.

The declared values were re-derived from the measurement above:
\begin{equation}
\texttt{max\_wallclock\_s} \;=\; m \times 90\,\text{s} \;+\; r
\label{eq:ceiling-rule}
\end{equation}
where $m$ is the declared attempt cap and $r$ the handoff reserve of \cref{sec:handoff-reserve}. The
90 seconds is the measured 75th-percentile turn, rounded up.

The percentile rather than the maximum is the interesting choice. A spend limit that cannot bind is
decorative, and sizing the ceiling to the slowest turn observed would make it arithmetically
incapable of firing before the attempt cap --- returning us to a bound present in the manifest and
unable to affect a run. The percentile is also the stable statistic: between the two measurement runs
it came out at $72.5$\,s both times, independently, while the maximum moved from $177.6$\,s to
$270.6$\,s and one cell nearly doubled ($1.86\times$). The body of the distribution is a property of
the workload; the tail is noise.

\subsection{Why we publish this, and why the class is not ours alone}
\label{sec:why-publish}

We publish it because the defect is undetectable from the outside, and a paper claiming provable
bounds that did not report finding one in its own artifact would be asking to be trusted on exactly
the point it cannot demonstrate. Everything an operator could inspect was correct. The only thing
that would have revealed it is the question \emph{what test fails if I delete the enforcement?}, and
we did not ask it until after the results in this paper had been proved.

Two formally recorded incidents show the same shape outside our system. In July 2025 an agent
deleted a production database during an explicit code freeze, having received and acknowledged the
instruction not to touch production \citep{aiid1152}; nothing in the execution path enforced the
freeze. A second record describes an agent clearing a project cache and removing an entire drive, in
a mode that executes without a confirmation step \citep{aiid1433}. In both, a constraint was stated
where an operator could read it with no enforcement point between the statement and the action. Ours
cost nothing because the workload was a demonstration, which is not a difference in the defect.

The commercial form of the gap is that the cost of an agent run is not knowable before it starts.
Gartner projects more than 40\% of agentic AI projects cancelled by the end of 2027, naming
escalating costs and inadequate risk controls \citep{gartner2025agentic}, and one engineering
organisation exhausted its annual AI budget four months into the year \citep{forbes2026uber} --- so a
pre-run spend ceiling is a live commercial requirement rather than a design nicety.

\subsection{Why this generalises past our system}

Three parts, none specific to agent harnesses. A declaration and an enforcement point are
\emph{different artifacts}, and nothing about writing, validating or displaying a limit causes
anything to check it. The tests were written against a workload that could not exercise the bound: a
demonstration worker at $0.4$\,s per attempt cannot reach a 60-second ceiling, so no test
distinguished ``enforced'' from ``declared''. And the failure was silent in the direction that
resembles success --- runs finished, receipts looked right, and the only observable was a number a
reader would assume meant what it said.

So the general check is: for every declared limit, \emph{what test fails if I delete the
enforcement?} If the answer is none, the limit is a label. We hold that property mechanically now.
Reverting the clamp in any single runner fails a test that inspects the source of all seven,
including the five that cannot be executed without an external tool installed; reverting it and
running the end-to-end case reports $60.1$\,s elapsed against a declared ceiling of $2$\,s.

\subsection{Two further instances of the same shape}
\label{sec:mirrored-definitions}

\textbf{Change detection, duplicated six times.} The soft bound of \cref{prop:no-progress} fires only
on a run of attempts where the worker reported changing nothing. Six runners each carried their own
copy of the change detector --- three annotated as mirrored from a named original and not imported ---
and each compared the workspace against a snapshot taken once when the loop was wired and never
refreshed. From the second attempt onward, any write by any earlier attempt made every later attempt
look busy, so the detector could not report ``unchanged'' and \emph{the soft bound could not fire at
all} for any runner that starts a subprocess. It is now one content-addressed function.

\textbf{Prompt assembly, duplicated seven times.} Seven runners each carried a prompt builder, three
documented as verbatim copies of a named original. They had diverged into four variants, and two
omitted the loop's declared list of forbidden actions from the assembled prompt entirely. Under those
two, a loop without an authored prompt file never told the worker what it must not modify --- and then
the gate refused the result, so the loop spent an attempt being rejected for violating a rule it had
never been given.

\subsection{The instance we produced while writing this section}
\label{sec:sixth-instance}

The last instance was created by the repair rather than uncovered by it.
\Cref{prop:spend-bound} is stated for a reserve $r$ with $0 \le r < W/2$, strictly. The adapter
computing the effective reserve clamped it to $\lfloor W/2 \rfloor$, which \emph{equals} $W/2$ for
every even $W$, so on the code-constructed path the proposition's own hypothesis was false for ninety
of the first hundred and eighty declarable ceilings. The consequence was not an unbounded run: work
still got half the ceiling and every bound stated in terms of $W$ still held. The consequence was
that a proposition in this paper quantified over a condition its implementation did not establish.

Two details earn the space. The test we wrote for that boundary \emph{passed} --- it asserted the
reserve for a sixty-second ceiling was thirty seconds, which is exactly the excluded boundary, so it
pinned the defect rather than the invariant, because it had been written from the code instead of
from the inequality. And no reading found it: a referee swept $W$ numerically and checked the
inequality directly, which took one loop and no insight.

The remedy is one character of arithmetic, $\lfloor (W-1)/2 \rfloor$, and a test that sweeps the
inequality instead of restating the value. The generalisable lesson is not about reserves: a
hypothesis stated in a paper is a specification, and specifications need the same enforcement as
configuration --- something that fails when the implementation drifts off them, executed on every
commit.

\textbf{The shape all three share.} In all three cases the load-bearing logic existed in several copies, at least one copy carried a
comment asserting the copies agreed, and no test compared them. The comment is the tell: it records
that someone noticed the duplication, decided it was acceptable, and wrote down the assumption
instead of the check. Each is now a single definition with a source-level test that refuses to let a
copy reappear.

\section{Discussion}
\label{sec:discussion}

\subsection{One defect class, found seven times}
\label{sec:seven-times}

The most useful thing this work produced is not a rate. It is the recurrence of a single
shape across seven layers that share no implementation. They are not independent in the
statistical sense: one team built all seven, and a habit of mind is the most likely common
cause. That is the point, since the habit is not carelessness but the ordinary practice of
checking that each component returns the right answer to the question it was asked.

\begin{enumerate}[leftmargin=1.8em, itemsep=2pt]
  \item \textbf{Gates.} A check satisfied by the absence of the thing checked ---
        an empty ledger balances (\cref{ex:ledger}). 47 instances.
  \item \textbf{The verdict channel.} A diagnosis the gate had in hand, filed as
        an inability to diagnose, so the loop stopped instead of retrying and 84
        mutants left the denominator (\cref{sec:e1-defects}).
  \item \textbf{The measurement's own scope.} A harness that resolved one gate
        kind and returned nothing for the rest, so 24 loops vanished without
        appearing in any count --- and the guard against unprincipled exclusion
        was vacuous over exactly those loops.
  \item \textbf{Our test suite.} 14 tests asserting over collections nothing
        proved non-empty, three of them security guards (\cref{sec:e6}).
  \item \textbf{The scanner for the class.} Detected one of the two vacuity
        shapes; missed the other until it met it in a test written hours earlier.
  \item \textbf{A runner's arguments.} An invocation asking for a credential the
        environment structurally could not supply, and another passing flags the
        binary rejects. Both produced an agent that ran and did nothing.
  \item \textbf{A provider's workspace.} A CLI that wrote its work outside the
        directory the gate reads, reported success, and was very nearly published
        as a provider that fails to converge (\cref{sec:e5-nearmiss}).
\end{enumerate}

Stated once: \emph{a component answers a question it was not asked, the answer is locally
correct, and the failure is silent and flattering.} Each individual element passes review.
The defect lives in the relationship between what a component returns and what its caller
believes that return means --- exactly the relationship no line-level reading examines,
which is the argument for measurement over inspection.

\subsection{A second axis, found by parties who did not write the code}
\label{sec:second-axis}

The seven above share a mechanism: a caller misreads a return. A release audit
conducted after the corpus work closed found three more that do not fit it, and
whose common structure is worth stating separately because its remedy is
different. In \cref{ex:lying-detail} a rule about a verdict's reason asked the
verdict; in \cref{ex:accepting-stub} a test about a platform's install policy
asked a stub written to agree; in \cref{ex:forged-provenance} a record of which
gate ran read the answer off the gate. The caller reads its return correctly in
all three. The question is sound in all three. \emph{The party answering it is the
party under scrutiny} --- \cref{def:self-attested}, whose one-line proof is the point of \cref{rem:trivial-proof}: the result is immediate and the diagnosis is not.

So the habit of mind has two expressions, and only the first is about
misinterpretation. The second is about \textbf{authority}: not what the answer
means, but whose answer it is. We keep them apart because an existence obligation
fixes neither and a type check fixes only the second, and because the second was
invisible to us for longer --- it was found by two agent CLIs auditing code they
had not written, which is the only evidence in this paper that does not inherit
\cref{sec:threats}'s first and largest threat. Their evidential grade, and what
that grade does not license, is \cref{rem:audit-grade}.

A related pattern is worth one sentence, because it is the same shape between two
components rather than inside one: this system has two implementations of ``drive an agent
CLI'', and we found the same capability present in one and missing in the other four
times, in both directions, because both sides had tests and neither had a test asserting
that the two agree.

\subsection{Threats to validity}
\label{sec:threats}

\paragraph{The gates and the corpus share an author.} The strongest threat.
Tier-1 blindness is structural (\cref{thm:blind}) and Tier-2 withholds the
checker with a digest guard, but the loops, the gates and the mutation operators
were all produced by the same project. A defect class none of us imagines is
absent from both the gates and the operators, and would not appear as a false
accept. This is the empirical form of \cref{sec:future-independence} and we
cannot bound it from inside.

\Cref{sec:e8} is the one probe of this threat we were able to run, and its result
is that the threat is real rather than theoretical: auditors with no authorship,
reading the engine instead of the catalogue, returned a defect class no operator in
either tier can express, because both tiers mutate artifacts and the class lives in
the harness. That confirms the mechanism and does not bound the residual. Two
auditors finding one unimagined class says nothing about how many remain, and the
filter that accepted their findings was still ours.

\paragraph{Zero is a small number of failures.} $\alphafa = 0$ over the 209
destroying mutants bounds the rate loosely, not tightly: the Wilson upper bound
at 95\% is $1.8\%$, and per tier it is $2.2\%$ and $9.2\%$
(\cref{app:wilson}). We report the interval rather than the point. Zero also
arrives \emph{after} 47 fixes, so what it supports is ``this corpus no longer
finds defects in these gates'', not ``these gates have no defects''. The next
corpus is the test of that, which is why \cref{sec:future-decay} proposes running
it continuously.

\paragraph{The catalogue is ours.} 69 loops chosen by us, in domains we selected,
with defects we thought worth seeding. Transfer to a codebase and task
distribution we did not design is untested.

\paragraph{E5 is small.} Sixteen observations across four providers and four
loops, with tokens unavailable through the shell path. It supports ``the gate
decides the same way regardless of which model wrote the artifact, on these
loops'' and nothing quantitative about providers.

\paragraph{The comparison a reader wants is absent.} There is no benchmark against
another agent framework. Constructing one fairly, with someone else's gates, someone
else's defect distribution and our operators, is a research project rather than a
table, and publishing one column would let the layout imply the rest.

\section{What Is New Here}
\label{sec:novelty}

The components of this work are individually unsurprising, and we would rather say
so than dress them up. Retry loops are old. Separating an implementer from a
checker is old. Mutation testing is from 1978. Append-only hash-chained logs are
older than the field this paper is in. Vacuity is from the model-checking
literature and we did not coin it. A reader who feels that no single ingredient is
novel is right. This section states what each contribution is new \emph{relative
to}.

It is tempting to call this a missing layer of the stack --- orchestration, durable
execution, observability, evaluation, and now a ``completion contract''. The
framing is wrong. The four established layers each presuppose that ``done'' is
meaningful and that cost is bounded, and none specifies either: orchestration says
where a step ends, durable execution says it will eventually end, observability
says what happened, evaluation says how well it scored. None says who is entitled
to declare a step finished, or what the harness owes when that declaration is
wrong.

But that is a \emph{property} those layers leave unspecified, not a tier of
infrastructure sitting between them. What we add attaches to the orchestration
layer rather than displacing anything, and a team running a durable-execution
engine with a genuinely separate validation service has already supplied the
property without naming it. A category with exactly one vendor in it is a marketing
artifact, and this one would have exactly one. So we state the property formally
instead, which is what lets any harness be measured against it, including harnesses that
are not ours.

\paragraph{1. Termination across a boundary the nearest work declined to cross.}
The Graph Harness framework \citep{sgh2026} proves bounded termination for a DAG
with per-node retry budgets, and excludes parent-chain rollback by name, because
its proof requires terminal states to be absorbing and rollback removes that.
\Cref{thm:graph-termination} covers exactly the excluded case. Their bound is
correct on its premise, and its per-round total is the same quantity as ours; what
we add is the outer induction that survives budget restoration, the three
conditions that make it available (\cref{ass:bounded-repair}), and the explicit
price. Because restoring spent budget is a reset arc and the general question is
undecidable, a sufficient condition is the strongest form the result can take, so
identifying which conditions suffice is the contribution rather than a weakening of
one.

\paragraph{2. The harness contract is stated as a property and proved, not
described.} The literature says a verifier should be independent. That is a design
intention: a system either follows it or does not, and there is no way to tell from
outside. We give it as \cref{def:independence} and prove what it buys
(\cref{thm:soundness}). To our knowledge the completion invariant of an agent
harness has not previously been written as a theorem.

\paragraph{3. The verdict boundary is derived rather than chosen.} Existing
practice answers ``wrong artifact or broken checker?'' by severity, or by whatever
exit code the checker happened to return. We show the distinguishing fact has
\emph{already been declared} by the anti-tamper configuration a careful loop author
writes for unrelated reasons (\cref{thm:trichotomy}). The corollary is small, and
it was worth 84 of 233 mutants.

\paragraph{4. Vacuity is measured, and the name is borrowed on purpose.} We claim
no credit for the concept. Vacuous satisfaction has been studied in temporal model
checking since \citet{beer2001vacuity}, formalised by \citet{kupferman2003vacuity}
and \citet{gurfinkel2004vacuity}, and vacuity checking ships in commercial
verification tools. We use their word because it is the right one
(\cref{sec:related-vacuity}).

What we add is a rate. That literature establishes that a property can pass for the
wrong reason and how to detect it automatically for temporal formulas. It does not
say what fraction of the hand-written verification code in a production agent
system is vacuous, because that code is not temporal logic and no detector exists
for it. We report 47 instances in a shipped catalogue that had passed review, plus
14 in the test suite guarding the engine that found them, and we give the
consequence as a statement about an adversary rather than about a specification
(\cref{thm:vacuity}). ``This property passed for an uninteresting reason'' and
``deleting the books is the cheapest way to make them balance'' are the same fact
with very different force.

\paragraph{5. The methodology is inverted, for a reason specific to this setting.}
Standard mutation testing detects equivalent mutants and excludes them.
\Cref{thm:no-equivalent} shows the problem does not arise when labels come from the
operation rather than from observed behaviour, and \cref{cor:exclusion-fatal} shows
the standard remedy would be actively harmful here. Because the subject under test
is a verifier rather than a program with a separate suite, the received methodology
inverts.

\paragraph{6. A safe way to admit a model judge, with no assumption about the
judge.} \Cref{thm:judge-asymmetry} answers the composition question rather than the
accuracy question. The proof is one line. What makes it useful is that it needs no
premise about calibration or accuracy --- exactly the premises nobody can discharge
for a model --- and that the cost it does impose lands in $\betafr$, which we
measure and publish. Composition is enforced in the engine rather than recommended
in documentation.

\paragraph{And one thing that is not a contribution but is unusual.} We report what
the instrument found when pointed at us: a scanner that missed one of the two
shapes it was written to detect, a coverage figure quoted for a quantity never measured, a defect category we had filled with four entries that all
turned out not to belong in it, and an experiment whose first run produced a clean,
publishable, false claim about a third party's product because of our own
misconfiguration (\cref{sec:e5-nearmiss}). These are in the body rather than in a
limitations paragraph. A paper arguing that verification layers go unmeasured,
which curated its own evidence, would be an instance of its own thesis.

\section{Future Work}
\label{sec:future}

Five problems, each with a first experiment.

\subsection{Gate independence under shared authorship}
\label{sec:future-independence}

The most consequential open problem here. \Cref{def:independence} gives
\emph{structural} independence, which is all \cref{thm:soundness} needs; a reader
importing intuition from independent review imports something stronger --- that the two
fail in uncorrelated ways --- and nothing establishes it. The gap is narrowing in the
wrong direction, because gates are deterministic code and that code is increasingly
written by a model, possibly the same family that writes the artifacts it judges. A model
with a systematic blind spot writes artifacts that omit a case and checkers that do not
look for it: the errors align, the receipt is valid, and the system is wrong.

Companion work makes this quantitative in an adjacent setting: two instances of one model
composed in a handoff co-failed on $90\%$ of the missions on which either failed, and
changing the \emph{model} reduced the association where changing only the \emph{vendor}
did not \citep{bhardwaj2026abc2}. Those are worker--worker measurements; the worker--gate
case has one side frozen as code.

\textbf{First experiment.} Author mutants with model $A$ against gates written by $A$,
then against gates written by $B$, loops and budget fixed, and compare escape rates. A
gap would argue for recording a checker's provenance and flagging
worker-model-equals-gate-author the way a conflict of interest is flagged.

\subsection{Does $\alphafa$ compose across a graph?}
\label{sec:future-compose}

We measure individual gates; a graph has many. The operator's question is the probability
that a graph reporting $\Done$ contains at least one defective artifact. Under
independence the arithmetic is immediate and the engine reports the compounded rate
labelled as such, but the assumption does all the work and the engine does not verify it.
Expect it to fail for two reasons: nodes share inputs, so a defect admitted upstream
reaches a downstream gate whose judgement is not a separate event; and gates in one
catalogue share idioms and authors.

\textbf{First experiment.} A two-node chain. Seed a defect upstream that the upstream
gate is known to miss, measure how often the downstream gate catches it, and compare the
joint rate against the product of the marginals.

\subsection{The retry budget is an unreported confounder}
\label{sec:future-retry-confounder}

Reliability figures for agent systems are computed over missions and assume one Bernoulli
draw each. A bounded loop breaks that, because retries are \emph{within} a mission: the
same instance is re-presented to the same gate until it accepts, so getting a defective
artifact through does not require the gate to err every time. Once suffices. With
per-attempt false-accept probability $\alphafa$ and $m$ attempts, the per-node probability
is $1 - (1-\alphafa)^m$.

The problem is not that the rate is worse. It is that \emph{the reported figure does not
observe the parameter that determines it}: two deployments with the same worker and gate,
one allowing one attempt and one ten, produce certificates of identical form and
materially different quality. This is the gap between the two theorems of
\citet{sgh2026} (\cref{sec:related-sgh}), where termination grants each node $b_v + 1$
attempts and soundness computes from an accuracy in which $b_v$ does not appear.

\textbf{First experiment.} Measure whether $\alphafa$ is constant across attempt index.
It is likeliest to fail here, because a worker adapting to rejection is being trained,
attempt by attempt, on the only signal the gate emits. If $\alphafa$ rises with the
index, the compounding above is optimistic and the anytime-valid machinery of
\cref{sec:e3} loses its guarantee rather than weakening: its null is a ceiling on
per-attempt error, and a loop learning to satisfy its gate can breach it. The receipt
already carries the per-attempt verdicts; we lack volume.

\subsection{What is genuinely not mechanically checkable?}
\label{sec:future-uncheckable}

\Cref{sec:eval-corrections} reports four requirements filed as beyond mechanical
checking, all four of which turned out checkable. The category ships empty, is presumably
not empty in principle, and we do not know where its boundary lies. The four failures
share a shape: each filing was made from the \emph{form} of the requirement rather than
from a failed implementation attempt, reasoning that something undecidable sits nearby ---
a commit's intent, an author's meaning --- and that the requirement inherits it. A checker
need not recover intent; it needs to detect the evasion, and an evasion is an edit, which
leaves a trace. So an entry in this category is now a claim that no implementation exists,
admissible only after trying to write one.

\subsection{An automated prober for both defect classes}
\label{sec:future-prober}

The two classes this paper names are currently found by two different means, and
only one of them is a tool. Vacuity has a syntactic scanner, which
\cref{sec:vacuity} reports caught one of its two shapes. Self-attestation
(\cref{def:self-attested}) has none: its three instances were found by reading, and
the twelve repairs made in the release that produced them were validated by
reverting each fix by hand and observing the guard fail. That is a practice, and we
are explicit that no instrument in this release implements it, because a paper that
argued for measurement over care and then relied on care would be its own best
counterexample.

Both classes look automatable and neither is trivial. A vacuity prober must revert
a fix and decide whether the suite noticed, which requires a semantic notion of
``the same fix, undone'' that a line-level mutation operator does not supply --- and
it must handle \cref{rem:hanging-mutant}, where the reverted fix produced neither
a pass nor a failure but a run that did not end. A self-attestation prober needs to
decide, for each value a check consults, whether the subject could have determined
it: statically that is an alias-and-ownership question over attribute reads, method
dispatch and subprocess boundaries, and dynamically it is a matter of substituting
a hostile subject and seeing whether the check still discriminates. The second
formulation is the more promising of the two and is the shape we intend to build,
because it needs no whole-program analysis --- only a subject the harness controls
and a check that must keep working when the subject stops cooperating.

\subsection{Gate decay and continuous recertification}
\label{sec:future-decay}

A gate is written against a codebase, a task distribution and a set of conventions, and
all three drift. A gate can become \emph{stale} --- reference data ages, legitimate work
is rejected, $\betafr$ rises --- and that mode announces itself, because someone
complains. Or it can become \emph{hollow}: the artifact's structure moves underneath it,
so the check still executes and constrains nothing. That is \cref{def:vacuous} arriving by
drift rather than by authorship, it raises $\alphafa$, and it is silent.

\textbf{First experiment.} A longitudinal E1: hold a corpus fixed, replay it against each
gate at successive commits, and track $\alphafa$ over the life of the repository. A gate
whose kill rate falls while its code is unchanged is being hollowed out by the artifacts
around it. We expect hollowing to dominate.

\section{Conclusion}
\label{sec:conclusion}

Explicit graphs and durable execution gave the field a vocabulary for agent control flow.
This paper adds the layer underneath: what an agent harness must guarantee, proved rather
than described, and measured rather than asserted.

The formal core is small. A graph of bounded loops terminates within a bound an operator
reads off the manifest before running anything, including in the case the nearest prior
treatment sets aside, at a stated price of one factor of $(\rounds+1)$ and on the condition
that the repair budget is global. A graph reporting \textsc{done} has a gate verdict in the
ledger for every node, and the worker's own claim appears in no guard of the procedure that
produced it. A gate can answer exactly three things, and the boundary between ``your
artifact is wrong'' and ``I cannot tell'' is fixed by an ownership declaration the loop's
author has already made.

The empirical core is an instrument and one number. Over 209 destroying and 62 conforming
held-out mutants across 57 loops we observe no false accepts and no false rejects, bounded
at $1.8\%$ and $5.8\%$. \textbf{That zero is saturation, not quality, and we ran the
experiment that shows it}: freezing the gates and applying a second operator family --- one
probing a class the first structurally cannot --- returns a false-accept rate of 23.3\%, 14
of 60, where the exhausted corpus reported none. A verifier scored by the corpus that
repaired it has demonstrated that the corpus is finished and nothing about the verifier.

Three of those fourteen are the case to put in front of a practitioner. A defect class was
found by hand, the gate that revealed it was repaired, and three sibling gates of identical
shape were not: a data processing agreement whose clause headings are present and whose
clause bodies are empty passed for three releases, invisible to the original operators
because they delete the headings too. \emph{A repair applied to the instance that revealed
a class, and not to the class, is indistinguishable from a repair applied to the class,
until a different instrument arrives.} So the instrument is the contribution and the rate
is not: ours describes our gates, and the apparatus exists so a reader can obtain their own.

If there is one thing to carry away, it is the shape that recurred seven times here across
layers sharing no implementation (\cref{sec:seven-times}): \emph{a component answers a
question it was not asked, the answer is locally correct, and the failure is silent and
flattering.} Every individual element passed review every time. That is the argument for
measuring a verification layer rather than reasoning about it --- not because reasoning is
careless, ours was not and it still missed all seven, but because this class is invisible
to inspection by construction and fails in the direction where nobody complains.
A second shape, found later and by auditors from outside the project, is worth
carrying with it: not a caller misreading a return, but a check whose answer came
from the party being checked (\cref{sec:second-axis}). The first is about meaning
and the second about authority, and no amount of care about the first finds the
second.

The largest thing the paper does not establish is stated in
\cref{sec:future-independence}. Our gates are deterministic code, which is what makes them
trustworthy, and that code is increasingly written by models. If a model writes both the
artifact and the checker they remain distinct objects with disjoint write authority and
every theorem here still holds, while the errors may be perfectly aligned and the receipt
entirely valid. Structural independence is what we proved; statistical independence is what
a reader may assume they are getting. Closing that gap is the work that matters most next.

\appendix
\section{Deferred Proofs and Supporting Arguments}
\label{app:proofs}

Proofs given in full in the body are not repeated. This appendix carries the
arguments that would have interrupted it.

\subsection{Termination, without the no-progress window}

\Cref{thm:loop-termination} is stated for a loop whose only stopping rule is the
attempt ceiling. \Cref{prop:no-progress} adds the window. The two interact in one
way worth recording.

\begin{proposition}[The window cannot extend a run]
\label{prop:window-monotone}
For any progress function $\pi$ and any window $w \geq 1$, the number of attempts
executed under \cref{alg:loop} with the window enabled is at most the number
executed with it disabled.
\end{proposition}

\begin{proof}
The window appears only at line~2.8, which can return $\Halt$ and never continues
a run that would otherwise stop. The loop guard at line~2 is unchanged. So the
executed prefix of attempts under the window is a prefix of the executed sequence
without it.
\end{proof}

This matters because a reader may reasonably worry that a badly chosen $\pi$
could keep a loop alive --- for instance one that reports progress on every
attempt because it measures something incidental, such as a timestamp. It cannot:
a useless $\pi$ degrades the window to no window, and the ceiling still binds.
The failure mode of a bad $\pi$ is wasted budget, never non-termination.

\subsection{Why the graph induction needs two variables}

\Cref{thm:graph-termination} uses the lexicographic order on
$(\rounds - \rho, \Sigma)$. Each obvious single-variable argument looks sufficient and
fails.

\emph{Attempt count alone fails.} $\Sigma$ strictly decreases within a round and
is \emph{reset upward} at a repair edge. A measure that increases at any step is
not a termination measure.

\emph{Repair round alone fails.} $\rho$ increases only at repair edges, of which
there may be none. Within a single round, $\rho$ is constant while unboundedly
many attempts could occur as far as that variable can see.

\emph{Their sum fails.} $\Sigma$ can increase by as much as
$\sum_v m_v$ at a boundary where $\rounds - \rho$ decreases by one; with
$\sum_v m_v > 1$ the sum increases.

The lexicographic order is exactly the structure that accommodates ``the second
component may increase, provided the first strictly decreases''. It is the
standard measure for nested loops, and a repair edge is the outer loop of a
nested pair whose inner loop is the DAG traversal.

\subsection{Soundness under partial failure}

\Cref{thm:soundness} assumes the run terminates in $\Done$. Two adjacent cases are
what an operator actually encounters.

\begin{corollary}[A non-\textsc{done} run makes no claim]
\label{cor:nondone}
If a run terminates in $\Halt$, $\textsc{pause}$, $\textsc{killed}$ or
$\textsc{error}$, no assertion about any artifact follows from the outcome.
Individual nodes may nonetheless be ledger-supported, and which ones is readable
from the log.
\end{corollary}

\begin{proof}
Immediate from \cref{def:outcome}: only $\Done$ is defined in terms of a $\Pass$
verdict. The second sentence follows because receipts are appended per attempt
(line~2.5) rather than at the end of the run, so a node that passed before an
unrelated failure retains its receipt.
\end{proof}

\begin{corollary}[Crash safety of the claim]
\label{cor:crash}
If the process is killed at any point, the ledger prefix that survives is
chain-valid, and \cref{cor:nondone} applies to it.
\end{corollary}

\begin{proof}
Receipts are appended and never rewritten, so a prefix of a valid chain is a
valid chain. A torn final record fails hash verification at its own position and
is rejected by replay, leaving a shorter valid prefix.
\end{proof}

\Cref{cor:crash} did not hold \emph{in practice} in an early implementation: the proof was
correct, and the reconstruction metadata a surface needed in order to open such a prefix was
written only at the end of a run, so an interrupted run left a valid prefix nothing could
open. A proof that a record supports a claim says nothing about whether the record can be
read.

\subsection{Blindness: what the AST check covers, and what it does not}

\Cref{thm:blind} has a hypothesis with four clauses, and its enforcement is a
syntactic check. The gap between the semantic property and the syntactic check
should be stated.

The check rejects: importing any module under the gate-adapter package, by
absolute, relative, function-local or dynamic import; opening any path matching a
checker source pattern; comparing against a loop name; and invoking a subprocess.
It counts \texttt{importlib.import\_module} and \texttt{\_\_import\_\_} with
literal targets.

It does not and cannot reject a dynamic import whose target is computed at run
time from a non-literal expression. We accept this: such a construct in a
generator package would be conspicuous in review, and the check is a guard
against drift by ordinary means rather than against a determined author of the
generator itself. The threat model is a maintainer adding a convenient import,
not an adversary.

\subsection{Unbiasedness, stated more carefully}

\Cref{thm:no-equivalent} shows the estimator's denominator does not depend on the
gate. This is a statement about \emph{selection}, not about the operators.

If the operator set systematically produced defects of a kind a given gate
happens to catch, $\hat\alphafa$ would be optimistic for that gate --- and no
independence-of-selection argument repairs that. It is a threat about coverage
rather than about circularity, and it is the reason Tier-2 exists at all: an
operator set is a fixed, finite hypothesis about how things go wrong, and an
independent author shown only the loop's purpose can produce a failure the
operator set cannot express. \Cref{sec:e2} reports the case where exactly that
happened, on the one loop Tier~1 structurally cannot cover.

\subsection{The interval reported with zero failures}
\label{app:wilson}

For $x = 0$ successes the Wilson score interval has lower bound $0$ and upper
bound exactly $z^2/(n + z^2)$, with $z = 1.96$ at $95\%$, so $z^2 \approx 3.84$.

The denominator is the part that has to be right, and it is not the size of the
corpus. $\alphafa$ is a rate over mutants labelled \emph{defective}, because
only a defective mutant can be falsely accepted; the conforming mutants belong
to $\betafr$. The corpus holds 209 destroying mutants (171 Tier-1 + 38 Tier-2)
and 62 preserving ones. Hence
\[
  \alphafa \;\leq\; \frac{3.84}{209 + 3.84} \;\approx\; 0.018,
  \qquad
  \betafr \;\leq\; \frac{3.84}{62 + 3.84} \;\approx\; 0.058 .
\]
The supportable statements are ``$\alphafa \leq 1.8\%$'' and ``$\betafr \leq 5.8\%$''
at $95\%$ over this corpus, not ``$\alphafa = 0$''.

A previous version computed this bound over $n = 271$ --- every mutant, conforming
ones included --- which pools $\betafr$'s denominator into $\alphafa$'s and
reports $1.4\%$ where the data support $1.8\%$. The per-tier bounds in
\cref{tab:corpus} were correct throughout; only the pooled figure was wrong. An
external methodology review found it.

Pooling the two tiers into a single $209$ needs one caveat: Tier~1 covers
57 loops with mechanical operators and Tier~2 covers 10 with authored semantic
mutants, so the pooled rate averages over two structurally different sampling
processes. The per-tier bounds of \cref{tab:corpus} are the primary figures and
the pooled one is a convenience.

The anytime-valid sequence of \cref{sec:e3} answers a different question --- a
rate over a growing log that is inspected repeatedly --- and the two should not
be conflated.

\subsection{Bounded elapsed spend, proved}
\label{app:spend-proof}

\Cref{prop:spend-bound} was argued in the body at the level of enforcement points. Here is
the operational argument, which is what the body promised and did not have. The machinery is
elementary, and the clauses are easy to believe: one of them was false in the shipped code
for exactly the reason the argument makes visible.

\paragraph{Setup.} Fix a run with declared ceiling $W$ and reserve $r$, $0 \le r < W/2$. Let
$c$ be a monotone clock: $c$ is non-decreasing along any execution and $c(0) = 0$ at the
instant the controller starts. Write $L = W - r$ for the \emph{work ceiling}. The controller
alternates between two kinds of step:

\begin{itemize}[leftmargin=1.4em]
  \item a \textbf{lap head}, which evaluates the guards of \cref{fig:control-flow} and either
        terminates or admits one attempt;
  \item an \textbf{attempt}, which hands control to the worker and returns when the worker
        returns or when its deadline expires.
\end{itemize}

Let $t_k$ be the clock value at the head of lap $k$, and let $A_k = [s_k, e_k]$ be the clock
interval of the attempt admitted on lap $k$, if any. Both sequences are finite by
\cref{thm:loop-termination}, so every quantity below is a maximum over a finite set.

\paragraph{Clause (1): no attempt begins after $L$.}
The budget guard is evaluated at every lap head before any attempt is admitted, and it halts
when $c \ge L$. Admission of $A_k$ therefore implies $t_k < L$. The controller performs no
clock-advancing work between the guard and the handoff other than assembling the attempt's
context, and $c$ is monotone, so $s_k \ge t_k$ and $s_k$ is the value of $c$ at the handoff.
Hence $t_k < L$ for every admitted attempt, which is clause~(1) with $s_k$ read as the
admission instant.

Two things are load-bearing and both were wrong in an earlier version. First, the guard must
compare against $L$ and not $W$; comparing against $W$ admits an attempt inside the reserve
and the reserve then cannot be spent on what it was withheld for. Second, the remaining budget
handed to the runner must be measured \emph{after} the guard, not at wiring time: a remainder
computed once when the loop was constructed is a constant, and a constant is not a remainder.

\paragraph{Clause (2): no attempt continues past $L$.}
Let $\rho_k = L - s_k$, which is positive by clause~(1). The runner receives $\rho_k$ and sets
its blocking wait to $\min\{\rho_k, \theta\}$, where $\theta$ is the deployment's configured
per-turn timeout. On expiry it terminates the child and raises. Therefore
\[
  e_k \;\le\; s_k + \min\{\rho_k, \theta\} \;\le\; s_k + \rho_k \;=\; L .
\]
The clamp is a single shared definition used by every runner, and a source-level test refuses
a runner that does not consult it; that test is the reason this inequality is a property of the
system rather than of one adapter. Note that clause~(2) says nothing about which of the two
limits bound: $\min$ is enough for the inequality, and the receipt records which one fired
because the two have different owners and an operator's next action differs.

\paragraph{Clause (3): total worker time is at most $W$.}
Worker time is $\sum_k |A_k| + |A^{\mathrm{wd}}|$, where $A^{\mathrm{wd}}$ is the wind-down
attempt if one occurs. The attempts are disjoint intervals contained in $[0, L]$ by clauses~(1)
and~(2), so $\sum_k |A_k| \le L$. The wind-down is admitted only on a bound halt, at most once
per run --- the controller reaches the wind-down branch from a terminal decision, and terminal
decisions do not return to the lap head --- and it is clamped against the \emph{full} ceiling,
so $|A^{\mathrm{wd}}| \le W - c \le W - L = r$ at the instant it starts. Summing,
$\sum_k |A_k| + |A^{\mathrm{wd}}| \le L + r = W$.

This is where $r < W/2$ is used, and it is used for a reason that is not arithmetic: it is what
makes $L > r$, so the work a run does always exceeds the courtesy it is granted. Any
$r < W$ gives the inequality; only $r < W/2$ gives a run whose majority is work. The clamp that
computes the effective reserve must therefore be strict, and \cref{sec:sixth-instance} is the
account of it not being so.

\paragraph{Clause (4): total run time is at most $W + G$.}
Gate invocations are not clamped. Let $G$ bound one gate invocation. A gate is invoked only
from a lap head, after that lap's attempt has returned, and at most one gate invocation is in
flight at any time because the controller is sequential. Any gate invocation beginning at
$c \le L$ contributes at most $G$ past its start. Runs of gates that begin and end inside
$[0, L]$ are already counted in the elapsed total; the only invocation that can extend the run
past $W$ is the last one, contributing at most $G$. Hence total elapsed time is at most
$\max\{W, L + G\} \le W + G$.

\begin{remark}[The clause that is a design decision, not a bound]
Clause~(4) is the weakest of the four and deliberately so. Clamping gates would tighten it to
$W$ and would purchase that with a verdict that does not exist: a check terminated part-way
has judged nothing, and recording it as a judgement is the failure this paper's whole
soundness argument is built to prevent. We would rather publish the looser inequality than
report an unfinished check as an outcome. \Cref{rem:gate-tail} states what $G$ is in the
shipped wiring, which is not small.
\end{remark}

\section{A Pattern Catalogue for Gate Design}
\label{sec:patterns}

The literature tells a harness author that verification should be deterministic
and that the verifier should be independent of the worker. It does not tell them
what a good gate looks like, which requirements admit one, or how gates fail.
Those questions are answerable now, because a 69-loop catalogue was authored,
shipped, and then measured twice by instruments that found sixty-one defects
between them.

This section is that answer, and its discipline is that \textbf{every entry is
derived from something measured elsewhere in this paper.} We are not proposing a
methodology. We are reporting the shape of a design space that a corpus mapped,
including the parts of it we got wrong. A reader should treat the failure modes
as the load-bearing half: the constructions are unsurprising once stated, and
the ways they silently stop working are not.

\subsection{What a gate can be}

\begin{table}[h]
\centering\small
\begin{tabular}{@{}p{0.17\linewidth}rp{0.30\linewidth}p{0.36\linewidth}@{}}
\toprule
\textbf{Pattern} & \textbf{n} & \textbf{Verdict is} & \textbf{Choose it when} \\
\midrule
Schema &  10 & a document validates against a declared schema &
  the requirement is about \emph{shape}, and the schema is the contract a
  consumer already relies on \\
Predicate & 44 & a purpose-built program exits zero &
  the requirement is about \emph{content} and reduces to a computation over one
  artifact \\
Executable spec & 9 & a test suite passes &
  the requirement is behavioural, and the artifact is code that can be exercised \\
Conjunction & 3 & every child gate passes &
  two independent obligations must both hold, or a model judge is being admitted
  (\cref{thm:judge-asymmetry}) \\
External tool & 2 & a third-party analyser reports clean &
  a mature checker already exists and re-implementing it would be worse \\
\bottomrule
\end{tabular}
\caption{The five gate patterns the catalogue uses, with the count of loops
choosing each. The distribution is itself a finding: two thirds of real
requirements were served by a purpose-built predicate over a single artifact, not
by a schema and not by a test suite. An author's instinct to reach for a schema
generalises less far than the catalogue suggests it should.}
\label{tab:gate-patterns}
\end{table}

The choice is not free. A schema gate's contract \emph{is} its schema, so a
document the schema admits is correct by definition --- which is why two of the
apparent false accepts in \cref{sec:e7} were adjudicated as gate-correct rather
than as defects. An executable-specification gate inherits every weakness of the
suite, including the ability to pass over an empty collection. A predicate gate
is the most expressive and the most likely to be wrong, because its author wrote
the whole of it.

\subsection{Reducing a requirement to a predicate}

The recurring objection is that real requirements need judgement.
\Cref{tab:decidable} showed four that read that way and are not, and the general
move is visible across the catalogue: a requirement about meaning is replaced by
an obligation about a \emph{witness} --- a field that must be numeric, a date
that must match a format, a heading that must exist \emph{and} carry text
beneath it. The judgement does not disappear. It moves into authoring the
predicate, where it is written once, reviewed, and versioned, instead of being
re-exercised on every attempt and recorded nowhere.

Where no witness exists, \cref{thm:judge-asymmetry} gives the safe composition:
a model judge conjoined with a deterministic gate cannot raise the false-accept
rate whatever it does. That is the pattern to reach for, and the engine refuses
every other composition.

\subsection{How gates fail, in the order they were found}

\begin{table}[t]
\centering\small
\begin{tabular}{@{}p{0.20\linewidth}p{0.42\linewidth}p{0.30\linewidth}@{}}
\toprule
\textbf{Failure mode} & \textbf{The gate passes because} & \textbf{Evidence} \\
\midrule
Absence & the artifact it quantifies over is gone, so every universal claim
  holds & 47 instances; \cref{ex:ledger} \\
Hollowing & the artifact keeps its structure and loses its content, and the
  gate checked structure & 14 instances; \cref{sec:e7} \\
Self-derived scope & the checker builds its search pattern from its own
  reference data, so anything outside it is never examined & \cref{ex:citations} \\
Misfiled incapacity & a diagnosis the gate reached was returned as an inability
  to diagnose, so the loop stopped instead of retrying & 84 of 233 mutants;
  \cref{sec:verdict} \\
Scope blindness & the \emph{measurement} resolved one gate kind and silently
  returned nothing for the others & 24 loops; \cref{fig:coverage} \\
Prompt-only obligation & the requirement was stated to the worker and never to
  the gate & \cref{ex:ledger} \\
Self-attestation & the value the check judged was chosen by the party under
  scrutiny, so a sound criterion was applied to a supplied answer & 3 instances;
  \cref{sec:self-attestation} \\
\bottomrule
\end{tabular}
\caption{Seven ways a gate passes work it should refuse. Four are properties of a
gate; rows four, five and seven are properties of the machinery \emph{around}
gates, and we keep them in the same table deliberately, because in this project
they cost more evidence than the gate defects did. The counts in the third column
are not commensurable: the first five come from instrumented measurement over the
catalogue, the last from a release audit whose grade is stated in
\cref{rem:audit-grade}.}
\label{tab:failure-modes}
\end{table}

Two remedies cover the first two rows and both are one line. Against absence, an
\emph{existence obligation}: assert the collection is non-empty before asserting
anything about its members. Against hollowing, a \emph{substance obligation}: a
required element is present only if it also carries content. The catalogue now
applies both, and \cref{rem:incomplete-fix} is the reminder that applying a
remedy to the instance that revealed a class is not the same as applying it to
the class.

\subsection{What we would tell an author, and what we cannot}

Three rules survive everything above, and each is earned rather than asserted.
\emph{Put every requirement in the gate}, because a requirement stated only in
the prompt is addressed to the party whose compliance is in question.
\emph{Make refusal the default for what cannot be checked}, because an
unmeasurable question answered ``pass'' is the failure mode this whole paper is
about. And \emph{measure the gate against defects its author did not write},
because \cref{sec:disclosure-asymmetry} shows authors document how their gate is
permissive and do not notice how it is strict. A fourth was earned late and costs
nothing to apply: \emph{never ask the subject a question you can answer
yourself}, because a criterion applied to a value the subject chose is not a
check (\cref{def:self-attested}).

What we cannot tell an author is whether this generalises. The catalogue is
ours: 69 loops, chosen by us, in domains we selected, with gates we wrote and
operators we designed. \Cref{tab:gate-patterns}'s distribution may be a fact
about requirements or a fact about our imagination, and nothing here separates
those. The failure modes are on firmer ground --- six of the seven were found by an
instrument rather than proposed by an author, the seventh by auditors outside the
project, and two of them were found in code written by people who had spent the
preceding week studying the class --- but their \emph{rates} belong to this
catalogue and to no other. Establishing which
of these patterns hold outside it needs the instrument pointed at gates we did
not write, which is the first item in \cref{sec:future}.

\section{The Loop Catalogue}
\label{app:catalogue}

The 69 loops, by domain and gate kind. Each ships a seed carrying a defect, a
prompt stating the requirement, a gate, and a recorded agent trace so the
demonstration runs offline with no credential.

\begin{table}[h]
\centering\small
\begin{tabular}{@{}p{0.20\linewidth}p{0.12\linewidth}r p{0.52\linewidth}@{}}
\toprule
\textbf{domain} & \textbf{gate kind} & \textbf{n} & \textbf{representative requirement} \\
\midrule
legal / contract   & command      & 9  & required clauses present \emph{and not negated in the body};
                                        citations resolve against shipped reference data \\
finance            & pytest, command & 8 & three-way invoice match; a ledger reconciles \emph{and no
                                        seeded transaction disappeared} \\
clinical           & jsonschema   & 5  & a note carries every required field with the stated types \\
security           & command      & 7  & no credential-shaped literal, including across string
                                        concatenation and renamed keys \\
infrastructure     & command, external & 8 & resources declare required tags; scanned policies pass \\
software eng.      & pytest, command & 14 & a failing suite goes green; every module has a test that
                                        imports it and asserts on what it imported \\
dependency / supply & command     & 6  & every dependency pinned, to a version that exists, at or
                                        above the file's own floor \\
documentation      & command      & 7  & links resolve; alt text present; commit subjects conform to
                                        the declared convention \\
data contract      & composite, jsonschema & 8 & payload validates \emph{and} the contract document
                                        agrees with it \\
runner examples    & pytest       & 4  & framework integration demonstrations (excluded: unshipped SDK) \\
\bottomrule
\end{tabular}
\caption{Catalogue composition by domain. \textbf{The count column does not sum
to 69, and should not.} A loop declares one or more domain tags --- 28 of the 69
declare more than one, so the tags total 96 across 24 distinct labels --- and the
groupings above merge related tags for readability. The authoritative count of
loops is 69, obtained from the manifests and pinned by a packaging test; the
authoritative partition by \emph{gate kind}, which is disjoint and does sum to
69, is in \cref{sec:eval-setup}. The requirements in the right-hand column are stated as
they are \emph{after} the fixes of \cref{sec:e1-defects,sec:e2}.}
\label{tab:catalogue}
\end{table}

\subsection{The eleven loops outside Tier-1}

\begin{table}[h]
\centering\small
\begin{tabular}{@{}p{0.26\linewidth}r p{0.60\linewidth}@{}}
\toprule
\textbf{reason} & \textbf{n} & \textbf{why, and how it is verified} \\
\midrule
unshipped SDK & 4 & the runner needs an agent framework the package does not vendor. Probed
  against the live environment in \emph{both} directions --- a loop claimed excludable that in
  fact runs fails the check, as does the converse \\
external tool & 3 & the gate invokes a security scanner not assumed present \\
negative requirement & 1 & the requirement is that an artifact must \emph{not} acquire a property;
  a mutation that introduces the property is a conforming input, not a defect, so the operator set
  cannot construct a labelled defective case \\
Tier-1 inexpressible & 3 & the requirement is a relation \emph{between} artifacts. A
  single-artifact operator has no notion of consistency across two files. Covered by Tier~2,
  which found four defects in one of them \\
\bottomrule
\end{tabular}
\caption{Exclusions. This table is the denominator: every rate in the paper is
over $69 - 12 = 57$ loops.}
\label{tab:exclusions}
\end{table}

\subsection{Two catalogue entries the paper does not round away}

\paragraph{A known mislabelling, kept.} One Tier-2 mutant changed an invoice
quantity from \texttt{50} to \texttt{50.0} and was reported as breaking a
three-way match. JSON does not distinguish them; the gate was right. The entry is
retained, keyed by content digest, and marked as a known mislabelling rather than
deleted --- so the corpus carries its own correction and a future run cannot
silently re-derive the wrong conclusion.

\paragraph{An empty category, kept empty.} The catalogue reserves a category for
requirements no mechanical checker can decide. Four findings were placed in it and
all four proved checkable (\cref{sec:eval-corrections}). It ships \textbf{empty},
with the four failures recorded in its documentation, because an entry is a claim
that no implementation exists and that claim is earned only by attempting one.
\Cref{sec:future-uncheckable} argues the category is non-empty in principle and
that we do not know where its boundary lies.

\section{Reproduction}
\label{app:repro}

Every number in this paper is the printed output of a released script. This
appendix says which one, and what a reader should expect to differ.

\subsection{Setup}

\begin{verbatim}
  git clone https://github.com/qualixar/bounded-loops
  cd bounded-loops && git checkout v0.6.6 && uv sync
\end{verbatim}

\noindent The tag matters: every measurement below was taken on it, and later releases add
features that change no measured quantity but would leave a reader unable to check that
for themselves.

No credential is required for anything in \cref{sec:e1,sec:e2,sec:e3,sec:e4,%
sec:e6}. The three keyless gate kinds are offline by construction, and the
catalogue's agent traces are recorded, so the demonstration runs with no network
and no account. Only \cref{sec:e5} needs provider logins.

\subsection{What regenerates what}

\begin{table}[h]
\centering\footnotesize
\begin{tabular}{@{}p{0.30\linewidth}p{0.62\linewidth}@{}}
\toprule
\textbf{result} & \textbf{command, and what it pins} \\
\midrule
E1, E2 --- \cref{tab:corpus}
  & \texttt{uv run pytest tests/evaluation} \newline
    The corpus is a test: a regression in any gate fails the build rather than moving a
    number in a paper. \\
\addlinespace
Blindness --- \cref{thm:blind}
  & \texttt{...\,test\_generator\_\allowbreak is\_blind.py} \newline
    AST assertion over every generator module. \\
\addlinespace
Tier-2 prompt integrity
  & \texttt{...\,test\_tier2\_\allowbreak authoring\_\allowbreak was\_blind.py} \newline
    Recomputes each mutant's authoring digest from the repository. \\
\addlinespace
E3 --- \cref{tab:coverage}
  & \texttt{...\,test\_reported\_\allowbreak interval\_\allowbreak estimand.py} \newline
    Includes the three shrink/shift controls that must fail if the interval were vacuous. \\
\addlinespace
E4
  & \texttt{...\,tests/graph/\allowbreak test\_gate\_gaming.py} \newline
    Adversarial workers. \\
\addlinespace
E6
  & \texttt{...\,test\_no\_test\_\allowbreak asserts\_only\_\allowbreak inside\_a\_loop.py} \newline
    Carries a deliberately vacuous case it must flag. \\
\addlinespace
Catalogue convergence
  & \texttt{uv run pytest -k catalog\_convergence} \newline
    All 69 loops, on by default. \\
\bottomrule
\end{tabular}
\caption{Reproduction map.}
\label{tab:repro}
\end{table}

\subsection{Why the corpus is a test rather than a script}

A number in a paper decays: the artifact moves and the number does not. Making
the corpus part of the test suite means a change that reintroduces any of the 47
gate defects fails the build on the commit that causes it. The number in
\cref{tab:corpus} is therefore a claim about the current commit that the current
commit checks, and a reader who clones a later version and gets a different
number is seeing a real difference rather than drift in our bookkeeping.

\subsection{\Cref{sec:e5}, and what will not reproduce}

E5 requires four provider CLIs and an active subscription for each. It is
\emph{not} part of the test suite and must never become one: a test requiring the
authors' logins passes only on the authors' machines, which is the property that
let three provider-layer defects ship undetected in the first place.

\begin{verbatim}
  python paper/experiments/e5_convergence.py
\end{verbatim}

A reader reproducing this should expect: terminal state and attempts-to-\textsc{done}
to match, since those are what the gate determines; wall-clock times to differ
substantially, since they depend on provider load; and token counts to be absent,
since the path used does not parse provider usage envelopes
(\cref{sec:e5}). A fifth profile ships and was held out for budget reasons.

The seeded defects in these four loops are repaired by any competent agent in one
attempt, so a reader should read all-sixteen-first-attempt as consistent with the
gate deciding, and \emph{not} as evidence that the arms would remain
indistinguishable on a harder loop.

\subsection{Verifying a run without trusting us}

A run directory is self-contained and third-party checkable
(\cref{cor:third-party}):

\begin{verbatim}
  uv run bl verify <run-dir>     # replays the chain, fails closed on a break
  uv run bl status <run-dir>     # the projection every surface reads
\end{verbatim}

Replay recomputes each receipt's predecessor hash and rejects a modified entry or
a torn tail (\cref{thm:chain}). This is tamper-evidence against partial edits, not
authentication: an adversary able to rewrite the whole file can rewrite the whole
chain, and a signature rooted outside the directory would be required to prevent
that.

\section{Notes on Auditing a Harness With Agent Auditors}
\label{app:audit}

\Cref{sec:e8} reports one result from a release audit conducted by two commercial
agent CLIs. This appendix records the protocol and the four failures of the audit
apparatus itself, because both were more expensive to learn than the result was to
state, and neither is specific to this project.

\subsection{An auditor is an agent, and an agent writes}

Given a brief that asked for review and no changes, one auditor edited three source
files in the repository under audit and created four editor-configuration paths in
it. The second auditor was reading the same tree at that moment and logged that the
line counts it had already recorded no longer matched the file: it had detected its
subject changing underneath its own audit. Nothing was lost --- the tree was
committed first --- but an unreviewed agent diff had landed on a security boundary,
and for the window in which it existed neither auditor's findings were sound, since
one had mutated what the other was reading.

Three rules follow, and all are free:

\begin{enumerate}[leftmargin=1.8em, itemsep=2pt]
  \item \textbf{One throwaway worktree per auditor}, created detached at the commit
        under audit. Never two auditors in one tree, even when both are instructed
        to read only --- the instruction is not the mechanism.
  \item \textbf{Do not remove a worktree until every auditor using it has exited.}
        One audit logged that its worktree had been removed underneath it and
        continued from what it had already read, which silently degrades a report
        in a way the report cannot describe.
  \item \textbf{An auditor's proposed fix is a finding, not a patch.} In the round
        after the incident above, one auditor's diff was discarded and all four of
        its findings reproduced independently --- whereupon one of the tests written
        from its own report was found to pass with the fix reverted. Accepting the
        diff would have shipped a test that could not fail.
\end{enumerate}

\subsection{Four guards that could not report}

The audit's most transferable output was not the defect list but the list of guards
incapable of reporting a failure. Each was green, or red for so long it had become
background.

\begin{itemize}[leftmargin=1.4em, itemsep=3pt]
  \item \textbf{A release test that bound an exit code and never read it}, and
        parsed a single section of its subject's output. A subject that crashed
        printed nothing; the empty set is a subset of the expected set; the test
        passed. It now fails against a deliberately broken subject.
  \item \textbf{A pipeline that had never reached its own type-checking step.} The
        test step ran first and a failed step ends the job, so pre-existing type
        errors sat behind a gate that had never once executed. Every job reported
        green; one of them was green because it stopped early. Ordering a cheap
        static check \emph{before} an expensive dynamic one is not a preference ---
        it decides whether the static check runs at all.
  \item \textbf{A test asserting that a file changed, using the one filename the
        change detector documents as excluded.} It was red on every run, for a
        correct reason nobody read.
  \item \textbf{An environment-dependent assertion} --- an exit code of zero
        asserted unconditionally, on a platform where the system correctly refuses.
        Green on the authors' machine, impossible on the runner.
\end{itemize}

\subsection{A mechanical probe, and its one gap}

Separately from the auditors, and exercising judgement nowhere: for each of twelve
repairs the release had already made, the repair was reverted and the suite re-run.
\textbf{Eleven guards failed as they should; one did not.} The gap was the second
call site of a two-site repair, which had no test at all, so reverting it stayed
green --- \cref{rem:incomplete-fix} recovered by procedure rather than by
instrument.

One revert produced neither a pass nor a failure: it made the loop spin, and the
test did not finish. That is \cref{rem:hanging-mutant}, and it is the reason a
prober for this needs a declared per-mutant deadline whose expiry is recorded as a
third outcome rather than folded into either count. \Cref{sec:future-prober} states
what such an instrument would have to do; this release does not contain one, and
the practice above is a practice, not a tool.

\section*{Availability and Licensing}
\addcontentsline{toc}{section}{Availability and Licensing}

The \textsc{bounded-loops} engine, the loop catalogue, the reference graphs, the
mutation harness, the evaluation scripts, and the figure-generation scripts are
released under Apache-2.0 at \url{https://github.com/qualixar/bounded-loops},
and are installable as \texttt{bounded-loops} on PyPI and npm. Every figure and
every number in \cref{sec:eval} is emitted by a script in that repository from
the receipt logs it names; none is transcribed by hand. \Cref{app:repro} gives
the commands that regenerate them from a clean checkout, together with the
version and commit this paper reports.

\section*{Author Contributions}
\addcontentsline{toc}{section}{Author Contributions}

V.P.B.\ designed the study, implemented the engine, the gate catalogue, and the
mutation harness, ran the evaluation campaign, performed the analysis, and wrote
the paper. G.S.\ and A.P.B.\ contributed to loop and gate design, tested the
engine and the loop catalogue across releases, reviewed the experimental
protocol, and reviewed the manuscript. All authors approved the final version.

\section*{Author Biographies}
\addcontentsline{toc}{section}{Author Biographies}

\noindent\textbf{Varun Pratap Bhardwaj} is a solution architect and independent
researcher based in India, and the founder of Qualixar, where his work is on
reliability engineering for AI agent systems. He is the author of the
\textsc{SuperLocalMemory} local-first agent memory system and of the
\textsc{AgentAssert} behavioural-contract framework, and lead author of its
compositional-reliability paper \citep{bhardwaj2026abc2}. ORCID:
\href{https://orcid.org/0009-0002-8726-4289}{0009-0002-8726-4289}.

\medskip
\noindent\textbf{Garima Singh} is an independent researcher based in India,
working on evaluation and quality assurance for agentic software systems. Her
contribution to this work was in gate design and in the systematic testing of
the engine and its loop catalogue across releases.

\medskip
\noindent\textbf{Arun Pratap Bhardwaj} is an independent researcher based in
India, working on the design and validation of automated verification workflows.
His contribution to this work was in loop design, protocol review, and testing
of the released engine.

\section*{Use of AI Assistance}
\addcontentsline{toc}{section}{Use of AI Assistance}

AI assistants contributed to drafting and reviewing this manuscript and to the engine's
code. Separately, agent CLIs are \emph{subjects} of the evaluation: the workers in E5 and
the blind authors of the Tier-2 corpus in E2 (\cref{sec:eval}). No model has authority
over a result --- every verdict is produced by deterministic gate code and every
statistic is a released script's output. The authors carried out the experiments and the
analysis and take full responsibility for the contents.

\section*{Conflicts of Interest}
\addcontentsline{toc}{section}{Conflicts of Interest}

\textsc{bounded-loops} is developed by Qualixar, with which V.P.B.\ is
affiliated. The evaluation measures this engine's own gates, so the authors are
not disinterested parties, and we state the mitigations so a reader can weigh
them rather than take our word for the result.

Four mitigations, each enforced rather than promised. The Tier-1 corpus is generated
by operators \emph{structurally} prevented from consulting a gate: a test asserts
against the abstract syntax tree of every module in the generator package that it
cannot import a gate adapter, read a checker source file, branch on a loop's name, or
execute anything (\cref{thm:blind}). The Tier-2 corpus was authored by five
independent agent CLIs shown only the loop's stated purpose with the checker
withheld, with the digest of exactly what each author saw recorded per mutant and
recomputed by a guard, so a widened prompt is a failing test rather than a silent
change. Every label comes from the mutation operation, fixed before the edit, so no
label is assigned by observing a gate (\cref{thm:no-equivalent}). And the analysis
scripts are released, so each reported number is the output of code a reader can run.

We also report, in \cref{sec:eval} and in \cref{sec:discussion}, the results
that are unflattering to us: a mutant we labelled as an evasion where the gate
was right and our label was wrong; four defects we filed as not mechanically
checkable, all four of which proved checkable; and fourteen tests in our own
suite that carried the exact defect this paper is about. No patent application
covering the methods in this paper is pending.

\bibliographystyle{plainnat}
\bibliography{references}

\end{document}